\documentclass[a4paper, 11pt]{article}

\usepackage[margin=1in]{geometry}
\usepackage[numbers,sort&compress]{natbib}
\usepackage{amsmath}
\usepackage{amssymb}
\usepackage{amsthm}
\usepackage{algorithm}
\usepackage{algorithmic}
\usepackage{booktabs}
\usepackage{enumitem}
\usepackage{multicol}
\usepackage{xcolor}
\usepackage{soul}
\usepackage{graphicx}
\usepackage[small]{caption}
\usepackage{microtype}
\usepackage{hyperref}
\usepackage{doi}

\hypersetup{
    colorlinks=true,
    linkcolor=black,
    citecolor=black,
    filecolor=black,
    urlcolor=[HTML]{0088CC},
    pdftitle={Seeing far vs. seeing wide: Volume complexity of local graph problems},
    pdfauthor={Will Rosenbaum, Jukka Suomela}
}

\theoremstyle{definition}

\newtheorem{prop}{Proposition}[section]
\newtheorem{lthm}[prop]{Theorem}
\newtheorem{lem}[prop]{Lemma}
\newtheorem{cor}[prop]{Corollary}

\newtheorem{rem}[prop]{Remark}

\newtheorem{eg}[prop]{Example}
\newtheorem{dfn}[prop]{Definition}

\newtheorem{que}[prop]{Question}

\newtheorem{obs}[prop]{Observation}
\newtheorem{ntn}[prop]{Notation}

\newcommand{\dft}[1]{\textbf{\textit{#1}}}

\newcommand{\DDIST}{\text{D-DIST}}
\newcommand{\RDIST}{\text{R-DIST}}
\newcommand{\DVOL}{\text{D-VOL}}
\newcommand{\RVOL}{\text{R-VOL}}
\newcommand{\query}{\mathrm{query}}
\newcommand{\DIST}{\mathrm{DIST}}
\newcommand{\cost}{\mathrm{cost}}
\newcommand{\VOL}{\text{VOL}}

\newcommand{\Qpull}{Q_{\mathrm{pull}}}
\newcommand{\Qpush}{Q_{\mathrm{push}}}

\newcommand{\lvl}{\mathrm{level}}

\newcommand{\proc}{\mathcal{P}}

\DeclareMathOperator{\dist}{dist}
\newcommand{\ports}{\mathcal{P}}
\newcommand{\Lin}{L_{\mathrm{in}}}
\newcommand{\Lout}{L_{\mathrm{out}}}

\newcommand{\abs}[1]{\left|#1\right|}
\newcommand{\paren}[1]{\left(#1\right)}
\newcommand{\set}[1]{\left\{#1\right\}}
\newcommand{\sucht}{\,\middle|\,}

\newcommand{\N}{\mathbf{N}}

\newcommand{\calE}{\mathcal{E}}
\newcommand{\calI}{\mathcal{I}}
\newcommand{\calG}{\mathcal{G}}
\newcommand{\calL}{\mathcal{L}}  %% a labeling
\newcommand{\calN}{\mathcal{N}}
\newcommand{\calO}{\mathcal{O}}
\newcommand{\calP}{\mathcal{P}}

\newcommand{\E}{\mathrm{E}}

\newcommand{\LeafColoring}{\textsf{LeafColoring}}
\newcommand{\BTL}{\textsf{BalancedTree}}
\newcommand{\HTHC}{\textsf{Hierarchical-THC}}
\newcommand{\HBTHC}{\textsf{Hybrid-THC}}
\newcommand{\HHTHC}{\textsf{HH-THC}}
\newcommand{\RTHC}{\mathrm{RecursiveHTHC}}

\DeclareMathOperator{\disj}{disj}

\newcommand{\parent}{\mathrm{P}}
\newcommand{\lc}{\mathrm{LC}}
\newcommand{\rc}{\mathrm{RC}}
\newcommand{\chin}{\chi_{\mathrm{in}}}
\newcommand{\chout}{\chi_{\mathrm{out}}}
\renewcommand{\ln}{\mathrm{LN}}
\newcommand{\rn}{\mathrm{RN}}

\newcommand{\RWtoLeaf}{\mathrm{RWtoLeaf}}

\renewcommand{\th}{{}^{\mathrm{th}}}

\title{\bf Seeing Far vs.\ Seeing Wide:\\
Volume Complexity of Local Graph Problems}

\author{
  Will Rosenbaum\\
  Max Planck Institute for Informatics\\
  will.rosenbaum@mpi-inf.mpg.de
  \and
  Jukka Suomela\\
  Aalto University\\
  jukka.suomela@aalto.fi
}
\date{}

\begin{document}

\maketitle

\begin{abstract}
Assume we have a graph problem that is \emph{locally checkable but not locally solvable}---given a solution we can check that it is feasible by verifying all constant-radius neighborhoods, but to find a feasible solution each node needs to explore the input graph at least up to distance $\Omega(\log n)$ in order to produce its own part of the solution.

Such problems have been studied extensively in the recent years in the area of distributed computing, where the key complexity measure has been \emph{distance}: how far does a node need to see in order to produce its own part of the solution. However, if we are interested in e.g.\ sublinear-time centralized algorithms, a much more appropriate complexity measure would be \emph{volume}: how large a subgraph does a node need to see in order to produce its own part of the solution.

In this work we study locally checkable graph problems on bounded-degree graphs and we give a number of constructions that exhibit different tradeoffs between deterministic distance, randomized distance, deterministic volume, and randomized volume:
\begin{itemize}
\item If the deterministic distance is linear, it is also known that randomized distance is near-linear. We show that volume complexity is fundamentally different: there are problems with a linear deterministic volume but only logarithmic randomized volume.
\item We prove a volume hierarchy theorem for randomized complexity: Among problems with (near) linear deterministic volume complexity, there are infinitely many distinct randomized volume complexity classes between $\Omega(\log n)$ and $O(n)$. Moreover, this hierarchy persists even when restricting to problems whose randomized and deterministic distance complexities are $\Theta(\log n)$.
\item Similar hierarchies exist for polynomial distance complexities: we show that for any $k, \ell \in \N$ with $k \leq \ell$, there are problems whose randomized and deterministic distance complexities are $\Theta(n^{1/\ell})$, randomized volume complexities are $\widetilde\Theta(n^{1/k})$, and whose deterministic volume complexities are $\widetilde\Theta(n)$.
\end{itemize}

Additionally, we consider connections between our volume model and massively parallel computation (MPC). We give a general simulation argument that any volume-efficient algorithm can be transformed into a space-efficient MPC algorithm.
\end{abstract}

\newpage

% !TEX root = lcl-volume.tex

\section{Introduction}

\begin{figure}[t]
\centering
\includegraphics[width=0.9\textwidth,page=1]{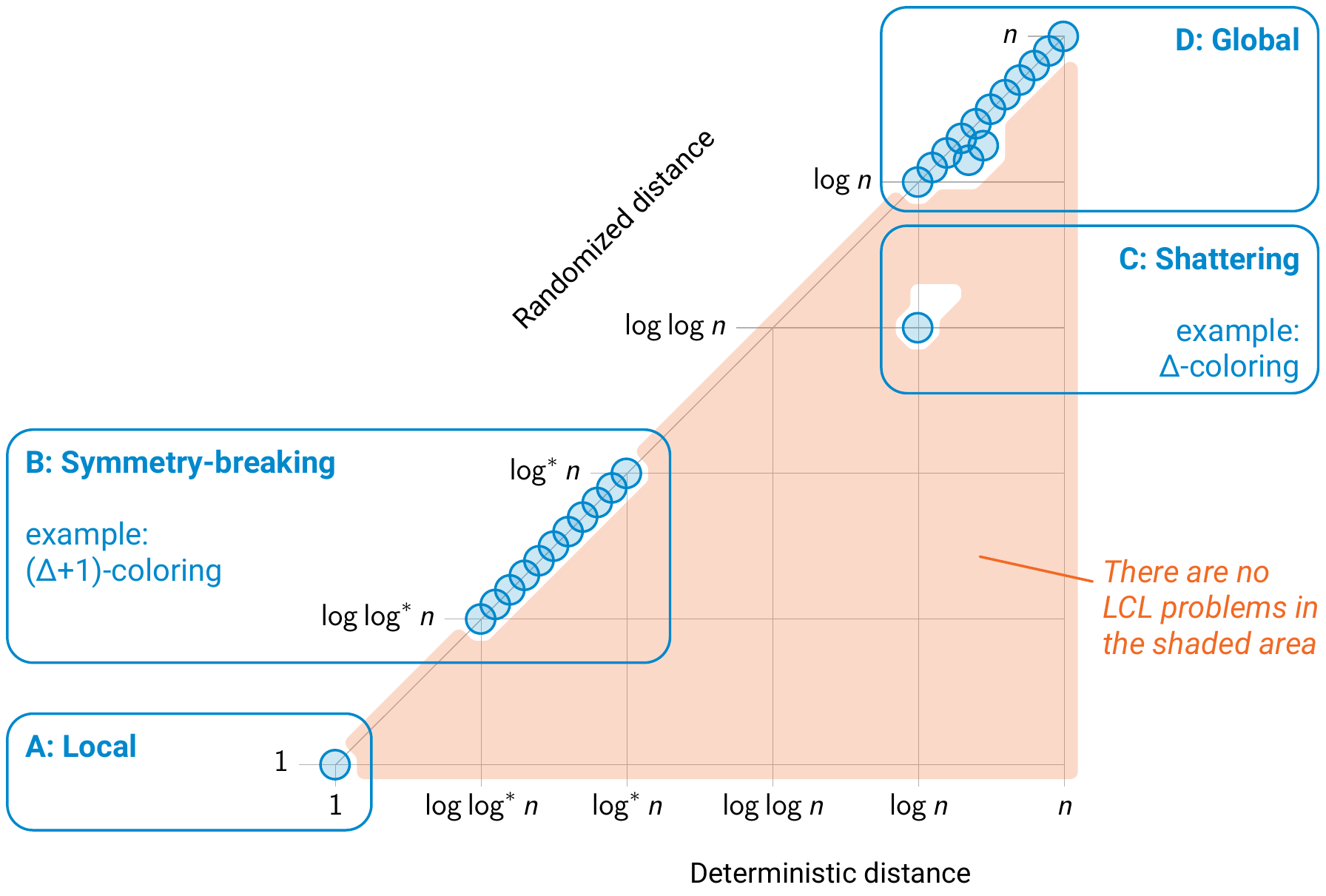}%
\caption{Prior work: locally checkable graph problems (LCLs) classified based on their distance complexity \cite{cole86deterministic,Linial1992,Naor1991,Naor1995,Chang2019,chang16exponential,ghaffari17distributed,Brandt2016,fischer17sublogarithmic,Balliu2019padding,Ghaffari2018,Balliu2018stoc,Balliu2018disc,Rozhon2019}. Blue dots represent examples of LCL problems with a known deterministic and randomized complexity, and the orange shading represents a region in which no LCL problems exist.}\label{fig:landscape}
\end{figure}

\paragraph{Distance complexity.}

In message-passing models of distributed computing, \emph{time} is intimately connected to \emph{distance}: in $T$ communication rounds, nodes can potentially learn some information that was originally within distance $T$ from them, but not further. This idea is formalized in the \emph{LOCAL model} \cite{Peleg2000,Linial1992} of distributed computing, in which a distributed algorithm with a running time $T$ is, in essence, a function that maps radius-$T$ neighborhoods to local outputs. The key question in the theory of distributed computing can be stated as follows:
\begin{quote}
\em
How far does an individual node need to see in order to produce its own part of the solution?
\end{quote}
To give some simple examples, assume we have got a graph with $n$ nodes and a maximum degree $\Delta = O(1)$, and all nodes are labeled with unique identifiers:
\begin{itemize}
\item Finding a proper vertex coloring with $\Delta + 1$ colors: Each node can pick its own color based on its radius-$O(\log^* n)$ neighborhood \cite{cole86deterministic,Linial1992,Naor1991}.
\item Finding a proper vertex coloring with $\Delta$ colors: Each node needs to see up to distance $\Omega(\log\log n)$ in order to succeed with high probability and up to distance $\Omega(\log n)$ if we are using a deterministic algorithm \cite{Brandt2016,chang16exponential,panconesi95delta,ghaffari17distributed}.
\end{itemize}
Graph coloring is an example of a \emph{locally checkable labeling} (LCL) \cite{Naor1995}. I.e., it is a graph problem in which we label nodes with labels from some finite set, and a solution is globally feasible if it looks feasible in a constant-radius neighborhood of each node. In the past several years our understanding of the distance complexity of LCLs has advanced rapidly \cite{Balliu2018stoc,Balliu2018disc,Brandt2016,chang16exponential,Chang2019,fischer17sublogarithmic,ghaffari17distributed,Ghaffari2018,Ghaffari2018a,chang18complexity,Brandt2017,Brandt2019automatic,Balliu2019padding,Balliu2019weak2col,Balliu2019decidable,Balliu2019sirocco,Rozhon2019}, and it is now known that \emph{all} LCL problems can be broadly classified in four classes, as shown in Figure~\ref{fig:landscape}. One of the key insights is that there are broad \emph{gaps} between the classes, and such gaps have immediate algorithmic applications: for example, if you can solve any LCL problem with $o(\log n)$ deterministic distance, it directly implies also a solution with $O(\log^* n)$ distance.

\paragraph{Volume complexity.}

While there has been a lot of progress on understanding how far each node must see in a graph to solve a given graph problem, this line of research has limited direct applicability beyond message-passing models of distributed computing. In many other settings---e.g., parallel algorithms and centralized sublinear-time algorithms---a key question is not how \emph{far} do we need to explore the input graph, but how \emph{many} nodes of the input graph we need to explore. One formalization of this idea is the (stateless) \emph{local computation algorithms} (LCAs, a.k.a.\ centralized local algorithms or CentLOCAL) \cite{Rubinfeld2011}, where the key question is this:
\begin{quote}
\em
How much of the input does an individual node need to see in order to produce its own part of the solution?
\end{quote}
We will refer to this as the \emph{volume} complexity of a graph problem. We will formalize the model of computing in Section~\ref{sec:model}, but in brief, the idea is this:
\begin{quote}
\em
In time $T$ each node can \textbf{adaptively} gather information about a connected component of size $T$ around itself. 
\end{quote}
A bit more precisely, in each time step a node can choose to query any neighbor of a node that it has discovered previously. The query will reveal the unique identifier of the node, its degree, and its local input (if any). In randomized algorithms, each node has an independent stream of random bits that is part of its local input. Eventually, each node has to stop and produce its own part of the solution (e.g.\ its own color if we are solving graph coloring). While we assume that a node gathers a connected region, we point out that we can make this assumption without loss of generality for a broad range of graph problems \cite{Goos2016nonlocal}.

Parnas and Ron~\cite{Parnas2007} introduced a general framework that transforms algorithms in the LOCAL model to LCAs. In their framework, an algorithm with complexity $f(n)$ yields an LCA with probe\footnote{The word ``probe'' is used in the LCA literature to refer to an atomic interaction with a data structure, whereas we use ``query.'' The latter is standard terminology, e.g., in the literature on sublinear time graph algorithms and property testing.} complexity $\Delta^{\Theta(f(n))}$. Recently, Ghaffari and Uitto~\cite{Ghaffari2019} asked if the $\Delta^{\Omega(f(n))}$ barrier inherent to Parnas and Ron's technique can be overcome by ``sparsifying'' the underlying LOCAL algorithm. They provide affirmative answers for several well-studied problems, such as maximal independent set, maximal matching, and approximating a minimum vertex cover. While there is a large body of work that introduces algorithms with a low volume complexity---see, e.g., \cite{Alon2012,Campagna2013,Even2014,Even2018,Feige2015,Ghaffari2019,Levi2014,Levi2017,Levi2017spanning,Mansour2012,Mansour2013,Parnas2007,Rubinfeld2011,Reingold2016}---what is currently lacking is an understanding of the landscape of the volume complexity.

\paragraph{Connections to Massively Parallel Computation.}

%[FIXME: Does this paragraph belong here?]
Another motivation for studying volume complexity is its connection to massively parallel computation (MPC) frameworks, such as MapReduce~\cite{Karloff2010}. In the MPC model, a system consists of $M$ machines each with $S$ local memory. An execution proceeds in synchronous rounds. In each round, each machine can communicate with all other machines---sending and receiving at most $S$ bits in total---and perform arbitrary local computations. The goal is to perform a task while minimizing the space requirement $S$ per machine as well as the number of communication rounds.

In the case where each machine represents a vertex in a network with maximum degree $\Delta$, any algorithm with distance complexity $T$ can be trivially simulated in the MPC model with space $S = \Delta^{O(T)}$ in $T$ rounds. Using graph exponentiation~\cite{Lenzen2010exponential}, this runtime can be improved to $O(\log T)$ rounds. Recently, sparsification---i.e., exploiting volume efficient algorithms---has been applied to give strongly sub-linear space algorithms in the MPC model~\cite{Ghaffari2019}. %[FIXME: more citations?].
The volume model we describe in Section~\ref{sec:model} allows us to formalize a close connection between volume and the MPC model. Specifically, in Section~\ref{sec:mpc} we show that any algorithm with volume complexity $\VOL$ can be simulated using space roughly $O(\VOL + n^c)$ and $O(\VOL)$ rounds in the MPC model for any positive constant $c$. In some cases, the runtime can be improved to $O(\log \VOL)$.

\subsection{Towards a Theory of Volume Complexity}

In this work, we initiate the study of the volume complexity landscape of graph problems. As the study of LCL problems has proved instrumental in our understanding of distance complexity, we will follow the same idea here. Some of the key research questions include the following:
\begin{itemize}
  \item What are possible deterministic and randomized volume complexities of LCL problems?
  \item Do we have the same four distinct classes of problems as what we saw in Figure~\ref{fig:landscape}, and similar gaps between the classes?
  \item For distance complexity, randomness is known to help exponentially for all problems of class C, while it is of limited use in class D and useless in classes A and B. Does a similar picture emerge for volume complexity?
  \item There are infinite families of distinct distance complexities in classes B and D (this is a distributed analogue of the time hierarchy theorem)---does it hold also for the volume complexity?
  \item How tightly can we connect the volume complexity of a problem with its computational complexity in other models of computing (e.g.\ time and message complexity in LOCAL and CONGEST models of distributed computing, and time complexity in various models of massively parallel computing)?
\end{itemize}

\subsection{Preliminary Observations}\label{ssec:prelim}

Let us now make some preliminary observations on what we can say about the volume complexity of the four classes of LCL problems that are listed in Figure~\ref{fig:landscape}. We will summarize these results in Figure~\ref{fig:landscape2}.

\begin{figure}[t]
\centering
\includegraphics[width=0.9\textwidth,page=2]{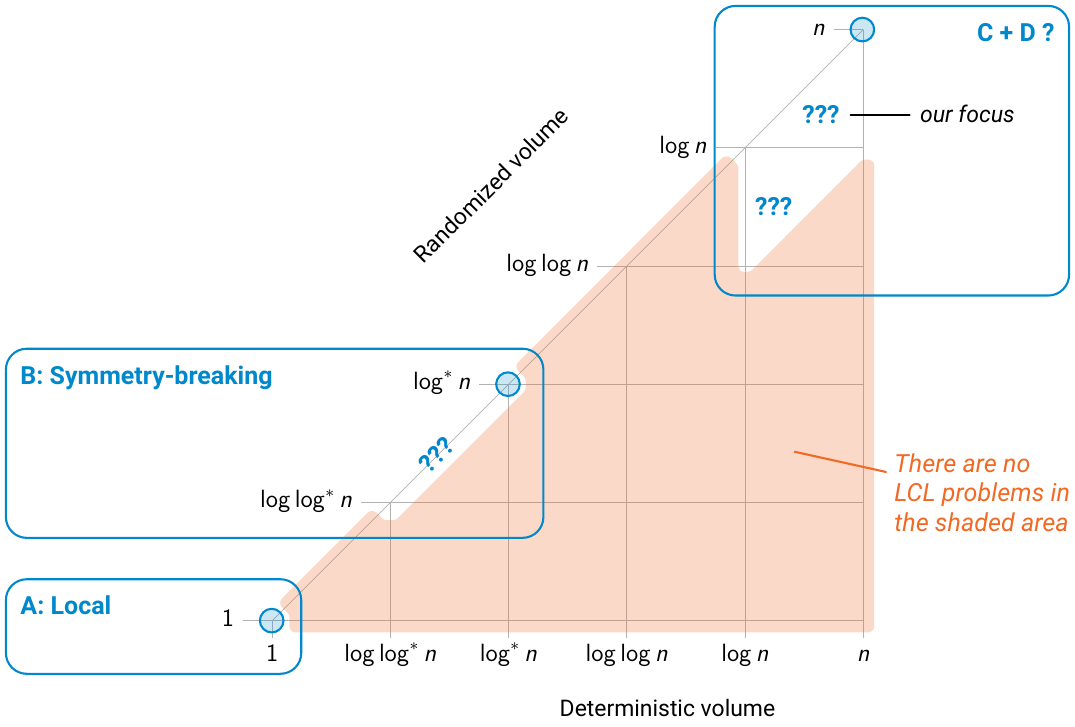}%
\caption{Preliminary observations on the landscape of volume complexities. Blue dots are some examples of LCL problems for which the volume complexity is easy to establish.}\label{fig:landscape2}
\end{figure}

\paragraph{Class A.}

Volume complexity is at least as much as the distance complexity. In graphs of maximum degree $\Delta = O(1)$, volume complexity is at most exponential in distance complexity. A distance-$T$ algorithm can be simulated if each node gathers a ball of volume $\Delta^{O(T)}$, and a volume-$T$ algorithm can be simulated if each node gathers a ball of radius $O(T)$. Hence it trivially follows that the following classes of LCL problems are equal:
\begin{itemize}[noitemsep]
    \item problems with distance complexity $\Theta(1)$,
    \item problems with volume complexity $\Theta(1)$.
\end{itemize}

\paragraph{Class B.}

Let us now look at the class of LCL problems that are solvable with distance between $\Omega(\log \log^* n)$ and $O(\log^* n)$. The trivial bounds for their volume complexity would be $\Omega(\log \log^* n)$ and $\Delta^{O(\log^* n)}$.

However, we can prove also a nontrivial upper bound. Any LCL problem in this class can be solved in two steps \cite{chang16exponential}:
\begin{enumerate}[noitemsep]
    \item find a distance-$k$ coloring for a suitable constant $k = O(1)$,
    \item apply a constant-distance mapping to the colored graph.
\end{enumerate}
It has already been known for decades that the first step can be solved in $O(\log^* n)$ distance \cite{cole86deterministic}. However, recently \citet{Even2014} introduced a graph coloring technique that makes it possible to solve the problem also in $O(\log^* n)$ volume. It follows that these classes of LCL problems are equal:
\begin{itemize}[noitemsep]
    \item problems with distance complexity between $\Omega(\log \log^* n)$ and $O(\log^* n)$,
    \item problems with volume complexity between $\Omega(\log \log^* n)$ and $O(\log^* n)$.
\end{itemize}
Moreover, the derandomization result by \citet{chang16exponential} can be used to show that randomness does not help in this region in either model (subject to some mild assumptions on the model of computing).

\paragraph{Classes C and D.}

Finally, we are left with the LCL problems that have deterministic distance between $\Omega(\log n)$ and $O(n)$ and randomized distance between $\Omega(\log \log n)$ and $O(n)$. Trivially, the volume complexity of any problem is bounded by $O(n)$, and hence the following four classes of LCL problems are equal:
\begin{itemize}[noitemsep]
    \item problems with randomized distance complexity between $\Omega(\log \log n)$ and $O(n)$,
    \item problems with deterministic distance complexity between $\Omega(\log n)$ and $O(n)$,
    \item problems with randomized volume complexity between $\Omega(\log \log n)$ and $O(n)$,
    \item problems with deterministic volume complexity between $\Omega(\log n)$ and $O(n)$.
\end{itemize}
In the distance model, it is known that in this region randomness helps at most exponentially~\cite{chang16exponential}. For example, if the randomized distance complexity is $O(\log \log n)$, then the deterministic distance complexity has to be $O(\log n)$. The same proof goes through verbatim for the volume model (under some technical assumptions on the model of computing), and hence we can conclude that e.g.\ randomized volume $O(\log \log n)$ implies deterministic volume $O(\log n)$.

\begin{figure}
\centering
\includegraphics[width=0.9\textwidth,page=3]{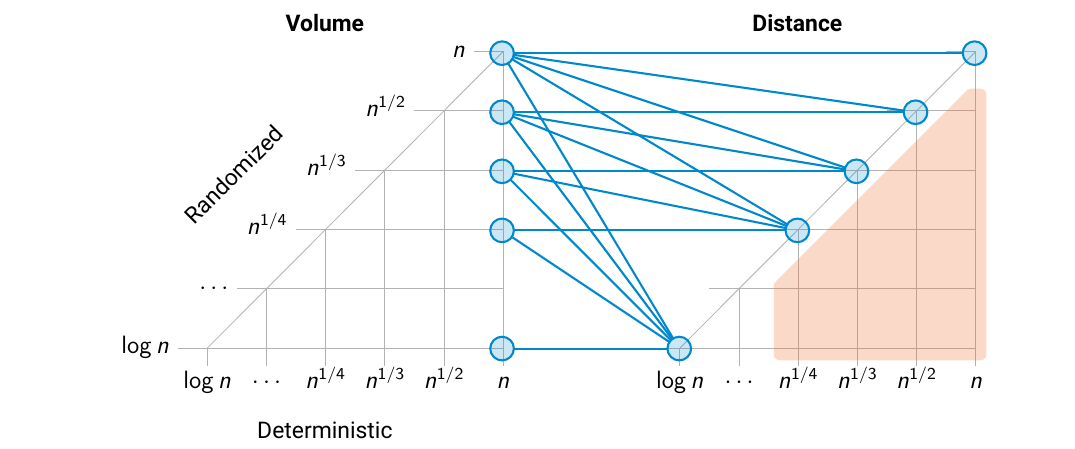}%
\caption{An overview of our contributions. Each blue line represents one LCL problem; the left end of the line indicates the randomized and deterministic volume complexity, and the right end of the line indicates the randomized and deterministic distance complexity.}\label{fig:results}
\end{figure}

\subsection{Our Contribution}

While problems of classes A and B are well-understood both from the perspective of volume and distance, the volume complexity of problems in classes C and D is wide open---indeed, it is not even known if there are distinct classes C and D for volume complexity.

In this work, we start to chart problems of class D, i.e., ``global'' problems that require $\Omega(\log n)$ distance and hence also $\Omega(\log n)$ volume for both deterministic and randomized algorithms. This is a broad class of problems, with infinitely many distinct distance complexities \cite{Chang2019,Balliu2018stoc,Balliu2018disc}.

We will show that in this region there are infinite families of LCL problems that exhibit different combinations of randomized volume, deterministic volume, randomized distance, and deterministic distance. The new complexities are summarized in Figure~\ref{fig:results} and Table~\ref{tab:results}. We make the following observations:
\begin{itemize}
\item There are \emph{infinitely many} LCLs with distinct randomized volume complexities between $\omega(\log n)$ and $o(n)$.
\item \emph{Randomness can help exponentially}, even if the deterministic volume complexity is $\Theta(n)$. This is very different from distance complexities, in which, e.g.\ a linear deterministic distance implies near-linear randomized distance \cite{Ghaffari2018}.
\item There are LCL problems in which distance complexity \emph{equals} randomized volume, and there are also LCL problems in which distance complexity is \emph{logarithmic} in randomized volume. Hence distance and volume are genuinely distinct concepts in this region. Moreover, our constructions yield a volume hierarchy theorem for randomized algorithms: There are infinitely many distinct randomized volume complexity classes between $\Omega(\log n)$ and $O(n)$, even when restricting attention to problems whose distance complexities are $\Theta(\log n)$. 
\end{itemize}

\begin{table}
\centering
\begin{tabular}{ccccc}
\toprule
Problem (section discussed) & $\RDIST$ & $\DDIST$ & $\RVOL$ & $\DVOL$\\
\midrule
$\LeafColoring$\hfill(\S\ref{sec:leaf-coloring}) & $\Theta(\log n)$ & $\Theta(\log n)$ & $\Theta(\log n)$ & $\Theta(n)$\\
$\BTL$\hfill(\S\ref{sec:balanced-tree}) & $\Theta(\log n)$ & $\Theta(\log n)$ & $\Theta(n)$ & $\Theta(n)$\\
$\HTHC(k)$\hfill(\S\ref{sec:hierarchical-coloring}) & $\Theta(n^{1/k})$ & $\Theta(n^{1/k})$ & $\widetilde\Theta(n^{1/k})$ & $\widetilde\Theta(n)$\\
$\HBTHC(k)$\hfill(\S\ref{sec:hybrid-coloring}) & $\Theta(\log n)$ & $\Theta(\log n)$ & $\widetilde\Theta(n^{1/k})$ & $\widetilde\Theta(n)$\\
$\HHTHC(k, \ell)$\hfill(\S\ref{sec:more-classes}) & $\Theta(n^{1/\ell})$ & $\Theta(n^{1/\ell})$ & $\widetilde\Theta(n^{1/k})$ & $\widetilde\Theta(n)$\\
\bottomrule
\end{tabular}
\caption{The new LCL problems constructed in this work. Here $k$ and $\ell$ are natural numbers, $k \le \ell$. We use $\widetilde\Theta$ to suppress factors that are poly-logarithmic in the argument.}\label{tab:results}
\end{table}

% !TEX root = lcl-volume.tex

\section{Model and Preliminaries}\label{sec:model}

We will now define the model of computing and the problem family that we study in this work. Here is a brief overview for a reader familiar with the LOCAL model \cite{Linial1992,Peleg2000} of distributed computing and LCAs (local computation algorithms, a.k.a., centralized local algorithms) \cite{Rubinfeld2011,Even2014}:
\begin{itemize}[noitemsep]
  \item Deterministic distance = round complexity in the deterministic LOCAL model.
  \item Randomized distance = round complexity in the randomized LOCAL model (like deterministic distance, but each node has a private random string).
  \item Deterministic volume $\approx$ probe complexity in the stateless deterministic LCA model.
  \item Randomized volume = like deterministic volume, but each node has a private random string.
\end{itemize}
Our goal here is to have a clean model that is as close to the standard LOCAL model as possible, but which captures the idea of paying for the volume that the algorithm explores. The deterministic volume model is very close to stateless deterministic LCAs---we restrict queries to a connected region, but for many graph problems this assumption does not matter~\cite{Goos2016nonlocal}. However, the randomized volume model is somewhat different from randomized LCAs; one key difference is that randomized LCAs typically have direct access to \emph{shared} randomness, while in our model each node has a \emph{private} random string. That said, low randomized volume clearly implies that there exists also an efficient randomized LCA for solving the problem. We will discuss different flavors of randomness in more detail in Section~\ref{sec:conc-random}.

\subsection{Graphs}

Our main object of study in this paper is distributed graph algorithms. In this context, an undirected graph $G = (V, E)$ represents both a communication network and the (partial) input to a problem. We denote the number of nodes in $G$ by $n = \abs{V}$. For each node $v \in V$, we denote its degree by $\deg(v)$, and we assume that for some fixed constant $\Delta \in \N$,  all nodes have degree at most $\Delta$. In any input, we assume that each node $v \in V$ is given a unique identifier from the range $[n^\alpha]$ for some arbitrary fixed $\alpha \geq 1$. For any positive integer $d$ and node $v \in V$, $N_v(d)$ denotes the $d$-radius neighborhood of $v$. That is, $N_v(d)$ is the induced subgraph of $G$ containing all nodes $w \in V$ with $\dist(v, w) \leq d$.

While we consider undirected graphs---where each edge serves as a bi-directional communication link---it is convenient to view each edge $\set{v, w} \in E$ as a pair of \emph{ordered} edges $(v, w)$ (from $v$ to $w$) and $(w, v)$ (from $w$ to $v$). We assume that input graphs additionally specify a \dft{port ordering}. For each vertex $v$ and incident edge $(v, w)$, there is an associated number $p(v, w) \in [\deg(v)]$---the \dft{port number} of $(v, w)$---such that $p$ is a bijection between (ordered) edges incident to $v$ and $[\deg(v)]$. Thus, on any input, we may speak unambiguously of $v$'s $i\th$ neighbor, as the neighbor $w$ satisfying $p(v, w) = i$ (if any).

The input to a graph problem may additionally specify an input string for each node $v \in V$. An \dft{input labeling} $\calL$ of a graph $G$ specifies $O(\log n)$-bit unique identifiers for each node, a port ordering, and any additional input required for the graph problem. We denote the input label of a particular node $v$ by $\calL(v)$. We also assume that $n$---the number of nodes in the graph---is provided as input to every algorithm.

\subsection{Algorithms and Complexity}

Each node $v \in V$ represents a single processor. Throughout an execution of an algorithm $A$ initiated at a vertex $v \in V$, $A$ maintains a set $V_v$ of \dft{visited nodes}, initialized to $V_v = \set{v}$. An execution proceeds in discrete \dft{steps}, where in each step, $A$ performs a single local \dft{query} of the form $\query(w, j)$ where $w \in V_v$ and $j \in [\deg(w)]$ is a port number. In response, $A$ receives
\begin{itemize}[noitemsep]
\item the identity of the vertex $u$ satisfying $p(w, u) = j$,
\item the degree $\deg(u)$, and
\item the entire input of $u$.
\end{itemize}
Additionally, $v$ updates $V_v \leftarrow V_v \cup \set{u}$. Following the response to a query, $A$ updates its local state, and determines its next query, or decides to produce output and halt. Given a graph $G = (V, E)$, labeling $\calL$ of $G$, and vertex $v \in V$, we denote the output of $A$ on $(G, \calL)$ initiated at $v$ by $A(v, G, \calL)$. The set of outputs of $A$ induces a new labeling $\calL'$, where $\calL'(v) = A(v, G, \calL)$. 

We consider both deterministic and randomized algorithms. For randomized algorithms, random bits used by the algorithm are treated as part of the input at each node. Specifically, each node $v \in V$ has a random string $r_v : \N \to \set{0, 1}$, where each bit $r_v(i)$ is an i.i.d.\ $0$--$1$ random variable with $\Pr(r_v(i) = 1) = \Pr(r_v(i) = 0) = 1/2$. Since we treat $r_v$ as part of $v$'s input, $r_v$ is seen by every node that queries in $v$. For technical reasons, we assume that algorithms access the random strings $r_v$ sequentially, and that for any algorithm $A$ and any labeled graph $(G, \calL)$ there exists some finite bound $b$ (which may depend on the input) such that with probability $1 - O(1/n)$ the execution of algorithm $A$ on $(G, \calL)$ accesses at most $b$ random bits.\footnote{With this assumption the derandomization result by \citet[Theorem 3]{chang16exponential} holds also in the volume model. This seems to be a very mild assumption, and it should be automatically satisfied for most ``natural'' models of computation, e.g., probabilistic Turing machines. However, in standard message passing models no computational assumptions are made about individual processors. We suspect that for LCL problems, our restriction on how randomness is used is essentially without loss of generality. See the discussion in Section~\ref{sec:conc-random}.}

We are primarily interested in two complexity measures: \emph{distance} and \emph{volume}.

\begin{dfn}
  \label{dfn:dist-cost}
  Let $A$ be an algorithm, $G = (V, E)$ a graph, $\calL$ a labeling of $G$, and $v \in V$ a node. Then the \dft{distance cost} of $A$ on $(G, \calL)$ initiated from $v$ is
  \[
  \DIST(A, G, \calL, v) = \max \set{\dist(v, w) \sucht w \in V_v},
  \]
  where $V_v$ is the set of nodes visited by the execution when $A$ terminates. Let $\calG_n$ denote the family of labeled graphs on at most $n$ nodes with maximum degree at most $\Delta$. The distance cost of $A$ on graphs of $n$ nodes is defined by
  \[
  \DIST_n(A) = \sup\bigl\{\DIST(A, G, \calL, v) \bigm| (G, \calL) \in \calG_n, G = (V, E), v \in V\bigr\}.
  \]
\end{dfn}

\begin{dfn}
  \label{dfn:vol-cost}
  Let $A$ be an algorithm, $G = (V, E)$ a graph, $\calL$ a labeling of $G$, and $v \in V$ a node. Then the \dft{volume cost} of $A$ on $(G, \calL)$ initiated from $v$ is
  \[
  \VOL(A, G, \calL, v) = \abs{V_v},
  \]
  where $V_v$ is the set of nodes visited by the execution when $A$ terminates. Let $\calG_n$ denote the family of labeled graphs on at most $n$ nodes. The volume cost of $A$ on graphs of $n$ nodes is defined by
  \[
  \VOL_n(A) = \sup\bigl\{\VOL(A, G, \calL, v) \bigm| (G, \calL) \in \calG_n, G = (V, E), v \in V\bigr\}.
  \]
\end{dfn}

\begin{rem}
  \label{rem:dist-vs-local}
  The distance cost of an algorithm in our model is closely related to the well-known LOCAL model of computation~\cite{Peleg2000,Linial1992}. In the LOCAL model, in $T$ \emph{rounds} each node can query all of its nodes within distance $T$. Thus, on input $(G, \calL)$, an algorithm $A$ can be implemented in $T$ rounds in the LOCAL model if and only if it $\DIST(A, G, \calL, v) \leq T$ for all $v \in V$. 
\end{rem}

\begin{dfn}
  Let $\Pi$ be a graph problem---that is, a family of triples $(G, \mathcal{I}, \mathcal{O})$, where $\mathcal{I}$ and $\mathcal{O}$ are input and output labelings (respectively) of $G$. We say that a \emph{deterministic} algorithm $A$ \dft{solves} $\Pi$ if for every allowable input $\mathcal{I}$ the output $\mathcal{O} = (\calL', G)$ formed by taking $\calL'(v) = A(v, G, \calL)$ satisfies $(G, \mathcal{I}, \mathcal{O}) \in \Pi$. A \emph{randomized} algorithm $A$ solves $\Pi$ if for all inputs $\mathcal{I}$
  \[
  \Pr_r((G, \mathcal{I}, \mathcal{O}) \in \Pi) = 1 - O(1/n)
  \]
  where the probability is taken over the (joint) randomness of all nodes, and $n$ is the number of nodes in $G$.

  Given a problem $\Pi$, the \dft{complexity} of the problem $\Pi$ is the infimum over all algorithms $A$ computing $\Pi$ of the cost of $A$. We denote the deterministic distance, randomized distance, deterministic volume, and randomized volume complexities of $\Pi$ by
  \begin{align*}
    \DDIST(\Pi),\quad
    \RDIST(\Pi),\quad
    \DVOL(\Pi),\quad
    \RVOL(\Pi)
  \end{align*}
  respectively.
\end{dfn}

\subsection{Comparing Distance and Volume}

Here we give an elementary relationship between distance and volume complexities.

\begin{lem}
  \label{lem:dist-vs-vol}
  Let $\Pi$ be a problem defined on the family of graphs of maximum degree at most $\Delta$. Then we have
  \begin{equation}
    \label{eqn:r-dist-vs-vol}
    \RDIST(\Pi) \leq \RVOL(\Pi) \leq \Delta^{\RDIST(\Pi)} + 1
  \end{equation}
  and
  \begin{equation}
    \label{eqn:d-dist-vs-vol}
    \DDIST(\Pi) \leq \DVOL(\Pi) \leq \Delta^{\DDIST(\Pi)} + 1.
  \end{equation}
\end{lem}
\begin{proof}
  For the first inequalities in Equations~(\ref{eqn:r-dist-vs-vol}) and~(\ref{eqn:d-dist-vs-vol}), suppose $A$ is an algorithm that solves $\Pi$ on $G = (V, E)$ with labeling $\calL$ using volume $m$. For any $v \in V$, let $V_v \subseteq V$ denote the subset of nodes queried by an execution of $A$ initiated from $v$, so that $\abs{V_v} \leq m$. Since the subgraph of $G$ induced by $V_v$ is connected, we have $\dist(v, w) \leq m$ for all $w \in V_v$, hence $\DIST(A, G, \calL, v) \leq m$.

  For the second inequalities, suppose $A$ solves $\Pi$ using distance at most $m$, and let $N_v(m)$ denote the $m$-neighborhood of $(v)$ (i.e., $N_v(m) = \set{w \in V \sucht \dist(v, w) \leq m}$). Since $G$ has maximum degree at most $\Delta$, we have $\abs{N_v(m)} \leq \Delta^m + 1$. Since $V_v \subseteq N_v(m)$, we have $\VOL(A, G, \calL, v) = \abs{V_v} \leq \Delta^m + 1$, which gives the desired result.
\end{proof}

\subsection{Comparing Volume and MPC}\label{sec:mpc}

In the MPC model~\cite{Karloff2010}, there are $M$ machines each with $S$ memory. An execution proceeds in synchronous rounds of all-to-all communication, with each node sending and receiving at most $S$ bits per round. For simplicity, consider the case where each machine stores the adjacency list of a single vertex in $G$ so that $M = n$ (and $S \gg \Delta$). Here we describe how an arbitrary algorithm with volume cost $\VOL$ can be simulated efficiently in the MPC model.

\begin{lem}
  \label{lem:volume-mpc}
  Suppose an algorithm $A$ for $\Pi$ has volume cost $\VOL$ when executed on a (labeled) graph $G$. Then for any number $c > 0$, there exists a (randomized) algorithm in the MPC model solving $\Pi$ in $O(\VOL)$ rounds using space $S = O(\VOL + n^c + \Delta)$ per node.
\end{lem}

We only sketch the proof of Lemma~\ref{lem:volume-mpc}. We show that for $S = O(n^c + \Delta)$, we can simulate each node $v \in V$ performing a single step (i.e., query and response) of $A$ in $O(1)$ rounds. The lemma follows by having each machine store the component queried by the vertex it represents in the simulation of $A$. Without loss of generality, assume the machines are labeled $1, 2, \ldots, n$. Each query is of the form $(v, w, i)$, interpreted as, \emph{``node $v$ queries for the $i\th$ neighbor of node $w$.''} We refer to $v$ as the \dft{source} of the query, and $w$ the \dft{destination}. The difficulty of the simulation arises because a single node $w$ could be the destination of many queries from different sources. Thus, queries cannot simply be sent directly from source to destination. However, all queries with a single destination must have at most $\Delta$ unique responses, corresponding to the (at most) $\Delta$ neighbors of the destination. The crux of our argument is showing how to identify the set of unique queries, and route the responses back to their sources in $O(1)$ rounds.

The basic idea is the following:

\begin{enumerate}
\item Sort the set of queries made in a single step $(v_1, w_1, i_1), (v_2, w_2, i_2), \ldots, (v_n, w_n, i_n)$ with $w_j \leq w_{j+1}$ for all $j$, breaking ties first by $i_j$, then by $v_j$. This sorting can be performed in $O(1)$ rounds using memory $O(n^c)$ per machine (with high probability) by applying an algorithm of Goodrich, Sitchinava, and Zhang~\cite{Goodrich2011sorting}. After sorting, all queries of the form $(\cdot, w_j, i_j)$ will be stored in consecutive machines.
\item If $(w_j, i_{j}) \neq (w_{j+1}, i_{j+1})$, machine $j$ sends its query to (the machine hosting) $w_j$, and $w_j$ sends its response to machine $j$ in the following round. Since only a single request of the form $(\cdot, w_j, i_j)$ is sent to $w_j$, $w_j$ receives/sends at most $\Delta$ requests in total.
\item Machines receiving responses propagate the responses backwards (to smaller $j$'s) in $O(1)$ rounds. Specifically, let $t$ be the round in which a machine $j$ receives the response to $(v_j, w_j, i_j)$ directly from $w_j$. Then in round $t+1$, $j$ sends the response to nodes $j - 1, j - 2, \ldots, j - n^c$. In round $t + 2$, each machine $j'$ with $j - n^c \leq j' \leq j$ and $(w_{j'}, i_{j'}) = (w_j, i_j)$ sends the responses to nodes $j' - n^c, j' - 2 n^c, \ldots, j' - n^{2 c}$. This continues for $O(1 / c) = O(1)$ rounds, at which point every machine $j$ stores the response to the query $(v_j, w_j, i_j)$. During each of these rounds, each node sends and receives at most $n^c$ messages.
\item Machine $j$ sends the response to the query $(v_j, w_j, i_j)$ to $v_j$. Each node sends and receives a single message.
\end{enumerate}

Lemma~\ref{lem:volume-mpc} follows directly from analyzing the simulation described above. We note that the same argument also applies to simulating LCAs in the MPC model.
% [FIXME: We are unaware of previous work showing a direct connection between LCAs and the MPC model].
We discuss further connections between the volume model and MPC in Section~\ref{sec:discussion}.

\subsection{LCLs}

In this paper, we are primarily interested in the study of locally checkable labeling problems (LCLs) \cite{Naor1995}. Suppose $\Pi = \set{(G, \calI, \calO)}$ is a graph problem such that the sets of possible input and output labels are finite. Informally, $\Pi$ is an LCL if a global output $\calO$ is valid if and only if $\calO$ is valid on a bounded radius neighborhood of every node in the network. Since we consider families of graphs such that maximum degree $\Delta$ is bounded, every LCL has a finite description: it is enough to enumerate every possible input labeling of every $c$-radius neighborhood of a node, together with the list of valid output labelings for each input-labeled neighborhood. Familiar examples of LCLs include $k$-coloring (for fixed $k$), maximal independent set, and maximal matching.

\begin{dfn}
  Fix a positive integer $\Delta$ and let $\calG^{\Delta}$ denote the family of graphs with maximum degree at most $\Delta$. Let $\Lin$ and $\Lout$ be finite sets of input and output labels, respectively. Suppose
  \[
  \Pi \subseteq \bigl\{(G, \calI, \calO) \bigm| G = (V, E) \in \calG^\Delta, \calI : V \to \Lin, \calO : V \to \Lout \bigr\}
  \]
  is a graph problem. We call $\Pi$ a \dft{locally checkable labeling problem} or \dft{LCL} if there exists an absolute constant $c$ such that $(G, \calI, \calO) \in \Pi$ if and only if for every $v \in V$, \[(N_v(c), \calI |_{N_v(c)}, \calO |_{N_v(c)}) \in \Pi.\] Here $N_v(c)$ denotes the distance $c$ neighborhood of $v$, and for a subgraph $H$ of $G$, $\calI |_H$ and $\calO |_H$ denote the restrictions of $\calI$ and $\calO$ (respectively) to $H$.
\end{dfn}

\subsection{Lower Bounds via Communication Complexity}

Here, we briefly review a technique (introduced in~\cite{Eden2018-lower}) of applying lower bounds from communication complexity to yield query lower bounds. The basic concepts for the technique are the notions of embedding of a function and query cost.

\begin{dfn}
  \label{dfn:embedding}
  For $N \in \N$, let $f : \set{0,1}^N \times \set{0,1}^N \to \set{0, 1}$ be a Boolean function. Let $\calG_n$ denote the set of (labeled) graphs on $n$ vertices. Suppose $\calE : \set{0, 1}^N \times \set{0, 1}^N \to \calG_n$, and let $g : \calG_n \to \set{0, 1}$. We say that the pair $(\calE, g)$ is an \dft{embedding} of $f$ if for all $x, y \in \set{0, 1}^N$, $f(x, y) = g(\calE(x, y))$.
\end{dfn}

Suppose two parties, Alice and Bob, hold private inputs $x$ and $y$ respectively, and wish to compute $f(x, y)$. Given an embedding $(\calE, g)$ of $f$ as above, any algorithm $A$ that computes $g$ on $\calG_n$ gives rise to a two-party communication protocol that Alice and Bob can use to compute $f$. Alice an Bob individually simulate an execution of $A$ on $\calE(x, y)$. Whenever $A$ queries $\calE(x, y)$, Alice and Bob exchange sufficient information about their private inputs $x$ and $y$ to simulate the response to $A$'s query to $\calE(x, y)$. If the responses to all such queries can be computed by Alice and Bob with little communication, then we may infer a lower bound on the number of queries needed by $A$ to compute $g$. Indeed, the number of queries needed to compute $g$ is at least the communication complexity of $f$ divided by the maximum number of bits Alice and Bob must exchange in order to answer a query.

\begin{dfn}
  \label{dfn:query-cost}
  Let $q : \calG_n \to \set{0, 1}^*$ be a query and $(\calE, g)$ an embedding of $f$. We say that $q$ has \dft{communication cost} at most $B$ and write $\cost_{\calE}(q) \leq B$ if there exists a (zero-error) two-party communication protocol $\Pi_q$ such that for all $x, y \in \set{0, 1}^N$ we have $\Pi_q(x, y) = q(\calE(x, y))$ and $\abs{\Pi_q(x, y)} \leq B$.
\end{dfn}

The main result of~\cite{Eden2018-lower} shows that given an embedding $(\calE, g)$ of $f$, the \emph{query} complexity of $g$ is bounded from below by the communication complexity of $f$ divided by the communication cost of simulating each query.

\begin{lthm}[\cite{Eden2018-lower}]
  \label{thm:query-lb-from-cc}
  Let $Q$ be a set of allowable queries, $f : \set{0,1}^N \times \set{0, 1}^N \to \set{0, 1}$, and $(\calE, g)$ an embedding of $f$. Suppose that each query $q \in Q$ satisfies $\cost_{\calE}(q) \leq B$, and $A$ is an algorithm that computes $g$ using $T$ queries (in expectation) from $Q$. Then $T = \Omega(R(f) / B)$, where $R(f)$ is the (randomized) communication complexity of $f$.
\end{lthm}

In Section~\ref{sec:balanced-tree}, we apply Theorem~\ref{thm:query-lb-from-cc} using an embedding of the disjointness function, $\disj$, defined as follows:
\[
\disj(x, y) =
\begin{cases}
  1 &\text{if } \sum_{i = 1}^N x_i y_i = 0\\
  0 &\text{otherwise}.
\end{cases}
\]
We apply the following fundamental result of Kalyanasundaram and Schnitger on the communication complexity of $\disj$.

\begin{lthm}[\cite{Kalyanasundaram1992, Razborov1992}]
  \label{thm:disj-lb}
  The randomized communication complexity of the disjointness function is $R(\disj) = \Omega(N)$. This result holds even if $x$ and $y$ are promised to satisfy $\sum_{i = 1}^N x_i y_i \in \set{0, 1}$.
\end{lthm}

\subsection{Tail Bounds}

In our analysis of randomized algorithms, we will employ the following standard Chernoff bounds. See, e.g.,~\cite{Mitzenmacher2005-probability} for derivations.

\begin{lem}[\cite{Mitzenmacher2005-probability}, Theorems~4.4 and~4.5]
  \label{lem:chernoff}
  Suppose $Y_1, Y_2, \ldots, Y_m$ are independent random variables with $\Pr(Y_i = 1) = p_i$ and $\Pr(Y_i = 0) = 1 - p_i$. Let $Y = \sum_{i = 1}^m$ and $\mu = \E(Y) = \sum_{i = 1}^m p_i$. Then for any $\delta$ with $0 < \delta < 1$ we have
  \begin{equation}
    \label{eqn:chernoff-upper}
    \Pr(Y \geq (1 + \delta) \mu) \leq e^{- \mu \delta^2 / 3}    
  \end{equation}
  and
  \begin{equation}
    \label{eqn:chernoff-lower}
    \Pr(Y \leq (1 - \delta) \mu) \leq e^{-\mu \delta^2 / 2}.    
  \end{equation}
\end{lem}

We will also require tail bounds for the \emph{negative binomial distribution}, defined as follows. For any positive integer $k$ and $p \in (0, 1]$, let $Y_1, Y_2, \ldots$ be a sequence of independent Bernoulli random variables with parameter $p$ (i.e., $\Pr(Y_i = 1) = p$ and $\Pr(Y_i = 0) = 1 - p$ for all $i$). Then the random variable
\[
N = \inf\set{m \sucht \sum_{i = 1}^m Y_i \geq k}
\]
is distributed according to the \dft{negative binomial distribution} $\calN(k, p)$. (For completeness, we use the convention that $\inf \varnothing = 0$.)

Notice that for $N \sim \calN(k, p)$, we have
\[
\Pr(N > m) = \Pr\paren{\sum_{i = 1}^m Y_i < k}.
\]
Setting $m = c \cdot k / p$ for any $c > 1$, the sum on the right has expected value $\mu = c \cdot p \cdot m = c \cdot k$. Taking $Y = \sum_{i = 1}^m Y_i$, we then obtain $N > m$ if and only if $Y < (1 - \delta) \mu$ for $\delta = (c - 1) / c$. Applying the Chernoff bound~(\ref{eqn:chernoff-lower}) to bound the right side of the expression above gives the following result.

\begin{lem}
  \label{lem:neg-binom}
  Suppose $N \sim \calN(k, p)$. Then
  \[
  \Pr(N > c \cdot k / p) \leq e^{-k (c - 1)^2 / 2 c}.
  \]
\end{lem}

% !TEX root = lcl-volume.tex

\section{Leaf Coloring}
\label{sec:leaf-coloring}

In this section, we describe an LCL problem, $\LeafColoring$, whose randomized distance, deterministic distance, and randomized volume complexities are $O(\log n)$, but whose deterministic volume complexity is $\Omega(n)$.

Before defining $\LeafColoring$ formally, we describe a ``promise'' version of the problem that restricts the possible input graphs. Specifically, consider the promise that all input graphs $G = (V, E)$ are binary trees in which every node has either $0$ or $2$ children, and all edges are directed from parent to child. Moreover, each internal node (i.e., node with $2$ children) has a pre-specified right and left child. Each node $v \in V$ is assigned an input color $\chin(v) \in \set{R, B}$ ($R$ for red, $B$ for blue). The $\LeafColoring$ problem requires each node $v$ to output a color $\chout(v) \in \set{R, B}$ such that (1) if $v$ is a leaf, $\chout(v) = \chin(v)$, and (2) if $v$ is internal, it outputs the same color as one of its children. 

In the non-promise version of $\LeafColoring$, the input may be an arbitrary graph (with maximum degree at most $\Delta$). In order to mimic the promise problem described above, each node receives as input a ``tree labeling'' (defined below) that assigns a parent, right child, and left child to each node. Using this assignment, each node can locally check that its own and its neighbors input labelings are locally consistent with a binary tree structure as described in the promise version of the problem. We observe (Observation~\ref{obs:leaf-coloring-pforest}) that the set of locally consistent nodes and edges form a binary sub-pseudo-forest of $G$ (i.e., a subgraph in which each connected component contains at most a single cycle). The generic $\LeafColoring$ problem then requires each leaf in this pseudo-forest to output its input color, while each internal node outputs the same color as one of its children.

\begin{dfn}
  \label{dfn:tree-labeling}
  Let $G = (V, E)$ be a graph of maximum degree at most $\Delta$, and $\ports = [\Delta] \cup \set{\bot}$. A \dft{(binary) tree labeling} consists of the following for each $v \in V$:
  \begin{itemize}[noitemsep]
  \item a \dft{parent}, $\parent(v) \in \ports$,
  \item a \dft{left child}, $\lc(v) \in \ports$,
  \item a \dft{right child}, $\rc(v) \in \ports$.
  \end{itemize}
  A \dft{colored tree labeling} additionally specifies for each $v \in V$
  \begin{itemize}
  \item a \dft{color}, $\chin(v) \in \set{R, B}$.
  \end{itemize}
  We refer to $R$ as \dft{red} and $B$ as \dft{blue} in our depictions of colored tree labelings. For a fixed node $v$, we call the labeling of $v$ \dft{well-formed} if the non-$\bot$ ports $\parent(v)$, $\lc(v)$, and $\rc(v)$ are pair-wise distinct. For example, we have $\parent(v) \neq \lc(v)$, unless both are $\bot$. 
\end{dfn}

\begin{ntn}
  While the labels $\parent$, $\lc$, and $\rc$ are formally elements of $\ports$, it will be convenient to associate, for example, $\parent(v)$ with the node adjacent to $v$ via the edge whose port label is $\parent(v)$. In particular, this convention allows us to compose labels; for example, $\parent(\lc(v))$ is the parent of $v$'s left child.
\end{ntn}

\begin{rem}
  \label{rm:tree-labeling-well-formed}
  In what follows, we assume without loss of generality that in all tree labelings, the labels of all nodes are well-formed in the sense of Definition~\ref{dfn:tree-labeling}. Indeed, an arbitrary labeling $\calL$ can be transformed to a well-formed instance in the following manner: If $v$ is not well-formed, it sets $\parent(v), \lc(v), \rc(v) = \bot$; if $v$ is well-formed, but, e.g., $\parent(v)$ is not, then $v$ sets $\parent(v) = \perp$. This preprocessing can be performed using $\Delta = O(1)$ queries.
\end{rem}

\begin{dfn}
  \label{dfn:consistent}
  Let $G = (V, E)$ be a graph and $\calL$ a well-formed tree labeling of $G$. We say that a node $v \in V$ is:
  \begin{itemize}
  \item \dft{internal} if
    \begin{enumerate}
    \item $\lc(v) \neq \bot$ and $\parent(\lc(v)) = v$,
    \item $\rc(v) \neq \bot$ and $\parent(\rc(v)) = v$;
    %% \item $\rc(v) \neq \lc(v)$,
    %% \item $\parent(v) \neq \lc(v), \rc(v)$;
    \end{enumerate}
  \item a \dft{leaf} if
    \begin{enumerate}
    \item $\lc(v) = \rc(v) = \bot$
    \item $\parent(v)$ is internal.
    \end{enumerate}
  \end{itemize}
  A node is \dft{consistent} if it is internal or a leaf. A node that is neither internal nor a leaf is \dft{inconsistent}.
\end{dfn}

\begin{dfn}
  \label{dfn:leaf-coloring}
  The problem $\LeafColoring$ consists of the following:
  \begin{description}
  \item[Input:] a colored tree labeling $\calL$
  \item[Output:] for each $v \in V$, a color $\chout \in \set{R, B}$
  \item[Validity:] for each $v \in V$ we have
    \begin{itemize}
    \item $\chout(v) = \chin(v)$ if $v$ is a leaf or inconsistent,
    \item $\chout(v) \in \set{\chout(\lc(v)), \chout(\rc(v))}$ if $v$ is internal.
    \end{itemize}
  \end{description}
\end{dfn}

\begin{figure}[t]
  \centering
  \begin{minipage}{0.45\textwidth}
    \includegraphics[width=\textwidth]{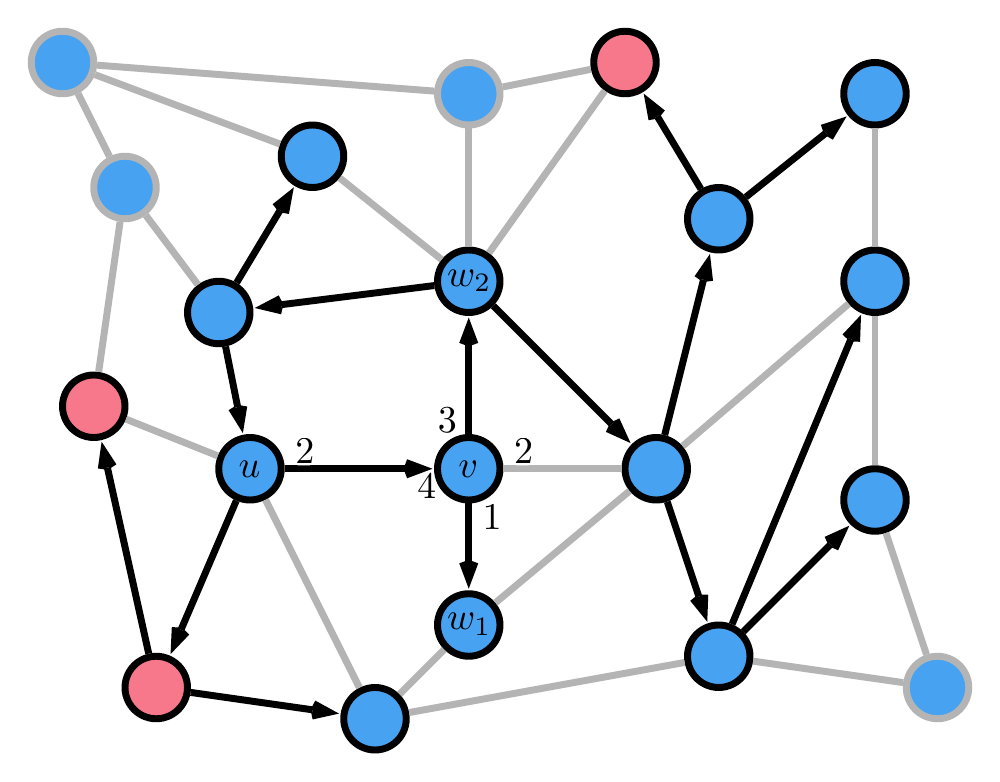}
  \end{minipage}
  \begin{minipage}{0.45\textwidth}
    \includegraphics[width=\textwidth]{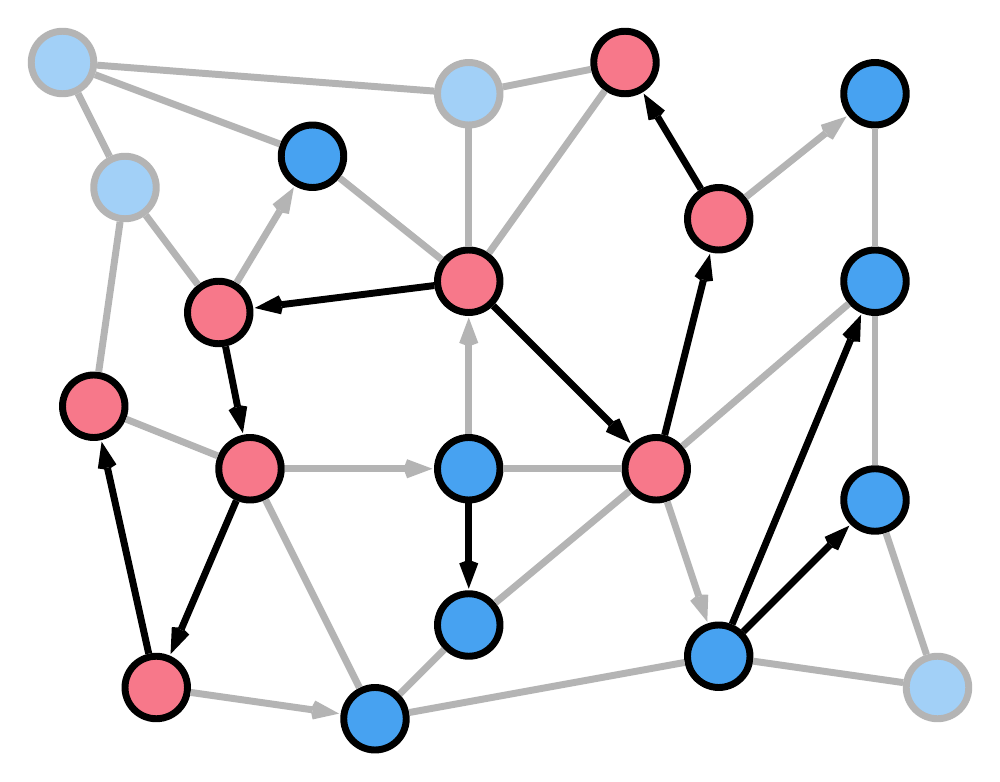}
  \end{minipage}
  \caption{An input (left) and valid output (right) for an instance of $\LeafColoring$. On the left, the consistent nodes have black borders, while the inconsistent nodes have gray borders. The subgraph consisting of black (directed) edges is the graph $G_T$ described in Observation~\ref{obs:leaf-coloring-pforest}. The node colors on the left indicate input colors, while the colors on the right are output colors. For example, the input label of $v$ is given by $\chin(v) = B$, $\parent(v) = 4$, $\rc(v) = 1$, $\lc(v) = 3$. In the output (right), black edges indicate that the parent and child in $G_T$ output the same color. Validity of the output follows because all leaves and inconsistent nodes output their input color, while all internal nodes output the color of some child in $G_T$.}
  \label{fig:leaf-coloring}
\end{figure}

\begin{lem}
  \label{lem:leaf-coloring-lcl}
  $\LeafColoring$ is an LCL.
\end{lem}
\begin{proof}
  It follows directly from Definition~\ref{dfn:tree-labeling} that a node being internal, leaf, or inconsistent is locally checkable. Thus, both validity conditions in Definition~\ref{dfn:leaf-coloring} are locally checkable as well.
\end{proof}

\begin{lthm}
  \label{thm:leaf-coloring}
  The complexity of $\LeafColoring$ is
  \begin{equation*}
    \begin{split}
      \RDIST(\LeafColoring) &= \Theta(\log n),\\
      \DDIST(\LeafColoring) &= \Theta(\log n),\\
      \RVOL(\LeafColoring) &= \Theta(\log n),\\
      \DVOL(\LeafColoring) &= \Theta(n).
    \end{split}
  \end{equation*}
\end{lthm}

Before proving Theorem~\ref{thm:leaf-coloring} in detail, we provide a high level overview of the proof. The upper bounds on $\RDIST$ and $\DDIST$ following from the observation that all internal nodes are within distance $O(\log n)$ of a leaf node. Thus, in $O(\log n)$ distance, each internal node finds its \emph{nearest} leaf (breaking ties by choosing the left-most leaf at minimal distance) and outputs the color of that leaf. In Section~\ref{sec:lc-dist-ub}, we show that this deterministic process correctly solves $\LeafColoring$ with distance complexity $O(\log n)$. This upper bound is tight, as in a balanced binary tree, the root has distance $\Omega(\log n)$ from its closest leaf. Thus, in order to distinguish the cases where all leaves are red vs blue, the root must query a node at distance $\Omega(\log n)$. (Note that if all leaves have input color, say, red, then the root must output red in any legal solution.) See Section~\ref{sec:lc-lb} for details.

The idea of the $\RVOL$ upper bound is that a ``downward'' random walk in a binary tree in which every internal node has two children will reach a leaf after $O(\log n)$ steps with high probability. To simulate such a random walk, each node in our volume-efficient algorithm chooses a single child at random. An execution of the algorithm from $v$ follows the path of chosen children until a leaf is found, and $v$ outputs the color of this child. Since all nodes along this path reach the same leaf, they all output the same color. The only complication that may arise is if $v$ encounters a (necessarily unique) cycle, in which case the path of chosen children returns to $v$. In this case, $v$ follows the edge to its child not chosen in the first step, and continues until a leaf is encountered. This second path is guaranteed to be cycle-free. Details are given in Section~\ref{sec:lc-rvol-ub}.

Finally, the argument for the lower bound on $\DVOL$ is as follows. Given any deterministic algorithm $A$ purporting to solve $\LeafColoring$ in using $q \ll n$ queries, we can adaptively construct a binary tree $G$ rooted at $v$ with $n \leq 3 q$ such that the execution $A$ initiated from $v$ never queries a leaf of $G$. By giving each leaf the input color that is the opposite of $v$'s output, we conclude that some node in $G$ must output incorrectly in this instance. See Section~\ref{sec:lc-lb} for details.

\subsection{Distance Upper Bounds}
\label{sec:lc-dist-ub}

\begin{obs}
  \label{obs:leaf-coloring-pforest}
  Let $G = (V, E)$ be a graph and $\calL$ a tree labeling of $G$. Define the directed graph $G_T = (V_T, E_T)$ by
  \[
  V_T = \set{v \in V \sucht v \text{ is internal or a leaf}}
  \]
  and
  \[
  E_T = \set{(u, v) \in V_T \times V_T \sucht u \text{ is internal and } u = \parent(v)}.
  \]
  That is, $G_T$ is the subgraph of internal nodes and leaves in $G$ where we consider only edges directed from internal parents to children. Then every node in $G_T$ has out-degree $0$ or $2$, and in-degree $0$ or $1$. In particular, this implies that $G_T$ is a (directed) pseudo-forest, and each connected component of $G_T$ contains at most one (directed) cycle. Moreover, all internal nodes have two descendants in $G_T$, and $v \in V$ is a leaf in the sense of Definition~\ref{dfn:tree-labeling} if and only if $v$ is a leaf in $G_T$.
\end{obs}

\begin{lem}
  \label{lem:near-leaf}
  Let $G = (V, E)$ and be a graph and $\calL$ a tree labeling of $G$. Suppose $v_0 \in V$ is an internal node. Then there exists a path $P = (v_0, v_1, \ldots, v_\ell)$ in $G_T$ with $\ell \leq \log n$ such that $v_\ell$ is a leaf and for all $i \in [\ell]$, $v_{i-1} = \parent(v_i)$.
\end{lem}
\begin{proof}
  Fix $v_0$ to be an internal node in $V$, and take $G_T$ as in Observation~\ref{obs:leaf-coloring-pforest}. By Observation~\ref{obs:leaf-coloring-pforest}, $v_0$ has at least one child $v_1$ such that the (directed) edge $(v_0, v_1)$ is not contained in any cycle in $G_T$. Thus, the set of descendants of $v_1$ forms a (directed) binary tree rooted at $v_1$.

  For each $r \in \N$, $r \geq 1$ define $B(r) \subseteq V_T$ to be the set of nodes containing $v_0$ and all descendants of $v_1$ (in $G_T$)  up to distance $r - 1$ (from $v_1$). Observe that if $B(r)$ contains only internal nodes, then
  \[
  \abs{B(r)} = 1 + \sum_{i = 0}^{r-1} 2^i = 2^r.
  \]
  In particular, if $r \geq \log n$, then $B(r)$ must contain a non-internal node, $w$. By Observation~\ref{obs:leaf-coloring-pforest}, $w$ is a leaf, which gives the desired result.
\end{proof}

We are now ready to prove the claims of Theorem~\ref{thm:leaf-coloring}.

\begin{prop}
  \label{prop:leaf-coloring-ddist-ub}
  Let $G = (V, E)$ be a graph on $n$ nodes, $\calL$ a colored tree labeling, and $v \in V$. Then there exists a deterministic algorithm $A$ that solves $\LeafColoring$ on $G$ with $\DIST(A, G, \calL, v) = O(\log n)$. Thus
  \[
  \DDIST(\LeafColoring), \RDIST(\LeafColoring) = O(\log n).
  \]
\end{prop}
\begin{proof}
  The algorithm solving $\LeafColoring$ works as follows. In $O(1)$ rounds, $v$ determines if it is internal, a leaf, or inconsistent. If it is not internal, $v$ outputs $\chout(v) = \chin(v)$. If $v$ is internal, in $r = O(\log n)$ rounds, $v$ queries its distance $\log n$ neighborhood, $N_v(r)$, and for each $w \in N_v(r)$, $v$ determines if $w$ is internal, leaf, or inconsistent. From this information, $v$ computes
  \[
  d = \min \set{d(v, u) \sucht u \text{ is a descendant of } v \text{ and } u \text{ is a leaf}}. 
  \]
  For any leaf $w$ that is a descendant of $v$ at distance $d$, we associate the path \[(v = w_d, w_{d-1}, \ldots, w_1, w_0 = w)\] from $v$ to $w$ with the sequence $P_w \in \set{\lc, \rc}^d$ where the $i\th$ term in the sequence indicates if $w_{d - i}$ is the left or right child of $w_{d - i + 1}$. The node $v$ then takes $w_0$ to be its ``left-most'' descendant leaf at distance $d$ and outputs $\chout(v) = \chin(w_0)$. That is $w_0$ is $v$'s distance $d$ leaf descendant that minimizes the associated sequence $P_{w_0}$ with respect to the lexicographic ordering.

  Let $(v = w_d, w_{d-1}, \ldots, w_1, w_0)$ denote the path from $v$ to $w_0$ described above. We claim that for all $i = 0, 1, \ldots, d$, $\chout(w_i) = \chin(w_0)$. In particular, this implies that $\chout(v) = \chout(w_{d-1})$, hence the second validity condition in Definition~\ref{dfn:leaf-coloring} is satisfied. We prove the claim by induction on $d$.

  For base case $d = 0$, $w_0$ is a leaf hence $\chout(w_0) = \chin(w_0)$ by the algorithm description. For the inductive step, suppose the lemma holds for all internal nodes having a descendant leaf at distance less than $d$. Suppose $v$ is a node who's nearest descendant leaf is at distance $d$. Let $w_0$ be the left-most such leaf, and let $(v = w_d, w_{d-1}, \ldots, w_0)$ be the path from $v$ to $w_0$ as above. Observe that $w_{d-1}$'s nearest descendant leaf is at distance $d - 1$, and $w_0$ is also the left-most such leaf for $w_{d-1}$. Therefore, $\chout(w_{d-1}) = \chin(w_0)$, as required.
\end{proof}

\subsection{Randomized Volume Upper Bound}
\label{sec:lc-rvol-ub}

\begin{prop}
  \label{prop:leaf-coloring-rvol-ub}
  Let $G = (V, E)$ be a graph on $n$ nodes. Then there exists a randomized algorithm that solves $\LeafColoring$ on $G$ in $O(\log n)$ volume. Thus, $\RVOL(\LeafColoring) = O(\log n)$.
\end{prop}

To prove Proposition~\ref{prop:leaf-coloring-rvol-ub} consider the algorithm, $\RWtoLeaf(v, \bot)$ (Algorithm~\ref{alg:rw-to-leaf}). If $v$ is a leaf or inconsistent, it outputs $\chin(v)$. Otherwise, if $v$ is internal, $\RWtoLeaf$ performs a (directed) random walk towards $v$'s descendants in $G_T$. When the random walk is currently at a node $w$, $w$'s private randomness is used to determine the next step of the random walk. This ensures that all walks visiting $w$ choose the same next step of the walk, hence all such walks will reach the same leaf.

The only complication arises if $G_T$ contains cycles (in which case, each connected component of $G_T$ contains at most one cycle by Observation~\ref{obs:leaf-coloring-pforest}). In this case, the random walk may return to the initial node $v_0$. If the walk returns to $v_0$, the algorithm steps towards the previously unexplored child of $v_0$. Since $G_T$ contains at most one directed cycle, the branch below $v_0$'s second child is cycle-free, thus guaranteeing that the walk eventually reaches a leaf.

\begin{rem}
  To simplify the presentation, we give an algorithm where the runtime (number of queries) is random, and may be linear in $n$. We will show that the runtime is $O(\log n)$ with high probability. In order to get a worst-case runtime of $O(\log n)$, an execution can be truncated after $O(\log n)$ steps---as $n$ is known to each node---with the node producing arbitrary output. 
\end{rem}

\begin{algorithm}
  \caption{$\RWtoLeaf(v, v_0)$. Random walk step from $v$ towards leaves. $v_0$ is the starting node of the walk, and $v$ is the current location of the walk.}
  \label{alg:rw-to-leaf}
  \begin{algorithmic}[1]
    \IF{$v$ is a leaf or inconsistent}
    \RETURN $\chin(v)$ \label{ln:rw-output}
    \ENDIF
    \IF{$v$ revisits $v_0$}\label{ln:rw-revisit}
    \STATE $b \leftarrow 1 - r_v(0)$ %\COMMENT{called if random walk returns to initial node}
    \ELSE
    \STATE $b \leftarrow r_v(0)$\label{ln:select-child}
    \ENDIF
    \IF{b = 0}
    \RETURN $\RWtoLeaf(\lc(v), v_0)$\label{ln:lc-step}
    \ELSE
    \RETURN $\RWtoLeaf(\rc(v), v_0)$\label{ln:rc-step}
    \ENDIF
  \end{algorithmic}
\end{algorithm}

\begin{proof}[Proof of Proposition~\ref{prop:leaf-coloring-rvol-ub}]
  Consider the algorithm where each node $v \in V$ outputs $\chout(v) \leftarrow \RWtoLeaf(v, v)$. If $v$ is a leaf or inconsistent, then Line~\ref{ln:rw-output} ensures that the first validity condition of Definition~\ref{dfn:leaf-coloring} is satisfied.

  Now consider the case where $v$ is internal. Let $\pi_v = (v = v_0, v_1, v_2, \ldots)$ denote the sequence of nodes visited by the random walk in the invocation of $\RWtoLeaf(v, v)$. Thus $\pi_v$ is a directed path in $G_T$. Suppose there exist indices $i < j$ with $v_i = v_j$, so that $\pi_v$ contains a cycle. By Observation~\ref{obs:leaf-coloring-pforest}, $G_T$ contains at most one cycle $C$, and all non-cycle edges are directed away from $C$ in $G_T$. Therefore, it must be the case that $i = 0$, so that the condition of Line~\ref{ln:rw-revisit} was satisfied when $\RWtoLeaf(v_j, v_0)$ was called. Thus, $v_{j+1} \neq v_1$, and $v_{j+1}$ is not contained in any cycle in $G_T$. Accordingly, define
  \[
  \pi'_v = 
  \begin{cases}
    \pi_v &\text{if } \pi \text{ is cycle-free}\\
    (v_0, v_{j+1}, v_{j+2}, \ldots) &\text{otherwise}
  \end{cases}
  \]
  Since $v_{j+1}$ is not contained in any cycle, $\pi'_v$ is a finite sequence, and by the description of $\RWtoLeaf$, $\pi'_v$ terminates at a leaf $v_\ell$. Let $w = \pi'_v(1)$ be the second node in the path $\pi'_v$. Then a straightforward induction argument (on $\ell$) shows that $\pi'_v = v \circ \pi'_w$. In particular, $\pi'_v$ and $\pi'_w$ terminate at the same leaf $v_\ell$ so that $\chout(v) = \chout(w) = \chin(v_\ell)$. Thus the second validity condition in Definition~\ref{dfn:leaf-coloring} is satisfied, as desired.

  It remains to bound the number of queries made by an invocation of $\RWtoLeaf$. Since checking if a node is internal, a leaf, or inconsistent can be done with $O(1)$ queries, each recursive call to $\RWtoLeaf$ can be performed with $O(1)$ queries. Thus, the total number of queries used by $\RWtoLeaf(v, v_0)$ is $O(\abs{\pi_v})$, where $\abs{\pi_v}$ denotes the length of $\pi_v$. We will show that with high probability for all $v \in V$, $\abs{\pi_v} = O(\log n)$, whence the desired result follows.

  First consider the case where an internal node $v$ is not contained in any cycle. Let $n_v$ denote the number of nodes reachable from $v$ in $G_T$. Since $v$ is not contained in any cycle, $n_v = 1 + n_u + n_w$ where $u$ and $w$ are $v$'s (distinct) children. Therefore, $n_u$ or $n_w$ is at most $n_v / 2$. For any edge $(w, w')$ in the path $\pi_v'$, we call the edge \emph{good} if $n_{w'} \leq n_{w} / 2$. Since $n_v \leq n$, there cannot be more than $\log n$ good edges in $\pi_v'$. Moreover, by the selection in Line~\ref{ln:select-child}, each edge in $\pi_v'$ is good independently with probability at least $1/2$.

  \begin{description}
  \item[Claim.] $\Pr(\abs{\pi_v'} \geq 16 \log n) \leq 1/n^3$
  \item[Proof of Claim.] For $i = 1, 2, \ldots$, let $Z_i$ be the indicator random variable for the event that the $i\th$ step of the random walk crosses a good edge. The $Z_i$ are independent, and we have $\Pr(Z_i = 1) \geq 1/2$ for all $i$. In fact, $\Pr(Z_i = 1) \in \set{1/2, 1}$. We define a coupled sequence of random variables $Y_i$ as follows. For $i \leq \abs{\pi_v'}$ we set
    \[
    Y_i =
    \begin{cases}
      Z_i &\text{if } \Pr(Z_i = 1) = 1/2\\
      \text{independent Bernoulli r.v. with } $p = 1/2$ &\text{if } \Pr(Z_i = 1) = 1.
    \end{cases}
    \]
    For $i > \abs{\pi_v'}$, $Y_i$ is an independent Bernoulli random variable with $p = 1/2$. Thus, the sequence $Y_1, Y_2, \ldots$ is a sequence of independent Bernoulli random variables with $p = 1/2$ (coupled to the sequence $Z_1, Z_2, \ldots$).

    Since every good edge $(v, w)$ satisfies $n_w \leq \frac 1 2 n_v$, $\pi_v'$ can contain at most $\log n$ good edges. Therefore, we have
    \[
    \sum_{i = 1}^{\abs{\pi_v'}} Y_i \leq \sum_{i = 1}^{\abs{\pi_v'}} Z_i \leq \log n,
    \]
    where the first equality holds because $Y_i \leq Z_i$ (by construction) for $i \leq \abs{\pi_v'}$. Now define the random variable $N$ by
    \[
    N = \inf\set{m \sucht \sum_{i = 1}^m Y_i \geq \log n}.
    \]
    Thus $N \sim \calN(\log n, 1/2)$ has a negative binomial distribution. Thus, by Lemma~\ref{lem:neg-binom}, $\Pr(N > 16 \log n) \leq e^{- 7^2 (\log n) / 2 \cdot 8} < n^{-3}$. The claim follows by observing that $\abs{\pi_v'} > 16 \log n$ implies that $N > 16 \log n$.
  \end{description}

  Applying a simple union bound, the claim shows that all $v$ not contained in some cycle in $G_T$ will output after $O(\log n)$ queries with high probability. If $v$ is contained in a cycle $C$, we consider two cases separately. If $\abs{C} \leq 16 \log n$, then the random walk will leave the cycle $C$ after at most $16 \log n$ steps (if it returns to the initial node). On the other hand, if $\abs{C} > 16 \log n$, essentially the same argument as given in the proof of the claim shows that the random walk started at $v$ will leave the cycle after at most $16 \log n$ steps with probability at least $1 - 1 / n^3$. Combining these observations with the conclusion of the claim, we obtain that for all $v \in V$
  \[
  \Pr(\abs{\pi_v'} \geq 32 \log n) \leq 2 / n^3.
  \]
  Taking the union bound over all $v$, we find that all nodes output after $O(\log n)$ queries with probability at least $1 - O(1/ n^2)$, which gives the desired result.
\end{proof}

\subsection{Lower Bounds}
\label{sec:lc-lb}

\begin{prop}
  \label{prop:leaf-coloring-rdist-lb}
  There exists a graph $G$ on $n$ nodes and a probability distribution on colored tree labelings of $G$ such that for any (randomized) algorithm $A$ whose distance complexity is less than $\log n - 1$, the probability that $A$ solves $\LeafColoring$ is at most $1/2$.
\end{prop}
\begin{proof}
  Let $G$ be a complete (rooted) binary tree of depth $k$, so that $n = 2^{k+1} - 1$. Consider the port ordering where the parent of each (non-root) node has port $1$, and the children of each (non-leaf) node have ports $2$ and $3$. Suppose the node identities are $1$ through $n$, where the root has ID $1$, its left and right children are $2$ and $3$, and so on. Finally, fix $\calL$ to be the tree labeling where the root has $\lc(v) = 1, \rc(v) = 2$, all (non-root) nodes have $P(v) = 1$, and all internal, non-root nodes have $\lc(v) = 2$, $\rc(v) = 3$. That is, $\calL$ is the tree labeling consistent with the tree structure of $G$. Finally, consider the distribution $D$ over input colorings where all internal nodes have $\chin(v) = R$, while all leaves have the same color $\chi_0$ chosen to be $R$ or $B$ each with probability $1/2$.

  Since every leaf $v$ in $G$ has $\chin(v) = \chi_0$, the first validity condition of Definition~\ref{dfn:leaf-coloring} stipulates that $\chout(v) = \chi_0$. A simple induction argument on the height of a node (i.e., the distance from the node to a leaf) shows that the unique solution to $\LeafColoring$ is for all $v \in V$ to output $\chout(v) = \chi_0$. 

  Suppose $A$ is any deterministic algorithm whose distance complexity is at most $k - 1$. Then an execution of $A$ initiated at the root $r$ of $G$ will not query any leaf. Therefore, $\Pr_D(\chout(r) = \chi_0) = 1/2$. By the conclusion of the preceding paragraph, the probability that $A$ solves $\LeafColoring$ is therefore at most $1/2$. By Yao's minimax principle, the no randomized algorithm with distance complexity at most $k - 1$ solves $\LeafColoring$ with probability better than $1/2$ as well.
\end{proof}

\begin{prop}
  \label{prop:leaf-coloring-dvol-lb}
  For any deterministic algorithm $A$, there exists a graph $G = (V, E)$ on $n$ nodes, a colored tree labeling $\calL$ on $G$ such that if $A$ uses fewer than $n / 3$ queries, then $A$ fails to solve $\LeafColoring$. Thus, $\DVOL(\LeafColoring) = \Omega(n)$.
\end{prop}
\begin{proof}
  Suppose $A$ uses fewer than $n / 3$ queries on all graphs $G$ on $n$ nodes. We define a process $\calP$ that interacts with $A$ and constructs a graph $G_A$ and a labeling $\calL$ such that $A$ does not solve $\LeafColoring$ on $G_A$. The basic idea is that $\calP$ constructs a binary tree $G_A$ such that $A$ never queries a leaf of the tree. The leaves of $G_A$ are then given input colors that disagree with $A$'s output.

  $\calP$ constructs a sequence of labeled binary trees $G_0, G_1, G_2, \ldots$ where $G_t$ is the tree constructed after $A$'s $t\th$ query. Initially $G_0$ is the graph consisting of a single node $v_0$ with ID $0$ and two ports, $1$ and $2$. The label of $v_0$ is $P(v_0) = \bot$, $\lc(v_0) = 1$, $\rc(v_0) = 2$. Suppose $\calP$ has constructed $G_{t-1}$, and $A$'s $t\th$ query asks for the neighbor of $v$ from port $i \in [3]$. If $i = 1$, $\calP$ returns $v$'s parent, and we set $G_t = G_{t-1}$. If $i = 2$ or $3$, $\calP$ forms $G_t$ by adding a node $w$ to $G_{t-1}$ together with an edge $\set{v, w}$. The ordered edge $(w, v)$ gets assigned port $1$, while $w$ has two ``unassigned'' ports $2$ and $3$. The label of $w$ is $P(w) = 1, \lc(w) = 2, \rc(w) = 3$, and $w$'s input color is $\chin(w) = R$.

  It is straightforward to verify (by induction on $t$) that at each step, $G_t$ is a subgraph of a binary tree $G_t'$ on at most $3 t$ nodes: take $G_t'$ to be the tree formed by appending a leaf to each unassigned port in $G_t$. If $A$ halts after $T$ queries and outputs $\chout(v_0) = C$, define $G_A = G_T'$, and complete the tree labeling of $G_A$ by assigning $P(w) = 1, \lc(w) = \bot, \rc(w) = \bot$ for all new leaves. Finally, for each leaf $w$ in $G_A$, set $\chin(w) = \chi_1 \neq \chi_0$ (the color not output by $v_0$). Since all leaves have $\chin(w) = \chi_1$, validity of $\LeafColoring$ requires that all nodes in $G_A$ output $\chi_1$. However, $\chout(v_0) = \chi_0 \neq \chi_1$, so that $A$ does not solve $\LeafColoring$ on $G_A$.
\end{proof}

% !TEX root = lcl-volume.tex

\section{Balanced Tree Labeling}
\label{sec:balanced-tree}

Here, we introduce an LCL called $\BTL$. The input labeling, which we call a ``balanced tree labeling'' extends a tree labeling (Definition~\ref{dfn:tree-labeling}) by additionally specifying ``lateral edges'' between nodes. We define a locally checkable notion of ``compatibility'' (formalized in Definition~\ref{dfn:compatible}) such that the subgraph $G_T$ of consistent nodes (in the underlying tree labeling) admits a balanced tree labeling in which all nodes are compatible if and only if $G_T$ is a balanced (complete) binary tree and $G$ contains certain additional edges between nodes at each fixed depth in $G_T$.

To solve $\BTL$, each node $v$ outputs a pair $(\beta(v), p(v))$, where $\beta(v)$ is a label in the set $\set{B, U}$ (for $B$alanced, $U$n-balanced) and $p(v) \in \calP$ is a port number. The interpretation is that if every vertex $w$ in the sub-tree of $G_T$ rooted at $v$ is compatible, then $v$ should output $(B, \parent(v))$. If $v$ is incompatible, it outputs $(U, \bot)$. Finally, if $v$ is compatible, but some descendant of $v$ is incompatible, then $v$ outputs $(U, p)$, where $p \in \set{\rc(v), \lc(v)}$ is a port number corresponding to the first hop on a path to an incompatible node below $v$. Thus, a valid output has the following global interpretation: Starting from any vertex, following the path of port numbers (edges) output each subsequent node terminates either at the root of a balanced binary tree, or at an incompatible node.

\begin{dfn}
  \label{dfn:btl}
  Let $G = (V, E)$ be a graph of maximum degree at most $\Delta$. A \dft{balanced tree labeling} consists of a tree labeling (Definition~\ref{dfn:tree-labeling}) together with the following labels for each node $v \in V$:
  \begin{itemize}[noitemsep]
  \item a \dft{left neighbor} $\ln(v) \in \calP$,
  \item a \dft{right neighbor} $\rn(v) \in \calP$.
  \end{itemize}
\end{dfn}

\begin{dfn}
  \label{dfn:compatible}
  Let $G = (V, E)$ be a graph and $\calL$ a balanced tree labeling on $G$. Suppose $v$ is consistent in the sense of Definition~\ref{dfn:consistent}. We say that $\calL$ is \dft{compatible} at a node $v$ if the following conditions hold:
  \begin{itemize}
  \item \dft{type-preserving}: If $v$ is internal (respectively a leaf), then $\rn(v)$ and $\ln(v)$ are internal (respectively leaves) or $\bot$.
  \item \dft{agreement}: If $\ln(v) \neq \bot$ then $\rn(\ln(v)) = v$; if $\rn(v) \neq \bot$ then $\ln(\rn(v)) = v$. 
  \item \dft{siblings}: If $\lc(v), \rc(v) \neq \bot$ (i.e., $v$ is internal) then $\rn(\lc(v)) = \rc(v)$ and $\ln(\rc(v)) = \lc(v)$.
  \item \dft{persistence}: If $v$ is internal and $w = \rn(v) \neq \bot$, then $w$ is internal and $\rn(\rc(v)) = \ln(\lc(w))$. Symmetrically, if $u = \ln(v) \neq \bot$ then $u$ is internal and $\ln(\lc(v)) = \rn(\rc(v))$.
  \item \dft{leaves}: If $v$ is a leaf then $\ln(v) \neq \bot \implies \ln(v)$ is a leaf and $\rn(v) \neq \bot \implies \rn(v)$ is a leaf.
  \end{itemize}
  The labeling $\calL$ is \dft{globally compatible} if every consistent vertex $v$ is compatible.
\end{dfn}

\begin{dfn}
  \label{dfn:btl-problem}
  The problem $\BTL$ consists of the following:
  \begin{description}
  \item[Input:] a balanced tree labeling $\calL$
  \item[Output:] for each $v \in V$, a pair $(\beta(v), p(v)) \in \set{B, U} \times \ports$
  \item[Validity:] for each consistent $v \in V$ we have
    \begin{enumerate}
    \item if $v$ is not compatible, then $v$ outputs $(U, \bot)$
    \item if $v$ is a compatible leaf then $v$ outputs $(B, \parent(v))$
    \item if $v$ is compatible and internal then
      \begin{enumerate}
      \item if $\lc(v)$ and $\rc(v)$ output $(B, \parent(\lc(v)))$ and $(B, \parent(\rc(v)))$, respectively, then $v$ outputs $(B, \parent(v))$
      \item if $\lc(v)$ (resp.\ $\rc(v)$) outputs $(U, \cdot)$, then $v$ outputs $(U, \lc(v))$ (resp.\ $(U, \rc(v))$)
      \end{enumerate}
    \end{enumerate}    
  \end{description}
\end{dfn}

\begin{lem}
  \label{lem:btl-lcl}
  $\BTL$ is an LCL.
\end{lem}
\begin{proof}
  As noted before, checking if a node is internal, a leaf, or inconsistent can be done locally. Also, it is clear that all of the conditions for compatibility (Definition~\ref{dfn:compatible}) are locally checkable. Thus, the validity conditions of $\BTL$ are also locally checkable.
\end{proof}

\begin{lthm}
  \label{thm:btl}
  The complexity of $\BTL$ is
  \begin{equation*}
    \begin{split}
      \RDIST(\BTL) &= \Theta(\log n),\\
      \DDIST(\BTL) &= \Theta(\log n),\\
      \RVOL(\BTL) &= \Theta(n),\\
      \DVOL(\BTL) &= \Theta(n).
    \end{split}
  \end{equation*}
\end{lthm}

The proof of Theorem~\ref{thm:btl} is as follows. In Section~\ref{sec:btl-valid}, we show that in a valid output for $\BTL$, any consistent internal node is either the root of a balanced binary tree, or it has an inconsistent descendant within distance $\log n$. Thus, each node can determine its correct output by examining its $O(\log n)$ radius neighborhood. The upper bound is tight, as a node may need to see up to distance $\Omega(\log n)$ in order to see its nearest incompatible or inconsistent node.

For the volume lower bounds, consider a balanced binary tree with lateral edges such that there exists a globally compatible labeling, and let $\calL$ be such a labeling (see Figure~\ref{fig:balanced-tree}). By modifying the input of a single pair of sibling leaves, we can form an input labeling $\calL'$ which is not globally compatible. The validity conditions of $\BTL$ imply that the root of the tree must be able to distinguish $\calL$ from $\calL'$ to produce its output. Therefore, to solve $\BTL$, the root must query a large fraction of the leaves---i.e., $\Omega(n)$ nodes. We formalize this argument using the communication complexity framework of Eden and Rosenbaum~\cite{Eden2018-lower} in Section~\ref{sec:btl-vol-lb}.

\subsection{Structure of Valid Outputs}
\label{sec:btl-valid}

\begin{lem}
  \label{lem:btl-global}
  Suppose $G = (V, E)$ is a graph, $\calL$ a balanced tree labeling of $G$, and $v \in V$ is consistent. Then either the sub-(pseudo)tree of $G_T$ rooted at $v$ is a balanced binary tree (i.e., all leaves below $v$ are at the same distance from $v$), or there exists a descendant $w$ of $v$ with $\dist(v, w) < \log n$ such that $w$ is incompatible. 
\end{lem}
\begin{proof}
  Let $H$ denote the sub-(pseudo)tree of $G_T$ rooted at $v$. Suppose $H$ is not a balanced binary tree. We will show that there is an incompatible $w$ in $H$ within $\log n$ distance from $v$. To this end for $d \in \N$, let $D_v(d)$ denote the set of descendants $w$ of $v$ in $H$ such that there exists a path from $v$ to $w$ of length $d$. That is, $D_v(0) = \set{v}$, and for each $d \geq 1$, $D_v(d)$ consists of the children of nodes in $D_v(d - 1)$.\footnote{Since $G_T$ is a pseudo-forest, and hence, could contain cycles, it may be that the same vertex is contained in $D_v(d_i)$ for different values of $d_1, d_2, \ldots$. Since each node has at most a single parent, there is still a unique path from $v$ to $w$ of each length $d_i$.} Let $H_d$ be the (induced) subgraph of $G$ with vertex set $D_v(0) \cup D_v(1) \cup \cdots \cup D_v(d)$.
  \begin{description}
  \item[Claim.] Suppose every vertex $v$ in $H_d$ is compatible. Then $D_v(d)$ is \emph{laterally connected}. That is, for every $u, w$ in $D_v(d)$, there exists a path connecting $u$ and $w$ consisting only of edges $\set{x, y}$ such that $y = \rn(x)$ (and symmetrically $x = \ln(y)$).
  \item[Proof of claim.] We argue by induction on $d$. The case $d = 0$ is trivial as $D_v(0)$ consists of a single vertex. Now suppose the claim holds for $d - 1$, and take $u, w \in D_v(d-1)$. By the inductive hypothesis, there exists a path $v_0, v_1, \ldots v_\ell$ with $\parent(u) = v_1$ and $\parent(w) v_\ell$, and without loss of generality (by possibly exchanging the roles of $u$ and $w$) we have $v_{i} = \rn(v_{i-1})$ for $i = 1, 2, \ldots, \ell$. For each $i$, let $u_i = \lc(v_i)$ and $w_i = \rc(v_i)$. By the siblings property of compatibility, we have $w_i = \rn(v_i)$, and by persistence, $v_{i+1} = \rn(w_i)$. Therefore, the sequence $v_0, w_0, v_1, \ldots, v_\ell, w_\ell$ forms a path. Since $u \in \set{v_0, w_0}$ and $w \in \set{v_\ell, w_\ell}$, the claim follows.
  \end{description}
  Using the claim, we will show that $H$ contains an incompatible node $w$. Since $H$ is assumed not to be balanced, there exist leaves $u$ and $u'$ at distances $d$ and $d'$ (respectively) from $v$ with $d' > d$. In particular, take $u$ to be the nearest leaf to $v$, and $u''$ be $u'$'s (unique) ancestor in $D_v(d)$. By the claim, there exists a path $v_0, v_1, \ldots, v_\ell$ between $u$ and $u''$ in $D_v(d)$ such that for each $i$ we have $v_i = \rn(v_{i-1})$. Without loss of generality, assume $u = v_0$ and $u'' = v_\ell$. Since $v_0$ is a leaf and $v_\ell$ is internal, there exists some $i \in [\ell]$ such that $v_{i-1}$ is a leaf, and $v_{i}$ is internal or inconsistent. However, this implies that $w = v_{i-1}$ is incompatible. Moreover, $\dist(v, w) = \dist(v, u) = d$ which is at most $\log n$ (as $u$ was chosen to be the nearest leaf to $v$), which gives the desired result.
\end{proof}

\begin{figure}
  \centering
  \includegraphics[width=0.75\textwidth]{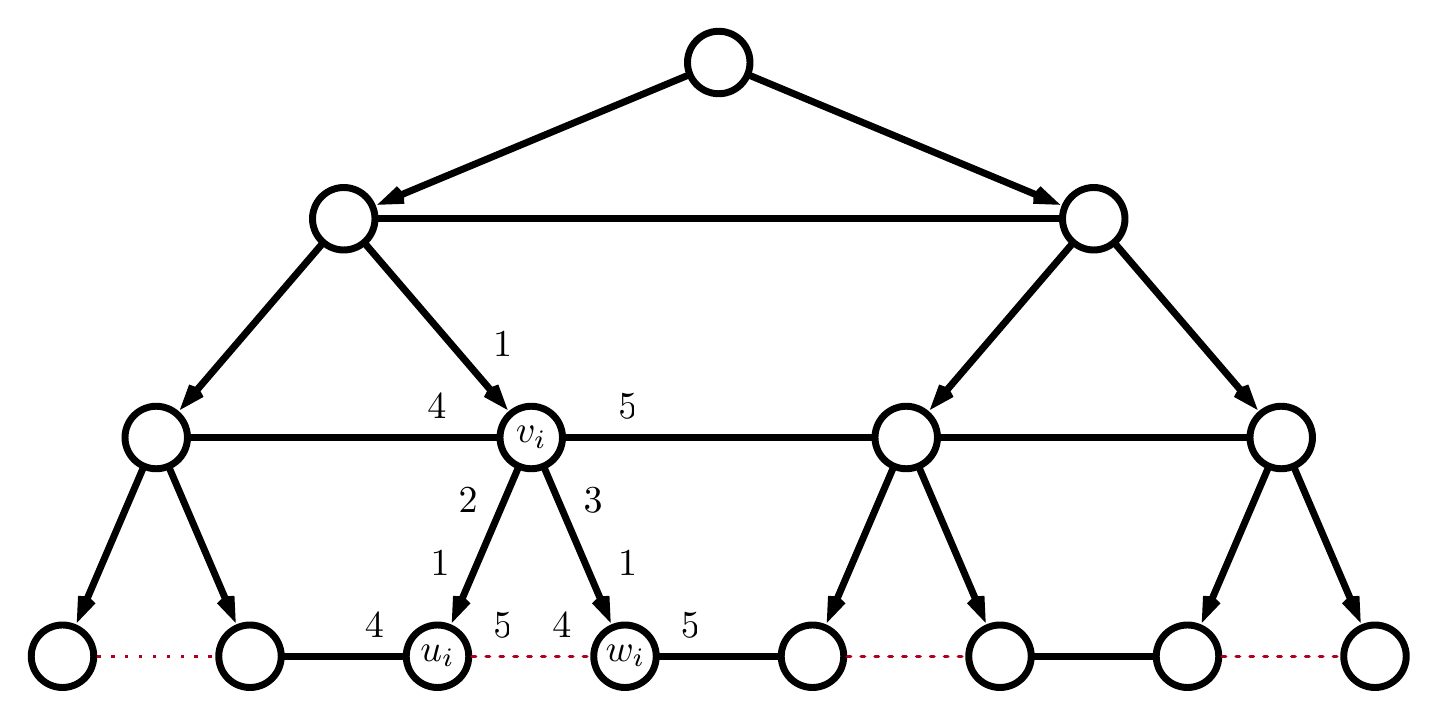}
  \caption{An instance of $\BTL$ constructed in the proof of Proposition~\ref{prop:btl-rvol-lb}. The diagonal edges are directed from parent to child, as specified by the underlying binary tree labeling. The horizontal edges are left and right neighbors. The dashed edges between leaves are included in $G$. Given strings $a, b \in \set{0,1}^N$, the labeling of the leaves are determined as follows: we have $\rn(u_i) = \ln(w_i) = \bot$ if $a_i = b_i = 1$; otherwise $\rn(u_i) = 5$ and $\ln(w_i) = 4$. The labeling constructed in this way is globally compatible if and only if $\disj(a, b) = 1$. In this case, the unique valid output is to label each node $(U, p)$.}
  \label{fig:balanced-tree}
\end{figure}

\begin{lem}
  \label{lem:btl-compatible}
  Suppose $G = (V, E)$ is a graph and $\calL$ a globally compatible labeling. Then in every valid solution to $\BTL$, every consistent node $v$ outputs $(B, \parent(v))$. Conversely, if $v$ has a descendant $w$ in $G_T$ that is incompatible, then $v$ outputs $(U, \cdot)$ in any valid solution to $\BTL$.
\end{lem}
\begin{proof}
  First consider the case where $\calL$ is globally compatible. We argue by induction on the height of $v$ that $v$ outputs $(B, \parent(v))$. For the base case, the height of $v$ is $0$, hence $v$ is a leaf. Then by Condition~2 of validity $v$ outputs $(B, \parent(v))$. For the inductive step, suppose all nodes at height $h - 1$ output $(B, \parent(v))$, and $v$ is at height $h$. In particular, the children of $v$ both output $(B, v)$. Therefore, $v$ outputs $(B, \parent(v))$ by Condition~3(a) of validity. This gives the first conclusion of the lemma.

  Now suppose $v$ has a descendant $w$ in $G_T$ that is incompatible. Let $v = u_\ell, u_{\ell-1}, \ldots, u_0 = w$ be the path from $v$ to $w$ in $G_T$. We argue by induction on $i$ that $u_i$ outputs $(U, \cdot)$. The base case $i = 0$ follows from Condition~1 of validity. The inductive step follows from Condition~3(b) of validity: since $u_i$ has a child that outputs $(U, \cdot)$ (namely, $u_{i-1}$), $u_i$ must output $(U, \cdot)$ as well. 
\end{proof}

\begin{prop}
  \label{prop:btl-ddist-ub}
  There exists an algorithm $A$ with the following property. Let $G = (V, E)$ be a graph on $n$ nodes and $\calL$ a balanced tree labeling.  Then $A$ solves $\BTL$ on $(G, \calL)$ and for all $v \in V$, $\DIST(A, G, \calL, v) = O(\log n)$. Therefore,
  \[
  \DDIST(\BTL), \RDIST(\BTL) = O(\log n).
  \]
\end{prop}
\begin{proof}
  Consider the following algorithm, $A$. Starting from a vertex $v$, $A$ searches the $O(1)$ neighborhood to determine if $v$ is internal, a leaf, or inconsistent. If $v$ is inconsistent, it outputs $(B, \bot)$. If $v$ is a leaf, it outputs $(B, \parent(v))$ if all of the conditions of compatibility (Definition~\ref{dfn:compatible}) are satisfied, and $(U, \bot)$ otherwise. Finally, if $v$ is internal, $A$ queries $v$'s $\log n + O(1)$ neighborhood in order to find its nearest leaf, which is at distance $d \leq \log n$. If $v$ sees a descendant within distance $d$ that is incompatible, then $v$ outputs $(U, p)$ where $p$ is the port towards the nearest such in compatible node, breaking ties by choosing the left-most descendant. Otherwise, $v$ outputs $(B, \parent(v))$.

  Conditions~1 and~2 in the validity of Definition~\ref{dfn:btl-problem} are trivially satisfied for all node. Consider the case where $v$ is internal and compatible. By Lemma~\ref{lem:btl-global}, if the subtree $H$ of $G_T$ rooted at $v$ is not balanced, then there is a (closest, leftmost) incompatible node $w$ in $H$ at distance at most $\log n$. Thus, in this case $v$ outputs $(U, p(v))$ where $p(v)$ is the port towards $w$. Similarly, the child $v' = p(v)$ will output $(U, p(v'))$ so that condition 3(b) of validity is also satisfied. Finally, if $H$ is balanced and all nodes in $H$ are compatible, then $v$ will output $(B, \parent(v))$, as will all other nodes $w$ in $H$. Thus condition 3(a) of validity is also satisfied.
\end{proof}

\subsection{Volume Lower Bounds}
\label{sec:btl-vol-lb}

\begin{prop}
  \label{prop:btl-rvol-lb}
  Any (randomized) algorithm $A$ that solves $\BTL$ with probability bounded away from $1/2$ requires $\Omega(n)$ queries in expectation. Thus
  \[
  \RVOL(\BTL), \DVOL(\BTL) = \Omega(n).
  \]
\end{prop}
\begin{proof}
  For any $k \in \N$ form the graph $G$ by starting with the complete binary tree of depth $k$. Assign IDs, port numbers, and labels as in the proof of Proposition~\ref{prop:leaf-coloring-rdist-lb}, so that the root has ID $1$, its left child has ID $2$, its right child has ID $3$, and so on. In particular, the nodes at depth $d$ have IDs $2^{d}, 2^{d} + 1, \ldots, 2^{d+1} - 1$. For each $d$, add lateral edges between nodes with IDs $2^{d} + i - 1$ and $2^{d} + i$ for all $i = 1, 2, \ldots, 2^{d} - 1$, and assign port numbers so that $2^{d} + i$'s port $4$ leads to $2^{d} + i - 1$, and port $5$ leads to $2^{d} + i - 1$ (for $1 \leq i \leq 2^d - 1$). Finally, for all nodes $v$ at depths $d \leq k - 1$, assign labels $\ln(v)$ and $\rn(v)$ to be consistent with the lateral edges described above. Thus, a balanced tree labeling $\calL$ has been determined at all nodes except the leaves of $G$. Note that $\calL$ is constructed such that all nodes at depth $d \leq k - 2$ are compatible. See Figure~\ref{fig:balanced-tree} for an illustration.

  Let $N = 2^{k-1} (= \Omega(n))$. We complete the labeling $\calL$ to be an embedding of the disjointness function $\disj : \set{0, 1}^N \times \set{0, 1}^N \to \set{0, 1}$, as follows. Given $i \in [N]$, let $v_i$ be the $i\th$ left-most node at depth $k - 1$ in $G$, and let $u_i$ and $w_i$ be its left and right child, respectively. Thus $u_1$ is the left-most leaf in $G$, $w_1$ its right sibling, and so on. For $i \leq N - 1$, assign $\rn(w_i) = v_{i+1}$ and $\ln(v_{i+1}) = w_i$, and take $\ln(v_1) = \rn(w_N) = \bot$. Finally, given any $a, b \in \set{0, 1}^N$, we complete the balanced tree labeling $\calL$ as follows:
  \begin{align*}
      \rn(u_i) &= \ln(w_i) = \bot &&\text{if } a_i = b_i = 1\\
      \rn(u_i) &= w_i,\ \ln(w_i) = u_i &&\text{otherwise.}
  \end{align*}

  For the labeling $\calL$ constructed as above, it is straightforward to verify that all nodes satisfy all conditions of compatibility with one possible exception: $v_i$ fails to satisfy the siblings condition if and only if $a_i = b_i = 1$. That is, $\calL$ is globally compatible if and only if $\disj(a, b) = 1$. Thus, by Lemma~\ref{lem:btl-compatible}, for any solution to $\BTL$ on input $\calL$, the root outputs $(B, \bot)$ if and only if $\disj(a, b) = 1$.

  Fix $v$ to be the root of $G$, and consider an execution of any algorithm $A$ solving $\BTL$ from $v$. We will apply Theorem~\ref{thm:query-lb-from-cc}. The observation that $v$ outputs $(B, \bot)$ if and only if $\disj(a, b) = 1$ shows that our construction of $\calE : (a, b) \mapsto \calL$ and $g(G, \calL) = 1$ if and only if $\calL(v) = (B, \bot)$ gives an embedding of $\disj$ in the sense of Definition~\ref{dfn:embedding}. Moreover, all labels in $\calL$ are independent of $a$ and $b$ except for the leaves, and for each $i$, the labels of $u_i$ and $w_i$ depend only on the values of $a_i$ and $b_i$. Therefore, all queries to $(G, \calL)$ have communication cost $0$, except the queries of the form $\query(v_i, \lc(v_i))$ and $\query(v_i, \rc(v_i))$. The latter queries can be answered by exchanging $a_i$ and $b_i$, hence the communication cost of such queries is $2$. Therefore, by Theorems~\ref{thm:query-lb-from-cc} and~\ref{thm:disj-lb}, the expected query complexity of any algorithm $A$ that solves $\BTL$ with probability bounded away from $1/2$ is $\Omega(N) = \Omega(n)$, as desired.
\end{proof}

\subsection{Proof of Theorem~\ref{thm:btl}}

\begin{proof}[Proof of Theorem~\ref{thm:btl}]
  By Proposition~\ref{prop:btl-ddist-ub} give the upper bounds on $\DDIST$ and $\RDIST$. The corresponding lower bound of $\Omega(\log n)$ follows by analyzing the same construction used in the proof of Proposition~\ref{prop:btl-rvol-lb}. Starting from the root $v$ of $G$, any algorithm that queries nodes only up distance $d \leq k - 1$ can be simulated by Alice and Bob without communication. Thus, such an algorithm cannot solve disjointness (hence $\BTL$) with probability bounded away from $1/2$.

  Finally, the lower bounds of Proposition~\ref{prop:btl-rvol-lb} are tight, as all LCLs trivially have $\RVOL,\allowbreak \DVOL = O(n)$.
\end{proof}

% !TEX root = lcl-volume.tex

\section{\boldmath Hierarchical \texorpdfstring{$2 \frac 1 2$}{2 1/2} Coloring}
\label{sec:hierarchical-coloring}

In this section, we describe a variant of the family of ``hierarchical $2 \frac 1 2$ coloring problems'' introduced by \citet{Chang2019}. Like the original problem, our variant, $\HTHC(k)$, has randomized and deterministic distance complexities $\Theta(n^{1/k})$. We will show that the problem has randomized volume complexity $O(n^{1/k} \log^{O(k)}(n))$, and deterministic volume complexity $\Omega(n / \log n)$.

Like the problems $\LeafColoring$ and $\BTL$, the input labels $\HTHC(k)$ induce a pseudo-forest structure on (a subgraph of) the input graph $G$. In the case of $\HTHC(k)$, the individual pseudo-trees have the following hierarchical structure: At level $1$ of the hierarchy connected components consist of directed paths and cycles. At a level $\ell > 1$, each node is the parent of the ``root'' of a level $\ell - 1$ component. To solve $\HTHC(k)$, each node $v$ produces an output color $\chout(v) \in \set{R, B, D, X}$. Nodes outputting $D$ are said to \dft{decline}, and nodes outputting $X$ are said to be \dft{exempt}. Nodes are only allowed to be exempt under certain locally checkable conditions, described below. Upon removal of exempt nodes, the nodes in each connected component at each level of the hierarchy are required to output $R$, $B$, or $D$ unanimously, with each ``leaf'' outputting its input color or $D$ if $\ell < k$.

\begin{figure}[t]
  \centering
  \includegraphics[scale=0.75]{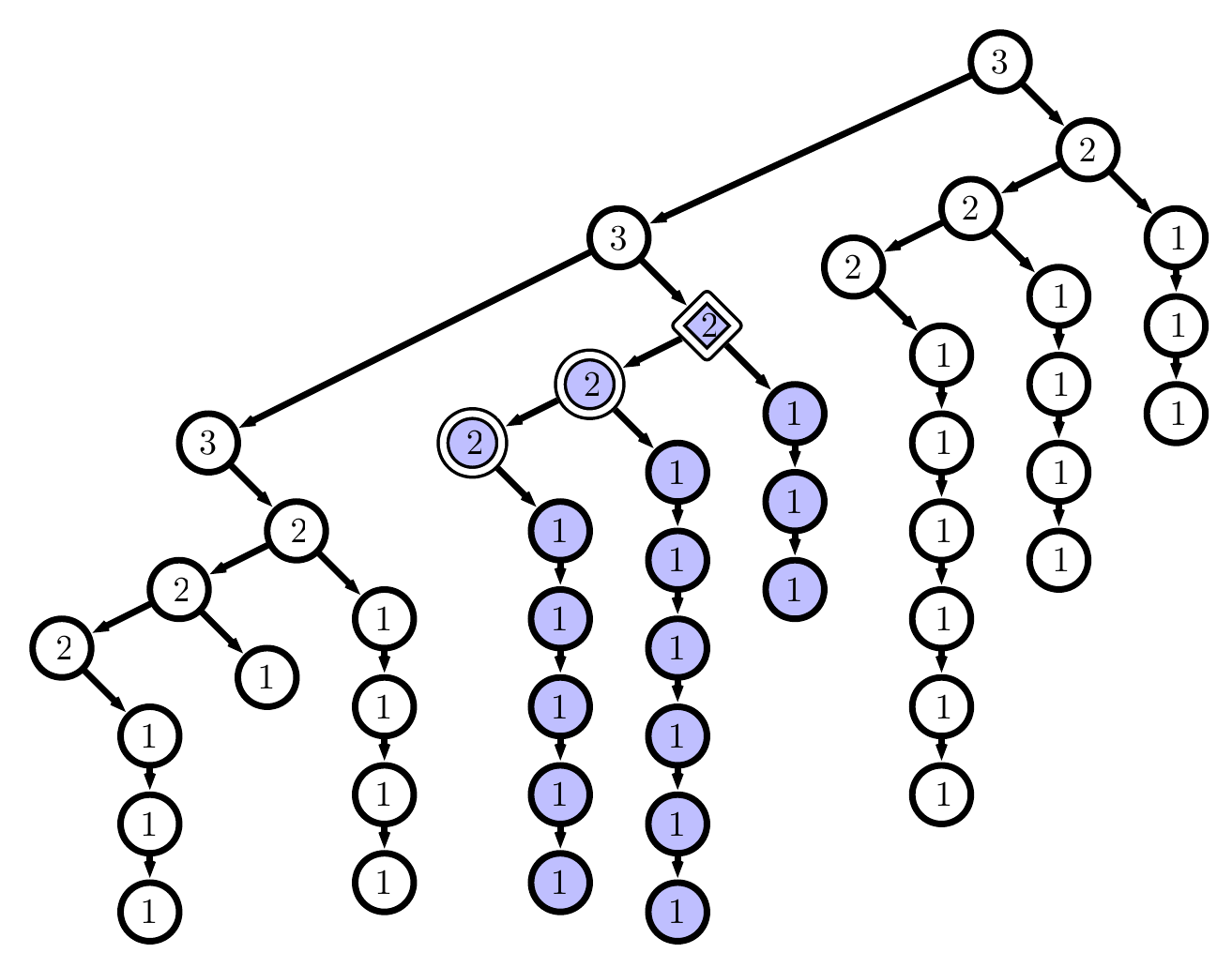}
  \caption{An example of the hierarchical forest induced by a tree labeling. Edges are oriented from parent to child. Left-diagonal and vertical edges indicate that a node is a left child, while right-diagonal edges indicate right children. The number in each node gives the node's level. The shaded region indicates a subtree $H_2$ up to hierarchy level~$2$. The level $2$ root in $H_2$ is indicated by a diamond-shaped node. The double struck node in $H_2$ (at level exactly $2$) comprise a level~$2$ backbone, $C_2$. The lower left node in $C_2$ is a level~$2$ leaf.}
  \label{fig:hf}
\end{figure}

\begin{dfn}
  Let $\calL$ be a (colored) tree labeling. Let $G'$ be the subgraph of $G$ consisting of edges $\set{u, v}$ where $u = \parent(v)$ and $v = \lc(u)$ or $\rc(u)$. The \dft{level} of a node $v$, denoted $\lvl(v)$, is defined inductively as follows: If $\rc(v) = \bot$, then $\lvl(v) = 1$. Otherwise, $\lvl(v) = 1 + \lvl(\rc(v))$. The \dft{hierarchical forest to level $k$}, denoted $G_k = (V_k, E_k)$, is the sub-(pseudo)-forest of $G'$ consisting of edges $\set{u, v}$ with $\lvl(u), \lvl(v) \leq k$ satisfying one of the following properties:
  \begin{itemize}
  \item $v = \parent(u)$, $u = \lc(v)$, and $\lvl(v) = \lvl(u)$, or
  \item $v = \parent(u)$, $u = \rc(v)$, and $\lvl(v) = \lvl(u) + 1$.
  \end{itemize}
\end{dfn}

\begin{obs}
  \label{obs:hf-local}
  The hierarchical forest to level $k$, $G_k$, is locally computable in the sense that each node $v$ can determine $\lvl(v)$, and which of its incident edges are in $G_k$ by examining its $O(k)$-radius neighborhood. Moreover, we assume without loss of generality that every non-$\bot$ label $\parent(v)$, $\lc(v)$, and $\rc(v)$ corresponds to an edge in $G_k$. That is, for example, we have $\set{v, \parent(v)} \in E_k$. 
\end{obs}

\begin{dfn}
  \label{dfn:l-root-leaf}
  Suppose $v$ is a vertex with $\lvl(u) = \ell$. Then we call $v$ a \dft{level $\ell$ root} if $\parent(u) = \bot$ or $u = \rc(\parent(u))$ (and hence $\lvl(\parent(u)) = \ell + 1$). We call $u$ a \dft{level $\ell$ leaf} if $\lc(v) = \bot$.
\end{dfn}

\begin{obs}
  \label{obs:hf-structure}
  Since each node in $G_k$ has at most one parent, $G_k$ is a pseudo-forest. Moreover, it has the following structure: For every $\ell \leq k$, each connected component of $G_k$ consisting of nodes $v$ with $\lvl(v) = \ell$ is a path or cycle, and every (directed) edge is of the form $(v, \lc(v))$. If $\ell = 1$, then for all such $v$ we have $\rc(v) = \bot$. If $1 < \ell \leq k$, then each $\rc(v)$ is the level $\ell - 1$ root of a (directed) subtree of $G_k$. Finally, if $\ell > k$, then $v$ is an isolated vertex in $G_k$.
\end{obs}

\begin{dfn}
  \label{dfn:hthc}
  For any fixed constant $k \in \N$, the problem $\HTHC(k)$ consists of the following:
  \begin{description}
  \item[Input:] a colored tree labeling $\calL$
  \item[Output:] for each $v \in V$, a color $\chout(v) \in \set{R, B, D, X}$
  \item[Validity:] for each $v \in V$, and $\ell = \lvl(v)$
    \begin{enumerate}
    \item if $\ell > k$, then $\chout(v) = X$
    \item if $v$ is a level $\ell$ leaf then $\chout(v) \in \set{\chin(v), D, X}$%; if $v$ is a level $\ell$ root, then $\chout(v) \in \set{R, B, D}$
    \item if $\ell = 1$ then
      \begin{enumerate}
      \item $\chout(v) \in \set{R, B, D}$, and
      \item if $v$ is not a level $1$ leaf, then $\chout(v) = \chout(\lc(v))$
      \end{enumerate}
    \item if $1 < \ell < k$ and $v$ is not a level $\ell$ leaf then either
      \begin{enumerate}
      \item $\chout(v) = \chout(\lc(v)) \in \set{R, B, D}$,
      \item $\chout(v) = X$ and $\chout(\rc(v)) \in \set{R, B, X}$, or
      \item $\chout(v) \in \set{\chin(v), D}$ and $\chout(\lc(v)) = X$
      \end{enumerate}
    \item if $\ell = k$ then $\chout(v) \in \set{R, B, X}$ and 
      \begin{enumerate}
      \item if $\chout(v) = X$ then $\chout(\rc(v)) \in \set{R, B, X}$, and
      \item if $v$ is not a level $\ell$ leaf and $\chout(v) \neq X$, then either
        \begin{itemize}
        \item $\chout(\lc(v)) \neq X$ and $\chout(v) = \chout(\lc(v))$, or
        \item $\chout(\lc(v)) = X$ and $\chout(v) = \chin(v)$
        \end{itemize}
      \end{enumerate}
    \end{enumerate}
  \end{description}
\end{dfn}

The following observation gives some intuition about valid outputs of $\HTHC(k)$.

\begin{obs}
  \label{obs:hthc-structure}
  Consider a valid output for $\HTHC$. Then Conditions~2 and~3 imply that each connected component of level 1 vertices in $G_k$ is unanimously colored either $D$ or $\chin(u)$ where $u$ is the (unique) level 1 leaf in the connected component. Similarly, Conditions~2 and~4 characterize valid colorings of each connected component of $G_k$ at levels $\ell$ satisfying $1 < \ell < k$, although components are no longer required to output unanimous colors. Instead, Condition 4(b) allows nodes $v$ to ``choose'' to output $X$ if $\rc(v)$ outputs a color in $\set{R, B, X}$. However, Conditions~4(a) and~4(c) require that nodes that are not allowed to choose $X$ must either output $\chout(\lc(v))$, or $\chin(v)$ (if $\chout(\lc(v)) = X$). Finally, Condition~5 restricts valid outputs at level $k$. By Condition~5(a) a level $k$ node $v$ is only allowed to output $X$ if $\chout(\rc(v)) \in \set{R, B}$. Meanwhile, Conditions~2 and~5(b) stipulate that on a path between $X$'s at level $k$, all nodes output $\chin(u)$, where $u$ is the parent of the left vertex outputting $X$.
\end{obs}

\begin{figure}[t]
  \centering
  \begin{minipage}{0.45\textwidth}
    \includegraphics[scale=0.75]{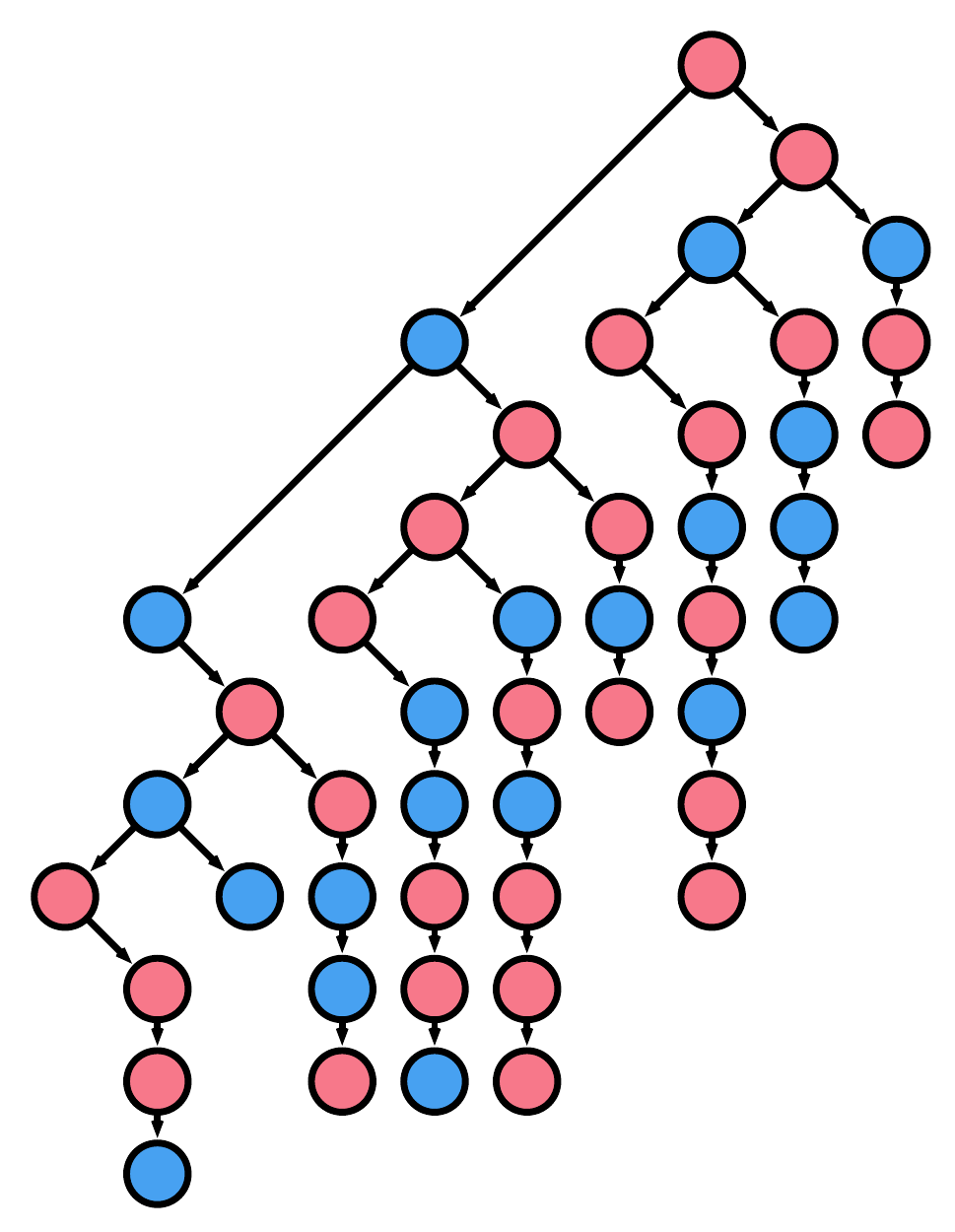}
  \end{minipage}
  \begin{minipage}{0.45\textwidth}
    \includegraphics[scale=0.75]{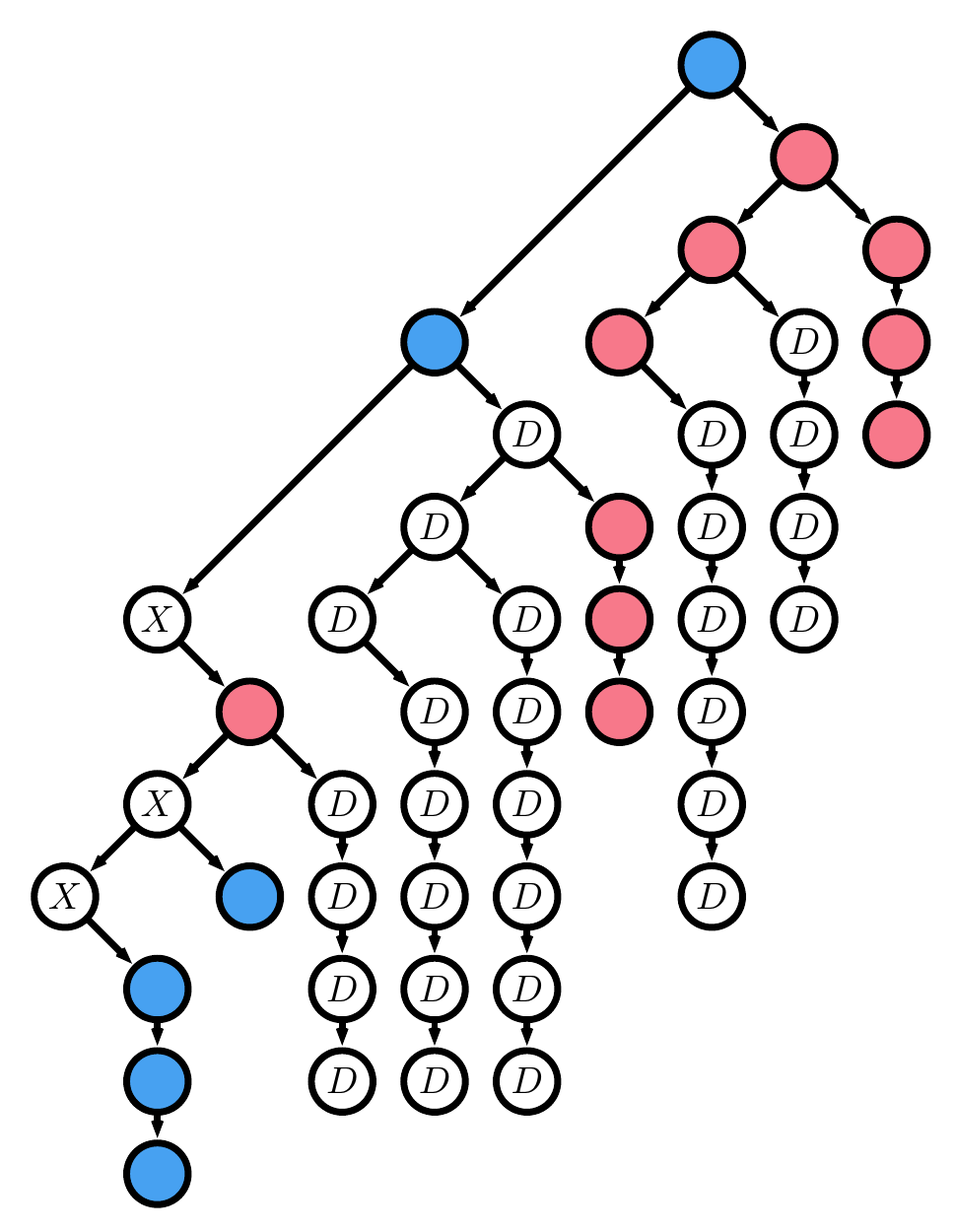}
  \end{minipage}
  \caption{Input (left) and valid output (right) for $\HTHC(3)$.}
  \label{fig:h-thc-input-output}
\end{figure}

\begin{rem}
  \label{rem:hthc-vs-chang}
  Our problem $\HTHC$ differs from the version of hierarchical $2 \frac 1 2$ coloring described by Chang and Pettie~\cite{Chang2019} in two respects. First, we require that connected components of non-exempt vertices are unanimously colored $R$, $B$, or $D$, whereas in~\cite{Chang2019}, such components must either be unanimously colored $D$ or \emph{properly} colored by $R$ in $B$ (i.e., an $R$ node's non-exempt neighbors must output $B$). By using unanimous (rather than proper) colorings in all cases, our version of $\HTHC$ allows us to impose more restrictions on valid outputs by designating input colors of nodes. This is helpful, for example, in the proof of Proposition~\ref{prop:hthc-vol-lb} (the deterministic volume lower bound), where the main claim in the proof relies on unanimous coloring of components. The second difference between $\HTHC$ and that of~\cite{Chang2019} is that in the latter problem, a node $v$ with $\chout(\rc(v)) \neq D$ is \emph{required} to output $X$, whereas our Conditions~4(b) and~5(a) merely allow $v$ to output $X$ if $\rc(v) \neq D$. Our relaxation of the exemption conditions does not affect the distance complexity of the problem, however our modification seems necessary in order for the problem to have small volume complexity.
\end{rem}

The following lemma is clear from previous discussion.

\begin{lem}
  \label{lem:hthc-lcl}
  For every fixed constant $k \in \N$, $\HTHC(k)$ is an LCL.
\end{lem}

We now state the main result of this section.

\begin{lthm}
  \label{thm:hthc}
  For each fixed positive integer $k$, the complexity of $\HTHC(k)$ satisfies
  \begin{equation*}
    \begin{split}
      \RDIST(\HTHC(k)) &= \Theta(n^{1/k}),\\
      \DDIST(\HTHC(k)) &= \Theta(n^{1/k}),\\
      \RVOL(\HTHC(k)) &= O(n^{1/k} \log^{O(k)}(n)),\\
      \DVOL(\HTHC(k)) &= \Omega(n / k \log n).
    \end{split}
  \end{equation*}
\end{lthm}

\subsection{Shallow and Light Components}
\label{sec:hthc-shallow-light}

Before proving the claims of Theorem~\ref{thm:hthc}, we provide some preliminary results on the structure of $G_k$ for any colored tree labeling $\calL$. For the remainder of the section, fix some tree labeling $\calL$, positive integer $k$, and let $G_k$ be the hierarchical forest to level $k$.

\begin{dfn}
  \label{dfn:shallow-light}
  For $\ell \in \N$ with $\ell \leq k$, let $C_\ell$ be a maximal connected component of $G_k$ consisting of nodes $v$ at level $\ell$. We say that $C_\ell$ is \dft{shallow} if $\abs{C_\ell} \leq 2 n^{1/k}$. Otherwise, if $\abs{C_\ell} > 2 n^{1/k}$, we say that $C$ is \dft{deep}.

  Let $H_\ell$ be a connected component of $G_k$ consisting of $C_\ell$ and all of descendants of nodes $v \in C_\ell$ (at all levels $\ell, \ell - 1, \ldots, 1$). We call $H_\ell$ \dft{light} if $\abs{H_\ell} \leq n^{\ell / k}$. Otherwise $H_\ell$ is said to be \dft{heavy}. Similarly, if $v$ is the level $\ell$ root of $H_\ell$, we call $v$ \dft{light} (resp.\ \dft{heavy}) if $H_\ell$ is light (resp.\ heavy).
\end{dfn}

\begin{lem}
  \label{lem:shallow-light}
  Let $C_\ell$ and let $H_\ell$ be as in Definition~\ref{dfn:shallow-light} with $2 \leq \ell \leq k$, and suppose $H_\ell$ is light. Then at most $n^{1/k}$ nodes in $C_\ell$ have heavy right children. 
\end{lem}
\begin{proof}
  Let $W \subseteq C_\ell$ be the nodes with heavy right children, and let $m = \abs{H_\ell}$. By the assumption that $H_\ell$ is light, we have $m \leq n^{\ell / k}$. On the other hand, we have $\abs{W} \cdot n^{(\ell - 1) / k} \leq m$, as each $w \in W$ has a heavy right child at level $\ell - 1$. Combining the two previous inequalities gives $\abs{W} \cdot n^{(\ell - 1) / k} \leq n^{\ell / k}$, which gives $\abs{W} \leq n^{1/k}$, as desired.
\end{proof}

Lemma~\ref{lem:shallow-light} implies the following dichotomy for light components, $H_\ell$: Either $C_\ell$ is shallow, or every subset $U$ of $C_\ell$ of size at least $2 n^{1 / k}$ has the property that at least half of the nodes $v \in U$ have light right children. In the case where $C_\ell$ is shallow, the nodes $v \in C_\ell$ can be validly colored according to Definition~\ref{dfn:hthc} by exploring all of $C_\ell$ using distance and volume $O(n^{1/k})$. Indeed, for any $\ell \in [k]$, it suffices for each $v \in C_\ell$ to output $\chin(u_0)$, where $u_0$ is either the (unique) leaf in $C_\ell$ (in the case $C_\ell$ is a path), or $u_0$ is the vertex with minimal ID (in the case when $C_\ell$ is a cycle).

On the other hand, if $C_\ell$ is deep (and $H_\ell$ is light, hence we must have $\ell \geq 2$), then every node $v \in C_\ell$ has a descendant $u \in C_\ell$ and ancestor $w \in C_\ell$, with $\dist(u, w) \leq n^{1/k}$, such that $u$ is a leaf or $u$ has a light right child, $u'$, and $w$ is a root or $w$ has a light right child, $w'$. In the case where $u' = \rc(u)$ is light, let $H_{\ell-1}$ be the sub-component of $H_\ell$ rooted at $u'$. Then working recursively, we will show that $H_{\ell - 1}$ can be validly colored using \emph{distance} $O(n^{1/k})$ such that $u'$ outputs a color $\chout(u') \in \set{R, B}$. Therefore, $u$ satisfies Condition~4(b) or the implication of~5(a) of validity, so that $\chout(u) = X$ satisfies validity. Similarly, if $w' = \rc(w)$ is light, $w$ can output $X$. Choosing $u$ and $w$ to be the closest descendant and ancestor of $v$ in $C_\ell$ with these properties, $v$ can then output $\chout(v) = \chin(\parent(u))$---as will all other nodes between $u$ and $w$---so that $v$ satisfies Condition~4(a/c) or~5(b). We formalize this procedure in Algorithm~\ref{alg:hthc-dist-ub}. The analysis and matching lower bound appear in Section~\ref{sec:hthc-dist}.

\begin{algorithm}[t]
  \caption{$\RTHC(v, \ell)$}
  \label{alg:hthc-dist-ub}
  \begin{algorithmic}[1]
    \STATE $C \leftarrow$ the level $\ell$ component of $G_k$ containing $v$\label{ln:backbone}
    \IF{$\abs{C} \leq 2 n^{1/k}$}\label{ln:check-shallow}
    \STATE $u_0 \leftarrow$ leaf in $C$ if $C$ is a path and node with minimal ID otherwise
    \RETURN $\chin(u_0)$
    \ELSIF{$\ell = 1$}\label{ln:check-one}
    \RETURN{$D$}
    \ELSIF{ $\RTHC(\rc(v), \ell - 1) \in \set{R, B, X}$}\label{ln:rc-recurse}
    \RETURN $X$\label{ln:out-exempt}
    \ENDIF
    \STATE $u, w \leftarrow v$
    \FOR{$i = 0$ \TO $2 n^{1/k}$}\label{ln:rthc-loop}
    \IF{$\RTHC(\rc(u)) = D$ and $u$ not a level $\ell$ leaf}\label{ln:u-recurse}
    \STATE $u \leftarrow \lc(u)$\hfill\COMMENT{no exempt left descendant found yet}
    \ENDIF
    \IF{$\RTHC(\rc(w)) = D$ and $w$ not a level $\ell$ root}\label{ln:w-recurse}
    \STATE $w \leftarrow \parent(w)$\hfill\COMMENT{no exempt ancestor found yet}
    \ENDIF
    \ENDFOR
    \IF{$u = v$}
    \RETURN $X$\hfill\COMMENT{$\chout(\rc(u)) \neq D$}\label{ln:out-x}
    \ENDIF
    \IF{$\dist(u, w) \leq 2 n^{1/k}$}\label{ln:shallow-cut}
    \IF{$\RTHC(\rc(u)) \in \set{R, B, X}$}\label{ln:leaf-recurse}
    \RETURN $\chin(\parent(u))$\hfill\COMMENT{$\chout(u) = X$}\label{ln:out-parent}
    \ELSE
    \RETURN $\chin(u)$\hfill\COMMENT{$u$ is a leaf and $\chout(u) = \chin(u)$}\label{ln:out-leaf}
    \ENDIF
    \ELSE
    \RETURN $D$\label{ln:out-decline}
    \ENDIF
  \end{algorithmic}
\end{algorithm}

The (deterministic) recursive approach to coloring nodes $v$ in deep components $C_\ell$ gives an $O(n^{1/k})$ \emph{distance} protocol. However, the volume of the protocol may still be large because all nodes between $u$ and $w$ are recursively checked for solvability with $\chout(u') \in \set{R, B}$. In order to solve $\HTHC$ in a volume-efficient manner, our next procedure samples a small fraction of candidates $u$ to try to (validly) color $\chout(u) = X$. By choosing each candidate in $C_\ell$ with probability $p = \Theta((\log n) / n^{1/k})$, the number of such candidates in any $2 n^{1/k}$ radius neighborhood of $v$ is $O(\log n)$. If $H_\ell$ is light, with high probability at least one of the candidates will correctly output $\chout(u) = X$, thus allowing $v$ to output the $\chin(\parent(u))$.\footnote{Note that sampling each candidate $u$ must be done using $u$'s private randomness to ensure that all nodes $v$ visiting $u$ agree on whether or not $u$ is sampled.} Each node $v \in C_\ell$ must visit at most $2 n^{1/k}$ nodes in $C_\ell$, and an inductive argument shows that each recursive call to a sampled $u$ incurs an additional volume of $O(n^{1/k} \log^{O(\ell)} n)$. The argument is formalized in Section~\ref{sec:hthc-vol-ub}.

For the deterministic volume lower bound, our argument essentially shows that if a deterministic algorithm $A$ has the property that many executions of $A$ on input $G$ never query a leaf of $G$, then $A$ cannot solve $\HTHC(k)$ on $G$. We formalize this idea in Section~\ref{sec:hthc-vol-lb}. Given any deterministic algorithm $A$ purporting to solve $\HTHC(k)$ with volume complexity $o(n / (k \log n))$, we describe a procedure $\proc$ that produces a labeled graph $G$ with $n$ vertices for which $A$ produces an incorrect output.

\subsection{Distance Bounds}
\label{sec:hthc-dist}

\begin{prop}
  \label{prop:hthc-dist-ub}
  There exists a deterministic algorithm $A$ such that for every graph $G$ on $n$ nodes, tree labeling $\calL$, and positive integer $k$, $A$ solves $\HTHC(k)$ with distance complexity $O(k n^{1/k})$. 
\end{prop}
\begin{proof}
  Consider the algorithm $A$ where each node $v$ at level $\ell \leq k$ outputs \[\chout(v) = \RTHC(v, \ell),\] and outputs $X$ if $\ell > k$. We first argue that the output of $A$ satisfies the validity conditions of Definition~\ref{dfn:hthc}. Fix a vertex $v$ at level $\ell$. If $\ell > k$, then $\chout(v) = X$ as prescribed by condition $1$. Now consider the case $\ell = 1$, and let $C$ be as in Line~\ref{ln:backbone}. By the provisions in Lines~\ref{ln:check-shallow} and~\ref{ln:check-one}, all $C$ unanimously output $D$ or $\chin(u_0)$. In particular, Conditions~2 and~3 are satisfied.

  Consider $1 \leq \ell < k$, and let $H$ be the connected component of $G_h$ consisting of $C$ together with all descendants of $C$.
  \begin{description}
  \item[Claim.] The output of each $v \in C$ is valid. Moreover, if $H$ is light, then $\chout(v) \in \set{R, B, X}$.
  \item[Proof of claim.] We argue by induction on $\ell$. The base case $\ell = 1$ is handled by the previous paragraph, and the observation that for $\ell = 1$, $H = C$ being light implies $C$ is shallow. Thus the condition of Line~\ref{ln:check-shallow} is satisfied and all nodes in $C$ output $\chin(u_0) \in \set{R, B}$.

    For the inductive step, suppose the claim holds at level $\ell - 1$. If $C$ is shallow, then all nodes $v \in C$ output $\chin(u_0) \in \set{R, B}$, so Conditions~2 and~4(a)/5(b) of validity are satisfied. Further, the level $\ell$ root (if any) outputs a color in $\set{R, B}$, so the conclusion of claim follows when $C$ is shallow. Now suppose $C$ is not shallow. Fix $v \in C$. If $\RTHC(\rc(v), \ell - 1) \in \set{R, B, X}$, then $\chout(v) = X$ and $v$ satisfies validity condition 4(b)/5(a). Otherwise let $u$ and $w$ be $v$'s descendant and ancestor (respectively) in $C$ during the execution of Line~\ref{ln:shallow-cut}.

    If the condition $\dist(u, w) \leq 2 n^{1/k}$ in Line~\ref{ln:shallow-cut} is satisfied, then all nodes $v' \in C$ on the path from $w$ to $u$ (possibly excluding $w$ or $u$ themselves) will store the same values for $w$ and $u$ when Line~\ref{ln:shallow-cut} is executed. Therefore, all such $v'$ will unanimously output $\chin(P(u))$ or $\chin(u)$ (both in $\set{R, B}$) according to Line~\ref{ln:out-parent} or~\ref{ln:out-leaf}. Thus all nodes between $w$ and $u$ in $C$ satisfy validity conditions 2/4(a)/4(c) or 5(b). Moreover, if $v$ is a root level $\ell$ root, then $v = w$ and $\chout(v) = \set{R, B, X}$ as claimed.

    Finally consider the case where $\dist(u, w) > 2 n^{1/k}$. By Lemma~\ref{lem:shallow-light} and the inductive hypothesis, this case can only occur if $H$ is heavy. In particular, this case can only occur for $\ell < k$. If $v$ is a level $\ell$ leaf, then $v = u$ and either $v$ outputs $X$ (at Line~\ref{ln:out-exempt}) or $D$ (at Line~\ref{ln:out-decline}). Either way, validity condition~2 is satisfied. If $v$ is not a leaf, then if $v$ outputs $X$ at Line~\ref{ln:out-exempt}, $v$ satisfies validity condition~4(b). Otherwise, $v$ outputs $D$ in Line~\ref{ln:out-decline}. Similarly, $\lc(v)$ will output $D$ or $X$ (according to if $\RTHC(\rc(\lc(v)), \ell - 1) \in \set{R, B, X}$), because $\lc(v)$ and $v$ agree on whether the condition $\dist(u, w) > 2 n^{1/k}$ in Line~\ref{ln:check-shallow} is satisfied (even though $v$ and $\lc(v)$ may store different values of $u$ and $w$). Accordingly, validity condition 4(a) or 4(c) is satisfied. Thus the claim holds for $v$ at level $\ell$, as desired.
  \end{description}

  By the claim, the output of $A$ satisfies validity (where again, we observe that if $\ell = k$, then $H$ is light. All that remains is to bound the distance complexity of $A$. To this end, a straightforward induction argument (on $\ell$) shows that $\RTHC(v, \ell)$ queries nodes at distance at most $O(\ell \cdot n^{1/k})$ from $v$. For $\ell = 1$, this is immediate, as $v$ can determine if $\abs{C} \leq 2 n^{1/k}$ using $O(n^{1/k})$ queries. For $\ell > 1$, the same applies. Further, $\RTHC(v, \ell)$ queries $O(n^{1/k})$ nodes in $C$ (at level $\ell$) in the loop starting at Line~\ref{ln:rthc-loop}, and each such query makes a single call to $\RTHC(v, \ell - 1)$. Thus, applying the inductive hypothesis, the distance complexity of $\RTHC(v, \ell)$ is $O(n^{1/k}) + O((\ell - 1) n^{1/k})$, which gives the desired result.
\end{proof}

\begin{prop}
  \label{prop:hthc-dist-lb}
  Any (randomized) algorithm $A$ that solves $\HTHC$ with probability bounded away from $1/2$ has distance complexity $\Omega(n^{1/k})$. 
\end{prop}

We omit a proof of Proposition~\ref{prop:hthc-dist-lb}, as the argument is essentially the same as the proof lower bound proof for Chang and Pettie's variant of hierarchical $2 \frac 1 2$ coloring (Theorem~2.3 in~\cite{Chang2019}). We note that the instance achieving the lower bound is a ``balanced'' instance of $\HTHC(k)$, where every ``backbone''---i.e., maximal connected component of $G_k$ consisting of nodes at the same level---has size $\Theta(n^{1/k})$.

\subsection{Randomized Volume Upper Bound}
\label{sec:hthc-vol-ub}

\begin{prop}
  \label{prop:hthc-vol-ub}
  There exists a randomized algorithm $A'$ that for every graph $G$ on $n$ nodes and tree labeling $\calL$, $A'$ solves $\HTHC$ with high probability using volume $O(n^{1/k} \log^{O(k)}(n))$.
\end{prop}

The algorithm $A'$ achieving the conclusion of Proposition~\ref{prop:hthc-vol-ub} is a slight modification of $\RTHC$. In $A'$, only a small fraction of the recursive calls to $\RTHC$ are made. Specifically, each node $v$ uses its private randomness to independently become a \dft{way-point} with probability $p = (c \log n) / n^{1/k}$ (for some constant $c$ to be chosen later). In $A'$, a recursive call to $\RTHC(v', \ell)$ (in Line~\ref{ln:rc-recurse}, \ref{ln:u-recurse}, \ref{ln:w-recurse}, or~\ref{ln:leaf-recurse}) is made if and only if $v'$ is a way-point. In particular, if $v$ is not a way-point, then Line~\ref{ln:rc-recurse} evaluates to false; if $u$ (resp.\ $w$) is not a way-point, then Line~\ref{ln:u-recurse} (resp.\ Line~\ref{ln:w-recurse}) evaluates to true.

The proof of correctness of $A'$ is essentially the same as in the proof of Proposition~\ref{prop:hthc-dist-ub}, at least in cases where the distribution of way-points is sufficiently well-behaved. The potential complication in the analysis arises in deep components $C$ (i.e., $C$ with $\abs{C} \geq 2 n^{1/k}$), because light components can be validly colored deterministically using $O(\abs{C}) = O(n^{1/k})$ volume. For deep components $C$, the choice of $p = (c \log n) / n^{1/k}$ ensures that in any segment of $C$ of length $O(n^{1/k})$ there are $O(\log n)$ way-points in the segment with high probability (Lemma~\ref{lem:waypoint-density}). Thus, we bound the number of recursive calls to $\RTHC$ made by $A'$. On the other hand, if $C$ is contained in a light component $H$, then by Lemma~\ref{lem:shallow-light} any segment of length at least $2 n^{1/k}$ will have at least $1/2$-fraction of its nodes being the parents of light right children. The choice of $p$ allows us to infer that some such ``light parent'' $u$ is a way-point (Lemma~\ref{lem:light-parent}). We then argue by induction that $u$ will output $X$, hence $C$ can be validly colored without any node outputting $D$.

\begin{dfn}
  Fix $\ell$ with $1 < \ell \leq k$ and let $C$ be a maximal connected component of $G_k$ consisting of nodes at level $\ell$. A \dft{short segment} is a path or cycle $S \subseteq C$ such that $\abs{S} \leq 4 n^{1/k}$. We say that $S$ is \dft{crowded} if it contains more than $8 c \log n$ way-points.
\end{dfn}

\begin{lem}
  \label{lem:waypoint-density}
  Suppose each $v \in V$ is chosen to be a way-point independently with probability $p = (c \log n) / n^{1/k}$. Then
  \[
  \Pr(G_k \text{ contains a crowded segment}) \leq O(1/n).
  \]
\end{lem}
\begin{proof}
  First observe that if $G_k$ contains a crowded short segment, then it contains a crowded \emph{maximal} short segment (i.e., a short segment that is not a subset of any other short segment). By associating each maximal short segment with its midpoint (in the case of a path), or node with lowest ID (in the case of a cycle), there are at most $n$ maximal short segments in $G_h$.

  Consider some fixed maximal short segment $S$. For $i \leq 4 n^{1/k}$, let $Y_i$ be an indicator random variable for the event that the $i\th$ node in $S$ is a way-point if $i \leq \abs{S}$, and $Y_i$ is an independent Bernoulli random variable with probability $p$ otherwise. Then
  \begin{align*}
    \Pr(S \text{ is crowded}) &= \Pr\paren{\sum_{i = 1}^{\abs{S}} Y_i \geq 8 c \log n} \leq \Pr\paren{\sum_{i = 1}^m Y_i \geq 8 c \log n}.
  \end{align*}
  Since the $Y_i$ are iid Bernoulli random variables, we can apply the Chernoff bound Lemma~\ref{lem:chernoff} to $Y = \sum_{i = 1}^m Y_i$. Note that $\mu = \E(Y) = 4 c \log n$. Therefore, Lemma~\ref{lem:chernoff} gives
  \[
  \Pr(Y \geq 8 c \log n) = \Pr(Y \geq 2 \mu) \leq e^{-\mu / 3} = e^{-(4 / 3) c \log n} = n^{-4 c / 3}.
  \]
  For any $c \geq 3/2$, this final expression is at most $1/n^2$. Thus taking the union bound over all maximal short segments, we find that
  \begin{align*}
    \Pr(G_h \text{ contains a crowded segment}) &\leq \sum_{S \text{ maximal, short}} \Pr(S \text{ is crowded})\\
    &\leq n \cdot \frac{1}{n^2} = \frac 1 n
  \end{align*}
  which gives the desired result.
\end{proof}

\begin{dfn}
  Let $\ell$, $C$, $H$, and $v$ be as in Lemma~\ref{lem:light-parent}. We call $u \in C$ a \dft{light way-point} if $u$ is a way-point and $\rc(u)$ is light. 
\end{dfn}

\begin{lem}
  \label{lem:light-parent}
  Fix $\ell$ with $1 < \ell \leq k$ and let $C$ be a connected component of $G_k$ consisting of nodes at level $\ell$, and let $H$ be the subgraph of $G_k$ consisting of nodes in $C$ together with all of their descendants. Suppose $C$ is deep (i.e., $\abs{C} > 2 n^{1/k}$ and $H$ is light (i.e., $\abs{H} \leq n^{\ell / k}$. Let $v \in C$, and suppose every node $v' \in C$ is chosen to be a way-point independently with probability $(c \log n) / n^{1/k}$. Then with probability $1 - O(1/n^2)$ there exists a descendant $u \in C$ of $v$ and ancestor $w \in C$ of $v$ such that $\dist(w, u) \leq 2 n^{1/k}$ and
  \begin{enumerate}[noitemsep]
  \item $u$ is either a light waypoint or a level $\ell$ leaf, and
  \item $w$ is either a light waypoint or a level $\ell$ root.
  \end{enumerate}
\end{lem}
\begin{proof}
  We consider the case where $C$ is a path. The case where $C$ is a cycle can be handled similarly. Let $u_0$ be the leaf in $C$ and let $u_0, u_1, \ldots$ be defined by taking $u_{i+1} = \parent(u_i)$. For $j = 0, 1, \ldots, \abs{C} - 2 n^{1/k}$, let $S_j$ be the segment of $C$ containing $u_j, u_{j+1}, \ldots, u_{j + 2 n^{1/k}}$. By Lemma~\ref{lem:shallow-light}, at least $n^{1/k}$ nodes $u_i \in S_j$ have light right children. Therefore
  \[
    \Pr(\text{no } u_i \in S_j \text{ is a light waypoint}) \leq \prod_{i = 1}^{n^{1/k}} (1 - p)
    = (1 - p)^{n^{1/k}}
    \leq e^{- p \cdot n^{1/k}}
    = e^{- c \log n}
    = n^{-c}.
  \]
  Taking any $c \geq 3$ and applying the union bound over all $j$, we have that every $S_j$ contains a light waypoint with probability $1 - O(1/n^2)$. In particular, this implies that with probability at $1 - O(1/n^2)$, the maximum distance between consecutive light way-points is at most $2 n^{1/k}$, which gives the desired result.
\end{proof}

\begin{cor}
  \label{cor:light-parent}
  Suppose every $v \in V$ is chosen to be a waypoint independently with probability $p = (c \log n) / n^{1/k}$. Then with probability $1 - O(1/n)$ the following holds: for every $\ell$ with $1 < \ell \leq k$ and every $v$ such that $v \in C \subseteq H$ where $H$ is light, then $v$ has a descendant $u$ and ancestor $w$ both in $C$ with $\dist(w, u) \leq 2 n^{1/k}$ such that $u$ is either a level $\ell$ leaf or a light waypoint and $w$ is either a level $\ell$ root or a light waypoint.
\end{cor}

\begin{proof}[Proof of Proposition~\ref{prop:hthc-vol-ub}]
  Suppose way-points are chosen in such a way that the conclusions of Lemma~\ref{lem:waypoint-density} and Corollary~\ref{cor:light-parent} hold. Note that such events occur with probability $1 - O(1/n)$. We claim that in this case, the modification $A'$ of $\RTHC$ as described in the paragraph following Proposition~\ref{prop:hthc-vol-ub} gives a valid solution to $\HTHC(k)$ using volume $O(n^{1/k} \log^{O(k)} n)$.

  The idea is to argue inductively that for each level $\ell$ and light $H_\ell$, $H_\ell$ is validly colored with the level $\ell$ root of $H_\ell$ outputting $\chout(v) \neq D$. The base case $\ell = 1$ is straightforward, as $H_1$ being light implies $C_1 = H_1$ is shallow. Therefore, all nodes in $H_1$ output $\chin(u)$ where $u$ is the leaf of $C_1$. For the inductive step suppose the claim is true for $\ell - 1$, and consider $v \in C_\ell$. If $v$ is a light waypoint, then by induction $\chout(\rc(v)) \neq D$, so that $v$ outputs $X$ in Line~\ref{ln:out-x}. Since all light way-points in $C_\ell$ output $X$ and the conclusion of Corollary~\ref{cor:light-parent} guarantees that every connected component of $C_\ell \setminus \set{\text{exempt nodes}}$ is of size at most $2 n^{1/k}$. Thus, each component is unanimously colored by $R$ or $B$ in Line~\ref{ln:out-parent}. This in turn gives a valid coloring where the root of $H_\ell$ does not output $D$, as desired.

  Finally, the upper bound on volume follows from Lemma~\ref{lem:waypoint-density}. By the algorithm description, an execution initiated at $v \in C_\ell$ queries $O(n^{1/k})$ nodes in $C_\ell$. Further, conclusion of the lemma implies that $v$ only recursively calls $\RTHC$ on $O(\log n)$ nodes in $C_\ell$. Thus, a straightforward induction argument shows that the total number of queries is $O(n^{1/k} \log^\ell n)$, which gives the desired result.
\end{proof}

\subsection{Deterministic Volume Lower Bound}
\label{sec:hthc-vol-lb}

\begin{prop}
  \label{prop:hthc-vol-lb}
  Any deterministic algorithm $A$ solving $\HTHC(k)$ requires volume $\Omega(n / k \log n)$.
\end{prop}
\begin{proof}
  Suppose $A$ is any deterministic algorithm with volume complexity at most $m$ purportedly solving $\HTHC(k)$. We assume $k \geq 2$, and describe a process $\proc$ that produces a graph $G$ on $n = O(k \cdot  m \log m)$ vertices and labeling $\calL$ on $G$ such that $\proc$ does not solve $\HTHC(k)$ on $G$. We begin by observing the following claim, whose (omitted) proof is straightforward.
  \begin{description}
  \item[Claim.] Let $A$ be any deterministic algorithm that solves $\HTHC(k)$, and let $H_{v}$ be the subgraph of some input graph $G$ queried by an instance of $A$ initiated at a vertex $v$ in $G$. Suppose $H_v$ has the property that every descendant (in $G_h$) of $v$ in $H_v$ has input color $B$ (resp. $R$). Then the output of $A$ must satisfy $\chout(v) \neq R$ (resp.\ $\chout(v) \neq B$).
  \end{description}
  The process $\proc$ constructs graphs in phases $k, k - 1, \ldots, 1$. In phase $i$, $\proc$ constructs a graph $G_i$ by simulating executions of $A$ initiated at nodes at level $i$. Each phase consists of $O(\log m)$ subphases, and in subphase $\proc$ simulates an execution of $A$ initiated at a single vertex. We illustrate an example iteration of $\proc$ with some algorithm $A$ in Figure~\ref{fig:h-thc-dlb}.

  \begin{figure}[p]
  \centering
  \includegraphics[scale=0.70]{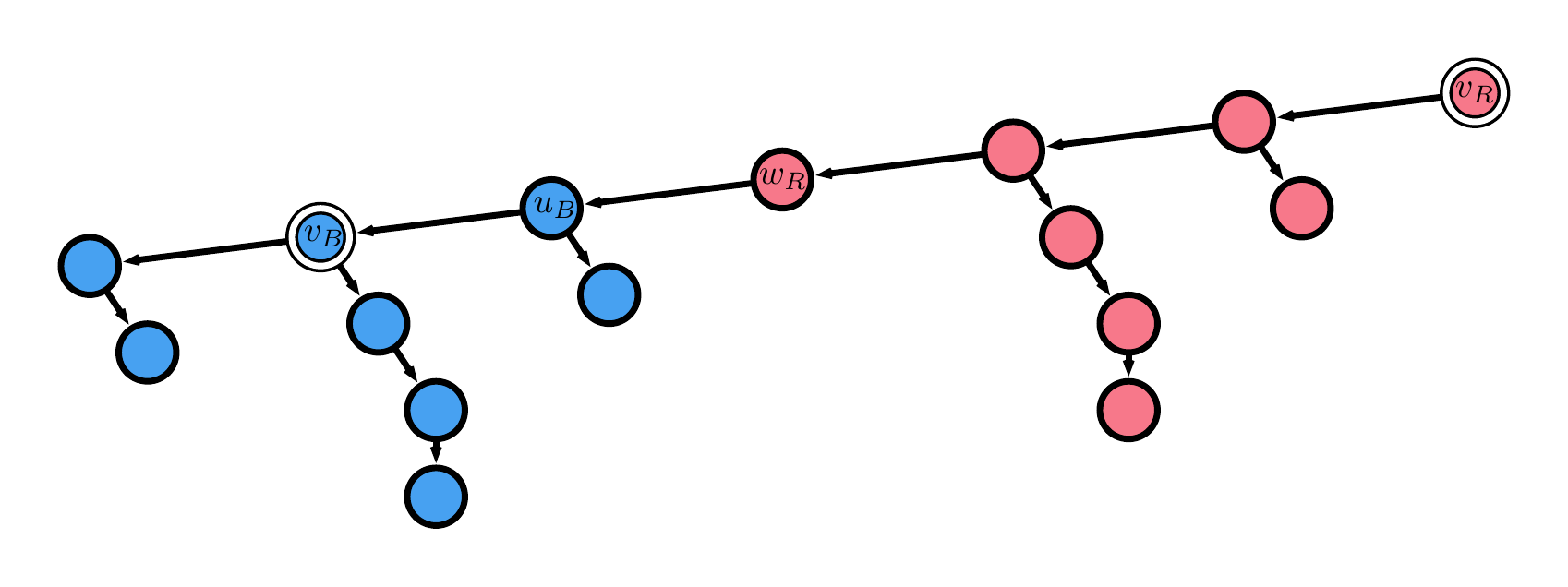}
  \includegraphics[scale=0.70]{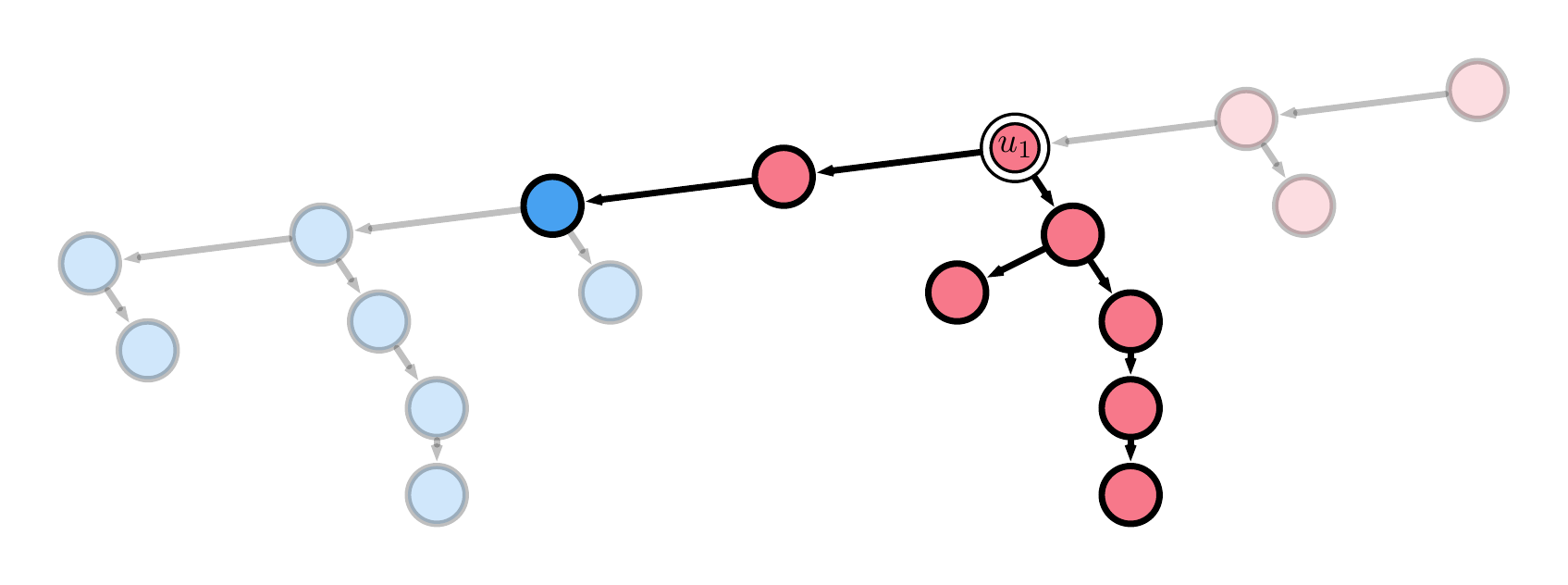}
  \includegraphics[scale=0.70]{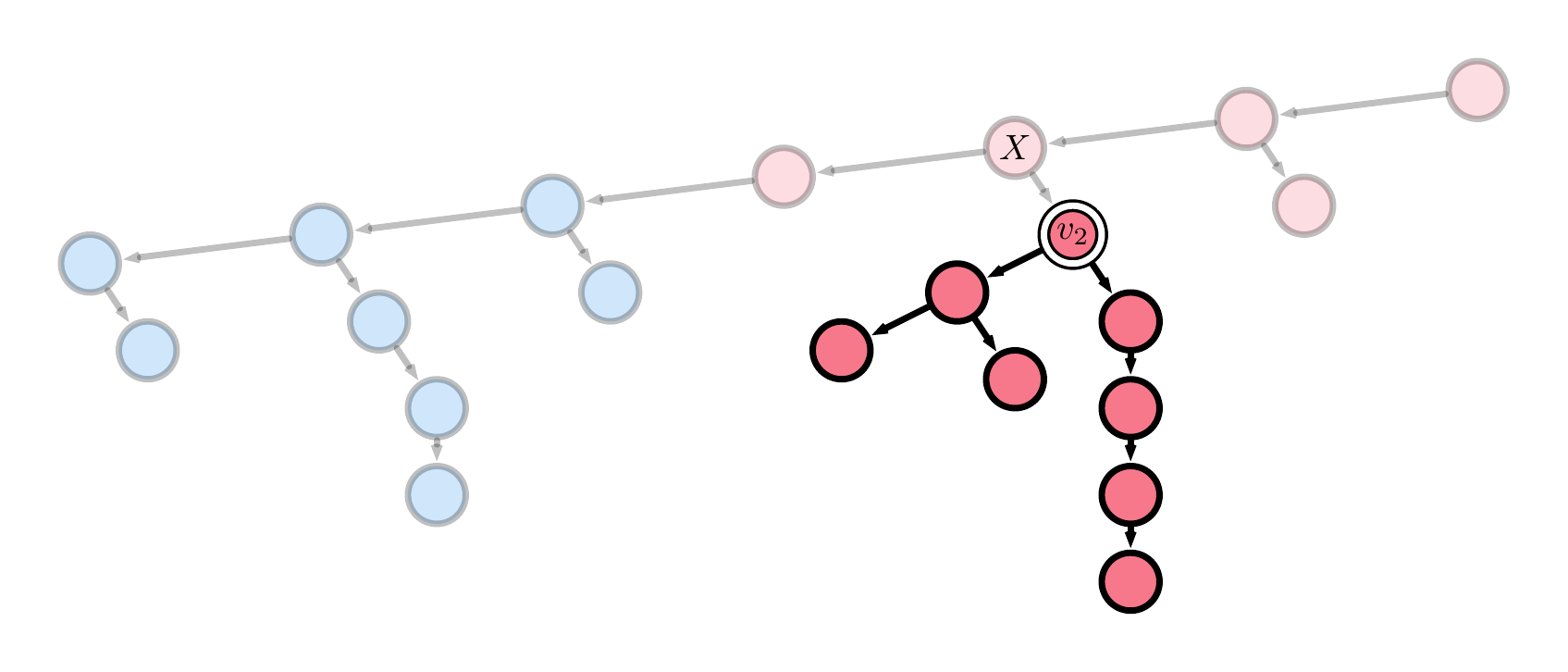}
  \includegraphics[scale=0.70]{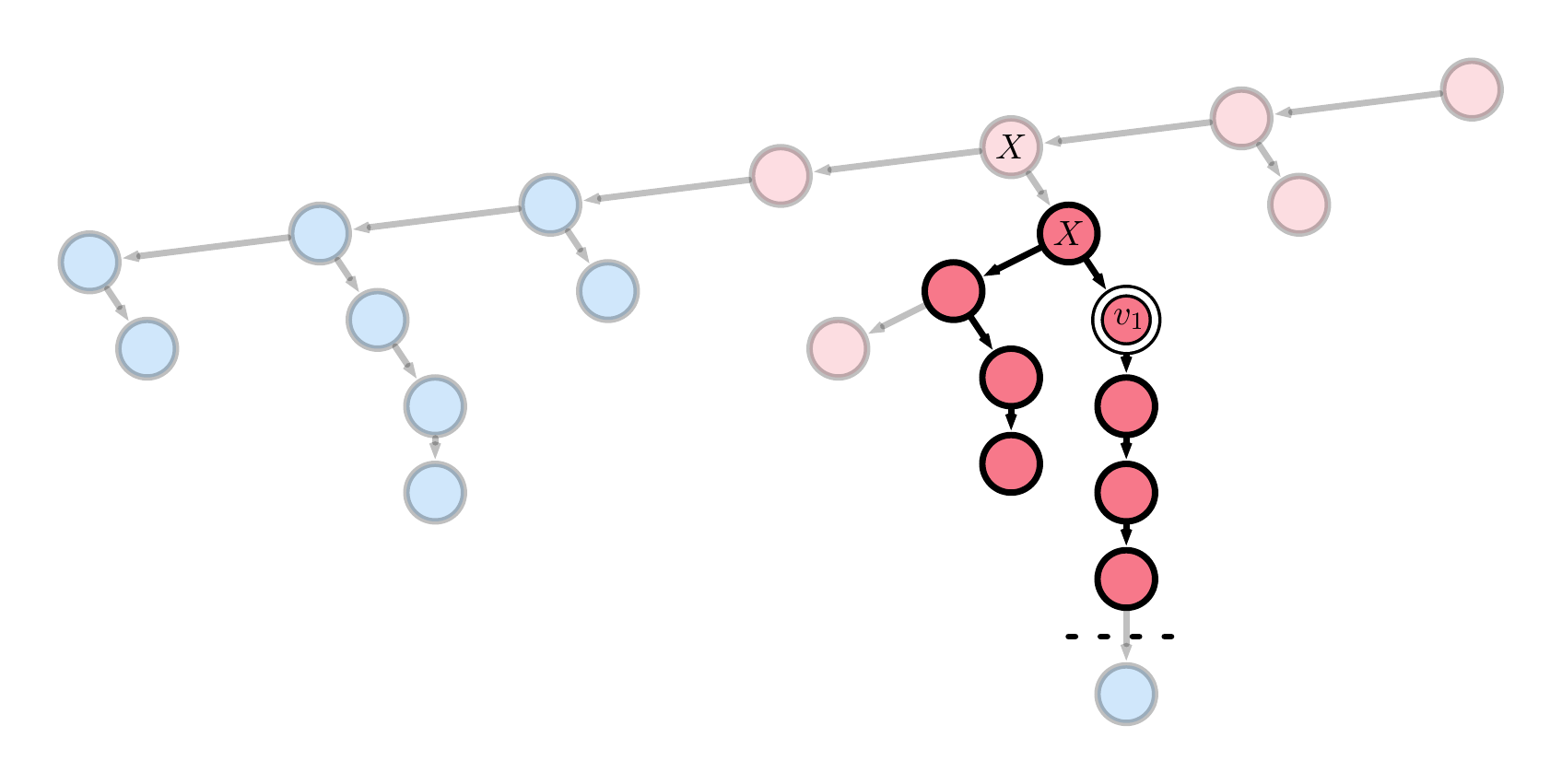}
  \caption{Illustration of the process $\proc$ interacting with a deterministic algorithm $A$ purportedly solving $\HTHC(3)$. The top figure shows the partial instance constructed at the end of the second subphase of phase $k = 3$, assuming the nodes $v_B$ and $v_R$ output $B$ and $R$ respectively. The second image shows the region explored by the node $u_1$. Assuming $u_1$ outputs $X$, the third image gives the result of the first sub-phase of phase~$2$, showing the region explored by $v_2$. Finally, if $v_2$ outputs $X$, the fourth image shows the region explored by $v_1$. Since $v_1$ sees only red vertices in the subtree below it, and $v_1$'s parent, $v_2$, outputs $X$, $v_2$ must output $R$. However, by construction, $v_1$ does not query a leaf. By appending a blue leaf to the region explored by $v_1$, we $\proc$ forces $A$ to produce an incorrect output on the instance.}
  \label{fig:h-thc-dlb}
  \end{figure}

  The procedure $\proc$ maintains the following invariants. Let $G_t$ denote the labeled tree constructed after $\proc$ simulates the $t\th$ query in its simulations.  Every node $v$ in $G_t$ has degree 2 or 3 (with some neighbors possibly not yet assigned), with $\parent(v) = 1$, $\lc(v) = 2$. If $v$ has degree $3$ (i.e., is at level $\ell > 1$) then $\rc(v) = 3$. $G_t$ is a tree with at most $k$ levels. $\proc$ maintains $\lvl(v)$ for each vertex $v$ in $G_t$ which will correspond to $v$'s final level in the completed graph $G$. In particular, if $\lvl(v) = \ell > 1$ then $\lvl(\rc(v)) = \ell - 1$. Finally, $\proc$ assigns IDs to newly added nodes serially so that the $j\th$ node created has ID $j$.

  $\proc$ begins phase~$k$, subphase~1 by simulating $A$ from a vertex $v_B$ with $\lvl(v_B) = k$ and ID $1$, and $\chin(v_B) = B$. During step $t$ (i.e., when $A$ makes its $t\th$ query), if $A$ queries a new node (necessarily a neighbor of some $u$ in $G_{t-1}$) $\proc$ forms $G_t$ by adding a corresponding node $u_t$ to $G_{t-1}$ whose ID and label maintain the invariants described above, and sets $\chin(u_t) = B$. At the end of subphase~1, every node $u$ queried from the execution of $A$ initiated at $v_B$ has $\chin(u) = B$. Thus, by the claim, we must have $\chout(v_B) \neq R$. Moreover, by Condition~5 of validity (Definition~\ref{dfn:hthc}), the output must satisfy $\chout(v_B) \neq D$, so that $\chout(v_B) \in \set{B, X}$. If $\chout(v_B) = X$, $\proc$ ends Phase~$k$, and sets $v_{k-1} = \rc(v_B)$.

  If $\chout(v_B) = B$, $\proc$ continues to subphase~2 as follows. $\proc$ simulates $A$ from a new vertex $v_R \notin G_t$---and not yet connected to $G_t$---exactly as in subphase~1, except that all nodes $u$ created in this subphase have $\chin(u) = R$. When the execution of $A$ initiated at $v_R$ terminates, $G_t$ consists of two connected components: a blue component (containing $v_B$) and a red component (containing $v_R$). As before, $\chout(v_R) \in \set{R, X}$, otherwise $A$ does not solve $\HTHC(k)$ on (some completion of) $G_t$. If $\chout(v_R) = X$, $\proc$ sets $v_{k-1} = \rc(v_R)$ and ends Phase~$k$. Otherwise (if $\chout(v_R) = R$), $\proc$ transforms $G_t$ by connecting the blue and red components as follows. Let $u_B$ be the highest ancestor of $v_B$ in the blue component, and let $w_R$ be the left-most descendant of $v_R$ in the red component. In particular, $\lvl(u_B) = \lvl(w_R) = k$. $\proc$ includes an edge between $u_B$ and $w_R$ and sets $w_R = \parent(u_B)$, $u_B = \lc(w_R)$. Thus, in $G_t$, $v_B$ becomes a left descendant of $v_R$.

  Consider the path $w_0 = v_R, u_1, \ldots, w_b = v_B$ from $v_R$ to $v_B$. By construction, we have $\lvl(w_i) = k$ for all $i$, and $b \leq 2 m$ (since the red and blue components each have size at most $m$). Since $\chout(v_B) = B$ and $\chout(v_R) = R$, in any completion of $G_t$, there must either exist some $w_i$ satisfying $\chout(w_i) = X$, or adjacent nodes $w_{i-1}, w_{i}$ such that $\chout(w_{i-1}) = R$ and $\chout(w_i) = B$. Since the latter case violates validity Condition~5(b), some output of a correct algorithm $A$ should satisfy the former condition. The procedure $\proc$ attempts to find such a vertex $w_i$ by using binary search on the path $(v_R, \ldots, v_B)$. More formally, let $u_1$ be the midpoint of the path from $v_R$ to $v_B$. In subphase~3, $\proc$ simulates an execution of $A$ from $u_1$. As before, every time a new node (not already in $G_t$) is queried, $\proc$ creates a new node $u$ (with the next value of ID) and forms $G_{t+1}$ by appending the node to $G_t$. The input color of $u$ is chosen to be the same as the input color of the node from which $u$ was queried. That is, if $u$ is a child/parent of $u'$, then $\chin(u) = \chin(u')$. Thus, at every step $G_t$ consists of blue and red components connected by a single edge.

  Consider the graph $G_t$ at the end of the execution of $A$ initiated at $u_1$. As before, we must have $\chout(u_1) \in \set{R, B, X}$. (Note that $u_1$ may not satisfy the hypothesis of the claim above, as $u_1$ may query nodes of both input colors). If $\chout(u_1) = X$, $\proc$ ends phase~$k$ and sets $v_{k-1} = \rc(u_1)$. If $\chout(u_1) = R$ (resp.\ $B$), $\proc$ sets $u_2$ to be the midpoint of $u_1$ and $v_B$ (resp.\ $v_R$) and repeats the procedure of the previous paragraph starting from $u_2$. Continuing in this way $\proc$ produces a sequence $u_1, u_2, \ldots, u_r$ of nodes in the path from $v_R$ to $v_B$ until either $\chout(u_r) = X$, or adjacent nodes in the path output $R$ and $B$, respectively, thus contradicting validity condition~5(b). Since $\proc$ uses binary search to determine $u_{i+1}$ from the previous simulations, phase~$k$ requires $O(m \log m)$ queries (i.e., $\abs{G_t} = O(m \log m)$ at the end of phase $k$). 

  If $\proc$ has not found outputs violating validity at the end of phase~$k$, then $\proc$ stores a node $v_{k-1}$ such that $v_{k-1} = \rc(\parent(v_k))$, $\lvl(v_{k-1}) = k-1$, and $\chout(\parent(v_{k-1})) = X$. By validity condition~5(a), $A$ must satisfy $\chout(v_{k-1}) \neq D$. Without loss of generality, assume that $v_{k-1}$ is in the red component of $G_t$ so that $\chin(v_{k-1}) = R$. (If $\chin(v_{k-1}) = B$, interchange the roles of $R$ and $B$ in the following discussion.) $\proc$ begins the first subphase of phase~$k-1$ by simulating an execution of $A$ initiated at $v_{k-1}$, where again the input color of each new node in $G_t$ is consistent with the red or blue component to which the new node is appended. Since this implies that all descendants of $v_{k-1}$ have input color $R$, $v_{k-1}$ should output in $\set{R, X}$. Again, if $\chout(v_{k-1}) = X$, $\proc$ sets $v_{k-2} = \rc(v_{k-1})$ and ends phase~$k-1$. Otherwise, if $\chout(v_{k-1}) = R$, then $\proc$ forms a new vertex $v_{k-1}'$ with input color $B$ with $\lvl(v_{k-1}) = k-1$. $\proc$ then simulates an execution of $A$ from $v_{k-1}'$ in a new blue component disconnected from the component of $G_t$ constructed so far, just as it did for $v_B$. $\calP$ additionally maintains the invariant that all new nodes in this component are at level at most $k - 1$ by setting labels so that $v_{k-1}'$ is the left-most descendant of all its ancestors of $v_{k-1}'$

  Since $v_{k-1}'$ is at level $k-1$, the claim implies that correct output can satisfy $\chout(v_{k-1}') \in \set{B, D, X}$. Again if $\chout(v_{k-1}') = X$, $\proc$ sets $v_{k-2} = \rc(v_{k-1})$ and ends phase~$k-1$. Otherwise, if $\chout(v_{k-1}') \in \set{B, D}$, the component of $G_t$ containing $v_{k-1}'$ is attached to the component containing $v_{k-1}$ so that $v_{k-1}'$ becomes a left descendant of $v_{k-1}$. Since $\chout(v_{k-1}) = R$ and $\chout(v_{k-1}') \in \set{B, D}$, validity conditions~4 now imply that there is some node $v_{k-1}''$ on the path between $v_{k-1}$ and $v_{k-1}'$ such that $\chout(v_{k-1}'') = X$. $\proc$ finds such a node $v_{k-1}''$ using binary search (simulating at most $\log m$ more executions of $A$), and sets $v_{k-2} = \rc(v_{k-1}'')$.

  Continuing in this way, $\proc$ either finds a node violating one of the validity conditions, or it constructs a sequence $v_{k}'', v_{k-1}'', \ldots v_2''$ such that $\lvl(v_i'') = i$ and $\chout(v_i'') = X$. Now consider the graph $G_t$ at the beginning of phase~1, and let $v_1 = \rc(v_2'')$. Assume without loss of generality that $v_1$ lies in a red component, so that $\chin(v_1) = R$, and similarly for all of $v_1$'s descendants of $G_t$ constructed so far. $\proc$ simulates $v_1$, and appends to $G_t$ accordingly, maintaining that all of $v_1$'s descendants input colors are $R$. As before, the claim implies that $\chout(v_1) \neq B$, and validity condition~3(a) implies that $\chout(v_1) \neq X$. Further, validity condition~4(b) (applied to $\parent(v_1) = v_2''$) implies that $\chout(v_1) \neq D$. Thus, we must have $\chout(v_1) = R$. Now $\proc$ appends to $G_t$ a single leaf $v_1'$ that is a left descendant of $v_1$ and setting $\chin(v_1) = B$. By validity conditions~2 and~3(a), $\chout(v_1) \in \set{B, D}$. However, this together with $\chout(v_1) = R$ implies that some node $v_1''$ on the path from $v_1$ to $v_1'$ violates validity condition~3(b).

  When the simulation of $v_1$ is completed, $\proc$ has simulated $O(k \cdot m \log m)$ nodes, hence $\abs{G_t} = O(k \cdot m \log m)$. To complete $G$, $\proc$ appends nodes to each ``unassigned'' port that are consisted with the level of each node containing the unassigned port (with arbitrary input colors). Since there are at most $2 \abs{G_t}$ unassigned ports, and a minimal tree with $\ell$ levels contains $O(\ell)$ nodes, the final graph $G$ satisfies $n = \abs{G} = O(k^2 m \log m)$. Therefore for constant $k$, we must have $m = \Omega(n / \log n)$, which gives the desired result.
\end{proof}

% !TEX root = lcl-volume.tex

\section{\boldmath Hybrid Balanced \texorpdfstring{$2 \frac 1 2$}{2 1/2} Coloring}
\label{sec:hybrid-coloring}

Here, we describe a family of LCLs, $\HBTHC(k)$, that are hybrids of $\BTL(k)$ introduced in Section~\ref{sec:balanced-tree} and $\HTHC(k)$ described in Section~\ref{sec:hierarchical-coloring}. Like $\HTHC(k)$, each node has an associated level $\ell \in [k + 1]$, which in the case of $\HBTHC(k)$ is explicitly given to each node as part of its input label. Each connected component (of $G_T$ induced by a tree labeling) at level $\ell = 1$ corresponds to an instance of $\BTL$, which may either be solved (with all nodes outputting $B$ or $U$) or declined (with all nodes outputting $D$). Nodes at levels $\ell \geq 2$ induce an instance of $\HTHC(k)$, except that validity conditions~4(b) and~5(a) are modified for level $\ell = 2$. Specifically, a level~$2$ node $v$ is allowed to output $X$ if and only if $\rc(v)$ (at level~$1$) outputs a value in $\set{B, U}$---i.e., if the instance of $\BTL$ below $v$ is solved.

The hybrid LCL $\HBTHC(k)$ described above is easier than $\HTHC(k)$ in terms of \emph{distance} complexity because the basic problem $\BTL$ can be solved in each level~$1$ component with distance complexity $O(\log n)$ (Proposition~\ref{prop:btl-ddist-ub}). Since nodes at level $\ell \geq 2$ are distance $\ell - 1$ from (the root of) such a component, every node at any level $\ell \geq 2$ can simply output $X$, knowing that every level~$1$ sub-instance can be solved.

In terms of \emph{volume} complexity, $\HBTHC(k)$ is no easier than $\HTHC(k)$ because there is no volume-efficient method of solving each level~$1$ sub-instance (Proposition~\ref{prop:btl-rvol-lb}). The same algorithmic technique used to solve $\HTHC(k)$ in randomized volume $O(n^{1/k} \log^{O(k)}n)$ can also be applied to solve $\HBTHC(k)$ with the same volume bound. As one would expect, the deterministic volume complexity of $\HBTHC(k)$ is (nearly) linear.

\begin{dfn}
  \label{dfn:hbthc}
  For any fixed constant $k \in \N$, the problem $\HBTHC(k)$ consists of the following:
  \begin{description}
  \item[Input:] a colored balanced tree labeling together with a number $\lvl(v) \in [k+1]$ for each $v \in V$
  \item[Output:] for each $v \in V$ a string encoding either a pair $(\beta(v), p(v)) \in \set{B, U} \times \calP$ or a symbol in $\set{R, B, D, X}$
  \item[Validity:] depending on the input label $\ell = \lvl(v)$, the output must satisfy
    \begin{description}
    \item[$\ell = 1$:] the output is valid in the subgraph of all nodes $w$ satisfying $\lvl(w) = 1$ in the sense of Definition~\ref{dfn:btl}, or $\chout(v) = D$ and all level~$1$ neighbors of $v$ (in $G_T$) also output $D$
    \item[$\ell = 2$:] the output satisfies conditions~2 and~4 of Definition~\ref{dfn:hthc}, except that~4(b) is replaced by ``$\chout(v) = X$ and $\chout(\rc(v)) \in \set{B, U}$''
    \item[$\ell > 2$:] the output is valid in the sense of Definition~\ref{dfn:hthc}
    \end{description}
  \end{description}
\end{dfn}

The following lemma is clear from previous discussion.

\begin{lem}
  \label{lem:hbthc-lcl}
  For every fixed positive integer $k$, $\HBTHC(k)$ is an LCL.
\end{lem}

We state the main theorem of this section below. As noted at the beginning of the section, the proofs of the various claims in Theorem~\ref{thm:hbthc} are analogous to the results appearing in Sections~\ref{sec:balanced-tree} and~\ref{sec:hierarchical-coloring}. Details are left to the reader.

\begin{lthm}
  \label{thm:hbthc}
  For each fixed positive integer $k$, the complexity of $\HBTHC(k)$ satisfies
  \begin{equation*}
    \begin{split}
      \RDIST(\HBTHC(k)) &= \Theta(\log n),\\
      \DDIST(\HBTHC(k)) &= \Theta(\log n),\\
      \RVOL(\HBTHC(k)) &= \widetilde\Theta(n^{1/k}),\\
      \DVOL(\HBTHC(k)) &= \widetilde\Theta(n).
    \end{split}
  \end{equation*}
\end{lthm}

\subsection{More Complexity Classes}
\label{sec:more-classes}

In this final technical section, we describe a family of LCLs---called \emph{hierarchical-or-hybrid $2 \frac 12$ coloring} ($\HHTHC$) with (randomized and deterministic) volume complexity $\Theta(n^{1/\ell})$, randomized volume complexity $\widetilde\Theta(n^{1/k})$, and deterministic volume complexity $\widetilde\Theta(n)$ for any $k, \ell \in \N$ with $k \le \ell$. The idea of $\HHTHC$ is that each node $v$ receives an input label for $\HTHC(\ell)$ or $\HBTHC(k)$ together with a single bit $b_v$. The nodes with input bit $b_v = 0$ should solve $\HTHC(\ell)$, while the nodes with input bit $b_v = 1$ should solve $\HBTHC(k)$.

\begin{dfn}
  \label{dfn:hhthc}
  For any fixed constants $k, \ell \in \N$, $k \leq \ell$, the problem $\HHTHC(k, \ell)$ consists of the following:
  \begin{description}
  \item[Input:] a colored balanced tree labeling together with a number $\lvl(v) \in [k+1]$ and bit $b_v \in \set{0, 1}$ for each $v \in V$
  \item[Output:] for each $v \in V$ a string encoding either a pair $(\beta(v), p(v)) \in \set{B, U} \times \calP$ or a symbol in $\set{R, B, D, X}$
  \item[Validity:] if $G_0$ and $G_1$ are the induced subgraphs of $G$ consists of nodes with inputs $b_v = 0$ and $b_v = 1$ respectively, then
    \begin{itemize}
    \item the output labeling in $G_0$ is a valid output for $\HTHC(\ell)$ (Definition~\ref{dfn:hthc}) (with the input level $\lvl(v)$ ignored)
    \item the output labeling in $G_1$ is a valid output for $\HBTHC(k)$ (Definition~\ref{dfn:hbthc})
    \end{itemize}
  \end{description}
\end{dfn}

Since determining membership in $G_0$ or $G_1$ can be done locally, $\HHTHC(k, \ell)$ is an LCL. To analyze, say, the randomized distance complexity of $\HHTHC(k, \ell)$, observe that each node $v$ can solve $\HHTHC(k, \ell)$ using distance at most the maximum of $O(\RDIST(\HTHC(\ell)))$ and $O(\RDIST(\HBTHC(k)))$. Thus,
\[
\begin{split}
\RDIST(\HHTHC(k, \ell)) &= O(\max \set{\RDIST(\HTHC(\ell)), \RDIST(\HBTHC(k))}) \\ &= O(n^{1/\ell}).
\end{split}
\]
Arguing similarly for the other complexity measures, we obtain the following result.

\begin{lthm}  
  \label{thm:hhthc}
  For all positive integers $k, \ell$ with $k \leq \ell$ the complexity of $\HHTHC(k, \ell)$ satisfies
  \begin{equation*}
    \begin{split}
      \RDIST(\HHTHC(k, \ell)) &= \Theta(n^{1/\ell}),\\
      \DDIST(\HHTHC(k, \ell)) &= \Theta(n^{1/\ell}),\\
      \RVOL(\HHTHC(k, \ell)) &= \widetilde\Theta(n^{1/k}),\\
      \DVOL(\HHTHC(k, \ell)) &= \widetilde\Theta(n).
    \end{split}
  \end{equation*}  
\end{lthm}

% !TEX root = lcl-volume.tex

\section{Discussion and Open Questions}
\label{sec:discussion}
\newcommand{\SO}{\textsf{SO}}

\subsection{Denser and Deterministic Volume Hierarchies}

In the preceding sections, our results establish a \emph{randomized} volume hierarchy for LCL problems: for every $k \in \N$ there exists an LCL whose randomized volume complexity is $\widetilde{\Theta}(n^{1/k})$. Moreover, this hierarchy persists even when restricting to problems whose randomized and deterministic distance complexities are both $\Theta(\log n)$. However, we do not establish an analogous hierarchy for \emph{deterministic} volume complexity. 

\begin{que}
  \label{que:det-volume}
  Does there exist an LCL $\Pi$ whose deterministic volume complexity is in $\omega(\log^* n) \cap o(n)$?
\end{que}

We conjecture a negative answer to Question~\ref{que:det-volume}. The conjecture holds in the relaxed setting where an algorithm is given not $n$ in advance and unique IDs are not required to be polynomial in $n$. The argument is as follows. Let $A$ be any deterministic algorithm, and let $G = (V, E)$ be a (large) instance of an LCL $\Pi$ such that $\VOL(A, G, \calL, v) \ll n$. For fixed $v \in V$, let $H$ be the induced (labeled) subgraph of $G$ consisting of nodes that are either queried by $A$ in an execution initiated from $v$, or neighbors of nodes queried by $A$. An instance of $A$ initiated at $v$ in $H$ will query precisely the same nodes and give the same output as it did in $G$. Since $\Pi$ is an LCL, the solution found by $A$ on $H$ is also a valid solution to $\Pi$. But $\VOL(A, H, \calL', v) = \VOL(A, G, \calL, v) \geq \abs{H} / (\Delta - 1)$. Thus, $A$ queries a constant fraction of $H$. If $\DDIST(\Pi) = \Omega(\log^* n)$, then we can construct an infinite family of graphs $H$ as above such that any deterministic algorithm requires linear volume on $H$, thus implying that $\DVOL(\Pi) = \Omega(n)$.

Recall from Section~\ref{ssec:prelim} that the gap between randomized and deterministic volume is at most exponential, and hence finding any LCL whose \emph{randomized} volume complexity is in $\omega(\log^* n) \cap o(\log n)$ would immediately imply a positive answer to Question~\ref{que:det-volume}. Thus, our conjecture also implies that no such LCLs exist.

We believe that the complexity classes for $\RVOL$ described in this paper---namely problems with complexities roughly $n^{1/k}$ for all $k \in \N$, and $\log n$---are far from exhaustive. 

\begin{que}
  \label{que:dense-classes}
  For what values of $\alpha \in [0, 1]$ is there an LCL problem $\Pi$ with $\RVOL(\Pi) = \widetilde{\Theta}(n^\alpha)$? Are there problems with complexities in $\omega(\log n) \cap \paren{\bigcap_{k \in \N} O(n^{1/k})}$?
\end{que}

We conjecture that problems $\Pi$ with $\RVOL(\Pi) = \widetilde{\Theta}(n^\alpha)$ exist for all $\alpha$ in a dense subset of $[0, 1]$. That is, for every $\alpha' \in [0, 1]$ and $\varepsilon > 0$, there exists an $\alpha \in [0, 1]$ with $\abs{\alpha' - \alpha} < \varepsilon$ and an LCL $\Pi$ with $\RVOL(\Pi) = \widetilde{\Theta}(n^\alpha)$. We believe that LCL constructions similar to those described in~\cite{Balliu2018disc} will yield more $\RVOL$ complexity classes.

\subsection{Graph Shattering}

The LCLs described in preceding sections all have deterministic and randomized distance complexities $\Omega(\log n)$. However, an interesting class of LCLs have randomized and deterministic distance complexities between $\Omega(\log \log n)$ and $O(\log n)$. A canonical example of such a problem is \dft{sinkless orientation} ($\SO$), whose randomized distance complexity is $\Theta(\log \log n)$, and whose deterministic distance complexity is $\Theta(\log n)$.

\begin{que}
  \label{que:sinkless-orientation}
  What are $\DVOL(\SO)$ and $\RVOL(\SO)$?
\end{que}

We note that a negative answer to Question~\ref{que:det-volume} would also settle Question~\ref{que:sinkless-orientation}. Indeed, if $\DVOL(\SO) = \Omega(n)$, then $\RVOL(\SO) = \Omega(\log n)$, as randomness helps at most exponentially~\cite{chang16exponential}, but at the same time we also have $\RVOL(\SO) = O(\log n)$, since we can simulate $O(\log \log n)$-distance algorithms with $O(\log n)$ volume. Conversely, $\RVOL(\SO) = o(\log n)$ would imply $\DVOL(\SO) = o(n)$, thus giving a positive answer to Question~\ref{que:det-volume}.

\subsection{Complexity Classes with Restricted Bandwidth}

The query model we consider, together with the associated volume complexities, can be viewed as a refinement of the LOCAL model in distributed computing, where an algorithm incurs a cost for each query made to a node. Like round complexity in the LOCAL model, volume complexities are purely combinatorial. The CONGEST model~\cite{Peleg2000} is a refinement of the LOCAL model where the complexity measure is communication based: in each round, each node can send at most $B$ (typically $O(\log n)$) bits to each of its neighbors. We believe it is interesting to compare the relative power of the query model and the CONGEST model. We observe that fairly naive bounds on the relative complexities in the query and CONGEST models are the best possible. 

\begin{obs}
  \label{obs:vol-vs-congest}
  Suppose a problem $\Pi$ can be solved in $T$ rounds in the deterministic CONGEST model. Then $\DVOL(\Pi) = \Delta^{O(T)}$. This follows immediately from Lemma~\ref{lem:dist-vs-vol}, as $T$ rounds in the CONGEST model can trivially be simulated in $T$ rounds in the LOCAL model. The same relationship holds for randomized algorithms: if $\Pi$ can be solved in $T$ rounds of CONGEST with randomness, then $\RVOL(\Pi) = \Delta^{O(T)}$. Moreover, these bounds are tight. $\BTL$ can be solved in $O(\log n)$ rounds of CONGEST by each inconsistent/incompatible node $v$ announcing this defect to its neighbors in $O(1)$ rounds. Then after $O(\log n)$ rounds of flooding (each node simply rebroadcasts an ``inconsistent'' message heard from any neighbor), every unbalanced node will witness an inconsistency, thus allowing it to (correctly) output. By the $\Omega(n)$ query lower bound of $\Omega(n)$ posited in Proposition~\ref{prop:btl-rvol-lb}, the bounds $\RVOL(\Pi), \DVOL(\Pi) = \Delta^{O(T)}$ are the best possible even when restricting attention to LCL problems.
\end{obs}

\begin{obs}
  \label{obs:congest-vs-vol}
  Suppose $\Pi$ is a problem with deterministic (resp.\ randomized) volume complexity $D$, such that the input of each node is of size $O(\log n)$ (including random bits in the randomized case). Then $\Pi$ can be solved in $\Delta^{O(D)}$ rounds in the deterministic (resp.\ randomized) CONGEST model as follows. Using a simple flooding procedure, each node can learn its distance $D$ neighborhood in $\Delta^{O(D)}$ rounds of CONGEST. Each node then simulates the query based algorithm on its $D$-neighborhood. Example~\ref{eg:congest-vs-vol} shows that the bound $\Delta^{O(D)}$ is the best possible, though the example is \emph{not} an LCL. 
\end{obs}

\begin{eg}
  \label{eg:congest-vs-vol}
  Let $G = (V, E)$ be a graph on $n = 2 (2^{k+1} - 1)$ nodes constructed as follows. $G$ consists of two balanced binary trees of depth $k$ rooted at nodes $u$ and $v$ respectively, together with an edge between $u$ and $v$. All internal nodes have a (specified) left and right child. Let $u_1, u_2, \ldots, u_{2^k}$ and $v_1, v_2, \ldots, v_{2^k}$ be the leaves below $u$ and $v$, respectively, where $u_1$ is $u$'s left-most descendant, $u_{2^k}$ is $u$'s right-most descendant, and similarly for $v_1, \ldots, v_{2^k}$. Each $v_i$ has an input bit $b_i$. To solve the problem $\Pi$, each $u_i$ must output $b_i$ (initially stored in $v_i$). This problem can easily be solved in $O(\log n)$ queries. However, in the CONGEST model, $\Pi$ requires $\Omega(n / B)$ rounds. Indeed, to solve $\Pi$, the entire vector $b_1 b_2 \cdots b_{2^k}$ must be transmitted across the single edge $\set{u, v}$, which requires $\Omega(2^k / B) = \Omega(n / B)$ rounds.
\end{eg}

While the problem $\Pi$ in Example~\ref{eg:congest-vs-vol} shows that in general the CONGEST round complexity can be exponentially larger than the volume complexity, $\Pi$ is not an LCL. It is not clear to us if this exponential gap is achievable for an LCL problem.

\begin{que}
  What is the largest possible gap between volume complexity and CONGEST round complexity for an LCL? Are there LCLs with (deterministic) volume complexity $D$ that require $\Delta^{\Omega(D)}$ rounds in the CONGEST model?
\end{que}

\subsection{Pushing, Pulling, and MPC}
\label{sec:pushing}

In our volume model (and LCA models), an algorithm interacts with the graph by ``pulling'' information. That is, the output of a node $v$ is determined by a sequence of (adaptively chosen) read operations; there is no direct interaction between executions instantiated from different nodes. One could consider a ``push'' model, in which an algorithm can send messages to other visited nodes in the network. An execution of an algorithm in such a model could proceed in three phases as follows:
\begin{enumerate}
\item Each node adaptively queries its local neighborhood using at most $\Qpull$ queries (as in our volume model).
\item Each node sends messages to at most $\Qpush$ nodes queried in phase~1.
\item Nodes output based on the result of the queries made in phase~1 and messages received in phase~2.
\end{enumerate}
Thus, such a model would allow for some limited interaction between nodes. Equivalently, this model can be viewed as a variant of the volume model where an execution of an algorithm is granted read-write access to (some subset of) nodes it visits, rather than read-only access. One could further generalize this model by allowing multiple iteration of the phases above, or limiting the size of messages sent in phase~2.

It is clear that an algorithm as above can be simulated in $\Qpull$ rounds of LOCAL computation. The pull/push model is trivially at least a strong as our volume model as well, as our volume model is the special case with $\Qpush = 0$. However, it is not clear how much more computational power is afforded by allowing algorithms to push messages. The following example shows that pushing can give an exponential improvement in the complexity of a problem.

\begin{eg}
  \label{eg:btl-push}
  Consider the problem $\BTL$. By Theorem~\ref{thm:btl}, $\RVOL(\BTL) = \Omega(n)$. However, $\BTL$ can be solved in the pull/push model using $O(\log n)$ queries. In phase~1, each node queries its $O(1)$ distance neighborhood to determine if it is compatible in the sense of Definition~\ref{dfn:compatible}. If so, $v$ waits until phase~3. Otherwise, $v$ queries its $\log n$ nearest ancestors (or until a root is encountered). In phase~2, an incompatible node sends its local view (after phase~1) to all of its ancestors. In phase~3, incompatible nodes output $(U, \bot)$. A compatible node receiving an ``incompatible'' message from a descendant outputs in accordance with the validity condition of Definition~\ref{dfn:btl-problem} (choosing the direction of the left-most nearest incompatible descendant). All other nodes output $(B, \parent(v))$. By Lemma~\ref{lem:btl-global}, all unbalanced internal nodes will receive an incompatible message from a descendant, so it is straightforward to verify that this procedure solves $\BTL$ with $\Qpush, \Qpull = O(\log n)$. This upper bound is tight, as $\RDIST(\BTL) = \Omega(\log n)$.
\end{eg}

Example~\ref{eg:btl-push} shows that allowing algorithms read-write access to queried nodes can give an exponential improvement in volume complexity. Moreover, the example gives an LCL whose \emph{deterministic} push-pull query complexity is $\Theta(\log n)$. We are unaware of any LCLs with deterministic volume complexity $\omega(\log^* n)$ and $o(n)$, and indeed we suspect that no such LCLs exist. Thus, it may be the case that randomness plays a much smaller role in the push-pull volume model than it does in the pull-only volume model.

\begin{que}
  \label{que:push-pull}
  What does the volume complexity landscape look like in the push-pull volume model? How is it different from the pull-only volume model?
\end{que}

One reason we believe the push-pull model may be interesting, is that it captures an aspect of interactivity of message passing models (i.e., LOCAL and CONGEST), but restricts the type of information that can be conveyed from one node to another. Thus it may be easier to reason about algorithms in the push-pull model than, say, CONGEST algorithms. In particular, we believe that efficient push-pull algorithms may be valuable in designing volume efficient protocols in the MPC model. For example, the $O(\log n)$ distance algorithm we describe for $\BTL$ can be simulated using $O(n)$ space and $O(\log \log n)$ rounds in the MPC model. However, a sparsified version of the push-pull algorithm described in Example~\ref{eg:btl-push} solves $\BTL$ using $O(n^c)$ space and $O(\log \log n)$ rounds for any positive constant $c$. (The $O(\log \log n)$ run-time is achieved by using graph exponentiation to propagate ``incompatible'' messages upward in the network. To achieve the same effect with $O(n^c)$ space, note that each node needs only to see its leftmost nearest incompatible neighbor. Thus, in each step of the graph exponentiation procedure, it suffices to propagate a single ``incompatible'' message upward from each node, requiring only $O(\log n)$ space per step.)

\begin{que}
  \label{que:mpc-push-pull}
  What algorithms using volume $\VOL$ in the push-pull model can be simulated using $O(\VOL + n^c + \Delta)$ space in the MPC model? Does the analogue of Lemma~\ref{lem:volume-mpc} hold in the push-model?
\end{que}

\subsection{Space and Time Efficient Algorithms in MPC}
\label{sec:conc-mpc}

While Lemma~\ref{lem:volume-mpc} shows that volume-efficient algorithms can be simulated in a \emph{space}-efficient manner, the stated round complexity may be large. For example, consider a problem $\Pi$ with distance and volume complexities $\RDIST(\Pi) = O(\log n)$ and $\RVOL(\Pi) = O(n^c)$. Then an MPC simulation of an $O(\log n)$ round LOCAL algorithm using graph exponentiation gives an $O(n)$-space, $O(\log \log n)$-round algorithm, while Lemma~\ref{lem:volume-mpc} gives $O(n^c)$-space and $O(n^c)$ rounds. Thus the improved volume efficiency may come at the cost of a doubly-exponential increase in run-time! For an arbitrary algorithm using volume $\VOL$, it does not seem likely that we can achieve, say, $O(\VOL)$ space and $O(\log \VOL)$ rounds by using graph exponentiation (or something similar). One reason for this potential inefficiency of simulating volume-efficient algorithms in the MPC model is that in the volume model, queries are adaptive. That is, executions do not know in advance which nodes will be queried. Thus graph exponentiation cannot be used naively without the ability to ``guess'' which nodes will actually be queried in a given execution. Yet, for some cases---such as our $O(\log n)$-volume algorithm for $\LeafColoring$---graph exponentiation can be used: Algorithm~\ref{alg:rw-to-leaf} can easily be simulated in $O(n^c)$ space and $O(\log \log n)$ volume.

\begin{que}
  \label{que:volume-mpc}
  What subclass of algorithms $A$ with (randomized) volume complexity $\VOL$ can be simulated in $O(\VOL + n^c + \Delta)$ space and $O(\log\VOL)$ rounds in the MPC model? What about in the push-pull model?
\end{que}

\subsection{Randomness}
\label{sec:conc-random}

In this final section, we pose some questions related to randomness in the volume model. Throughout the paper, we assume that randomness is provided via a private random string $r_v \colon \N \to \set{0, 1}$ for each vertex $v \in V$. When an algorithm $A$ executed from a node $v$ queries another node $w$, $A$ has access to $r_w$. We call this model the \dft{private randomness} model. In our model, we assumed that an algorithm $A$ accesses random strings sequentially, and that the number of bits accessed by $A$ is bounded with high probability---some assumption of this flavor is needed in the proof of the derandomization result by \citet[Theorem 3]{chang16exponential}. We believe that this assumption is not necessary, at least for LCLs.

\begin{que}
  \label{que:bounded-randomness}
  Suppose $\Pi$ is an LCL and $A$ an algorithm solving $\Pi$ in $R(n)$ rounds in the private randomness model. Is there always an algorithm $A'$ and a function $f \colon \N \to \N$ that solves $\Pi$ in $O(R(n))$ rounds with high probability using $f(n)$ random bits per node? If so, how slowly can $f(n)$ grow?
\end{que}

In addition to the amount of randomness used per node, we think it is interesting to consider other random models. We describe three below, in decreasing power of computation.
\begin{description}
\item[public randomness:] There is a single random string $r \colon \N \to \set{0, 1}$ that is seen by every node.
\item[private randomness:] Each node $v$ has an independent random string $r_v \colon \N \to \set{0, 1}$. When $v$ is queried, $r_v$ is given to the process querying $v$.
\item[secret randomness:] Each node $v$ has an independent random string $r_v$, but $r_v$ is known only to $v$---algorithms querying $v$ do not have access to $r_v$.
\end{description}

It is straightforward to show that any private random protocol can be simulated in the public random model, and that any secret random protocol can be simulated in the private randomness model. However, it is not clear if there is strict separation in the computational power of these models for LCLs.

\begin{que}
  \label{que:public-randomness}
  Are there strict separations between the public, private, and secret randomness models for LCLs?
\end{que}

We suspect that the public and private randomness models are essentially the same for LCL problems. On the other hand, in all of the randomized algorithms described in this paper, randomized coordination (i.e., non-secret randomness) seems essential. Since secret randomness does not allow for coordinated random choices between nodes, it is not clear that secret randomness should give much power over deterministic computation.

To see an example where secret randomness does help, consider the promise version of $\LeafColoring$ (Section~\ref{sec:leaf-coloring}) where all leaves are promised to have the same input color. To solve the promise problem, it is enough for all internal nodes to query a single leaf. Using secret randomness, each internal node can perform a ``downward'' random walk, and all internal nodes will visit some leaf with high probability after $O(\log n)$ queries. (The analysis is identical to the proof of Proposition~\ref{prop:leaf-coloring-rvol-ub}.) Nonetheless, we do not know of any non-promise LCL problem for which there is a gap between secret randomness and deterministic volume complexities.

\section*{Acknowledgments}

We thank Alkida Balliu, Faith Ellen, Mohsen Ghaffari, Juho Hirvonen, Fabian Kuhn, Dennis Olivetti, and Jara Uitto for numerous discussions related to the volume complexity, in particular related to the preliminary observations in Section~\ref{ssec:prelim}.

\urlstyle{same}
\bibliographystyle{plainnat}
\bibliography{lcl-volume}

\end{document}